%% file: main.tex
\documentclass[acmsmall,nonacm,screen]{acmart}
\pdfoutput=1

\input{preamble}

\usepackage{placeins}

\begin{document}

\title{\taypsi: Static Enforcement of Privacy Policies for Policy-Agnostic Oblivious Computation}

\author{Qianchuan Ye}
\orcid{0000-0002-5977-5236}
\affiliation{%
  \institution{Purdue University}
  \city{West Lafayette}
  \country{USA}
}
\email{ye202@purdue.edu}

\author{Benjamin Delaware}
\orcid{0000-0002-1016-6261}
\affiliation{%
  \institution{Purdue University}
  \city{West Lafayette}
  \country{USA}
}
\email{bendy@purdue.edu}

\begin{abstract}
Secure multiparty computation (MPC) techniques enable multiple parties to
compute joint functions over their private data without sharing that data with
other parties, typically by employing powerful cryptographic protocols to
protect individual's data. One challenge when writing such functions is that
most MPC languages force users to intermix programmatic and privacy concerns in
a single application, making it difficult to change or audit a program's
underlying privacy policy. Prior policy-agnostic MPC languages relied on dynamic
enforcement to decouple privacy requirements from program logic.  Unfortunately,
the resulting overhead makes it difficult to scale MPC applications that
manipulate structured data. This work proposes to eliminate this overhead by
instead transforming programs into semantically equivalent versions that
statically enforce user-provided privacy policies. We have implemented this
approach in a new MPC language, called \taypsi; our experimental evaluation
demonstrates that the resulting system features considerable performance
improvements on a variety of MPC applications involving structured data and
complex privacy policies.
\end{abstract}

\maketitle

\input{body}

\bibliography{ref}

\appendix
\include{appendix}

\end{document}

%% file: preamble.tex
\usepackage{amsmath}
\usepackage{amsfonts}
\usepackage{amsthm}
\usepackage{upgreek}
\usepackage{natbib}
\usepackage{mathtools}
\usepackage{stmaryrd}
\usepackage{braket}
\usepackage{mathpartir}
\usepackage{siunitx}
\usepackage{IEEEtrantools}
\usepackage{hyperref}
\usepackage[nameinlink]{cleveref}
\usepackage{xcolor}
\usepackage{xparse}
\usepackage{etoolbox}
\usepackage{listings}
\usepackage{tabularx}
\usepackage{xfrac}
\usepackage{booktabs}
\usepackage{xspace}
\usepackage{tikz-cd}
\usepackage{pgf}
\usepackage{tikz}
\usepackage{wrapfig}
\usepackage{subcaption}
\usepackage{pifont}
\usepackage[utf8]{inputenc}
\usepackage[T1]{fontenc}
\usepackage{microtype}
\usepackage{hhline}
\usepackage{multirow}
\usepackage{enumitem}
\usepackage{ragged2e}
\usepackage{fancyvrb}
\usepackage{algpseudocodex}

\newcommand{\NN}{\ensuremath{\mathbb{N}}}
\newcommand{\ZZ}{\ensuremath{\mathbb{Z}}}

\newcommand{\BB}{\ensuremath{\mathbb{B}}}

\newcolumntype{C}{>{$}c<{$}}
\newcolumntype{L}{>{$}l<{$}}
\newcolumntype{R}{>{$}r<{$}}
\newcolumntype{P}[1]{>{\RaggedRight\arraybackslash}p{#1}}

\newcommand{\production}{\Coloneqq}

\newcommand{\hookuparrow}{\mathrel{\smaller{\rotatebox[origin=c]{90}{$\hookrightarrow$}}}}

\expandafter\patchcmd\csname \string\lstinline\endcsname{%
  \leavevmode
  \bgroup
}{%
  \leavevmode
  \ifmmode\hbox\fi
  \bgroup
}{}{%
  \typeout{Patching of \string\lstinline\space failed!}%
}
\newcommand{\li}{\let\par\endgraf\lstinline}

\newcommand*\ruleline[1]{\par\noindent\raisebox{.8ex}{\makebox[\linewidth]{\hrulefill\hspace{1ex}\raisebox{-.8ex}{#1}\hspace{1ex}\hrulefill}}}

\newcommand{\jbox}[1]{\begin{flushright}\fbox{#1}\end{flushright}}

\newbool{nofancystyle}
\newcommand{\fancystyle}[1]{\ifbool{nofancystyle}{}{#1}}

\definecolor{kwcolor}{HTML}{6F42C1}
\definecolor{dfcolor}{HTML}{D73A49}
\definecolor{typecolor}{HTML}{22863A}
\definecolor{metacolor}{HTML}{005CC5}
\definecolor{cmcolor}{HTML}{6A737D}

\newcommand{\kwstyle}{\fancystyle{\color{kwcolor}}}
\newcommand{\dfstyle}{\fancystyle{\color{dfcolor}}}
\newcommand{\typestyle}{\fancystyle{\color{typecolor}}}
\newcommand{\metastyle}{\fancystyle{\color{metacolor}}}
\newcommand{\cmstyle}{\booltrue{nofancystyle}\color{cmcolor}}

\newcommand{\symbolstyle}{\fancystyle{\color{darkgray}}}

\newcommand{\oblivstyle}[1]{$\widehat{\mbox{#1}}$}

\newcommand{\outputstyle}[1]{$\dot{\mbox{#1}}$}

\fvset{fontsize=\footnotesize}
\definecolor{lightergray}{gray}{0.9}

\newcommand{\loadt}{\ensuremath{\lambda_{\mathtt{OADT}}}\xspace}

\newcommand{\loadtpsi}{\ensuremath{\lambda_{\mathtt{OADT}\Uppsi}}\xspace}
\newcommand{\taype}{\textsc{Taype}\xspace}
\newcommand{\taypsi}{\textsc{Taypsi}\xspace}
\newcommand{\psistructs}{$\Uppsi$-structures\xspace}
\newcommand{\psistruct}{$\Uppsi$-structure\xspace}
\newcommand{\oil}{\textsc{Oil}\xspace}
\newcommand{\tpsi}{\ensuremath{\Uppsi}}
\newcommand{\graylight}[1]%
{{\setlength{\fboxsep}{0pt}\colorbox{lightergray}{\strut #1}}}

\NewDocumentCommand{\fmath}{m}{\TextOrMath{{\footnotesize\ensuremath{#1}}}{#1}}

\newcommand{\empctx}{\fmath{\cdot}}
\newcommand{\gctx}{\fmath{\Sigma}}
\newcommand{\tctx}{\fmath{\Gamma}}
\newcommand{\sctx}{\fmath{\mathcal{S}}}
\newcommand{\lctx}{\fmath{\mathcal{L}}}
\newcommand{\hasview}{\fmath{\preccurlyeq}}
\newcommand{\ple}{\fmath{\sqsubseteq}}
\newcommand{\sctxoadt}{\fmath{\sctx_{\upomega}}}
\newcommand{\sctxjoin}{\fmath{\sctx_{\sqcup}}}
\newcommand{\sctxintro}{\fmath{\sctx_{I}}}
\newcommand{\sctxelim}{\fmath{\sctx_{E}}}
\newcommand{\sctxcoer}{\fmath{\sctx_{\uparrow}}}
\newcommand{\cstrs}{\fmath{\mathcal{C}}}
\newcommand{\cstr}{\fmath{c}}
\newcommand{\fctx}{\fmath{\mathcal{F}}}
\newcommand{\mctx}{\fmath{\mathcal{M}}}

\NewDocumentCommand{\eqcls}{m}{\fmath{[#1]}}

\newbool{ctxnotfirst}
\NewDocumentCommand{\typectxaux}{m}%
{\ifbool{ctxnotfirst}{;}{\booltrue{ctxnotfirst}}#1}
\NewDocumentCommand{\typectx}{>{\SplitList{;}} m}%
{{\ProcessList{#1}{\typectxaux}}}

\NewDocumentCommand{\steparrowaux}{m m}%
{\longrightarrow\IfBooleanT{#1}{^{\ast}}\IfValueT{#2}{^{#2}}}
\NewDocumentCommand{\steparrow}{s o}{\steparrowaux{#1}{#2}}
\NewDocumentCommand{\stepbodyaux}{m m m m}%
{#1 \steparrowaux{#3}{#4} #2}
\NewDocumentCommand{\stepbody}{>{\SplitArgument{1}{~>}} m m m}%
{\stepbodyaux #1 {#2}{#3}}
\NewDocumentCommand{\stepaux}{m m m m}%
{\IfValueTF{#2}{#1 \vdash \stepbody{#2}{#3}{#4}}{\stepbody{#1}{#3}{#4}}}
\NewDocumentCommand{\step}{s o >{\SplitArgument{1}{|-}} m}%
{\fmath{\stepaux #3 {#1}{#2}}}

\NewDocumentCommand{\ovaltyaux}{m m}{#1 \Leftarrow #2}
\NewDocumentCommand{\ovalty}{>{\SplitArgument{1}{<~}} m}%
{\fmath{\ovaltyaux #1}}

\NewDocumentCommand{\indistaux}{m m}{#1 \approx #2}
\NewDocumentCommand{\indist}{>{\SplitArgument{1}{~}} m}%
{\fmath{\indistaux #1}}

\NewDocumentCommand{\haskindaux}{m m}{#1 \dblcolon #2}
\NewDocumentCommand{\haskindinternal}{>{\SplitArgument{1}{::}} m}%
{\haskindaux #1}
\NewDocumentCommand{\haskind}{m}{\fmath{\haskindinternal{#1}}}

\NewDocumentCommand{\typerelaux}{m m}{#1\IfValueT{#2}{ \sim #2}}
\NewDocumentCommand{\typerel}{>{\SplitArgument{1}{~}} m}{\fmath{\typerelaux #1}}

\NewDocumentCommand{\hastypebody}{m m}{#1 : \typerel{#2}}
\NewDocumentCommand{\hastypeaux}{>{\SplitArgument{1}{::}} m m}%
{\hastypebody #1\IfValueT{#2}{ \rhd #2}}
\NewDocumentCommand{\hastypeinternal}{>{\SplitArgument{1}{|>}} m}%
{\hastypeaux #1}
\NewDocumentCommand{\hastype}{m}{\fmath{\hastypeinternal{#1}}}

\NewDocumentCommand{\kindaux}{m m}%
{\IfValueTF{#2}%
  {\typectx{#1} \vdash \haskindinternal{#2}}%
  {\tctx \vdash \haskindinternal{#1}}}
\NewDocumentCommand{\kind}{>{\SplitArgument{1}{|-}} m}%
{\fmath{\kindaux #1}}

\NewDocumentCommand{\typeaux}{m m}%
{\IfValueTF{#2}%
  {\typectx{#1} \vdash \hastypeinternal{#2}}%
  {\tctx \vdash \hastypeinternal{#1}}}
\NewDocumentCommand{\typeinternal}{>{\SplitArgument{1}{|-}} m}%
{\typeaux #1}
\NewDocumentCommand{\type}{m}{\fmath{\typeinternal{#1}}}

\NewDocumentCommand{\typequivbodyaux}{m m}{#1 \equiv #2}
\NewDocumentCommand{\typequivbody}{>{\SplitArgument{1}{==}} m}{\typequivbodyaux #1}
\NewDocumentCommand{\typequivaux}{m m}%
{\IfValueTF{#2}%
  {#1 \vdash \typequivbody{#2}}%
  {\typequivbody{#1}}}
\NewDocumentCommand{\typequiv}{>{\SplitArgument{1}{|-}} m}%
{\fmath{\typequivaux #1}}

\NewDocumentCommand{\paredbodyaux}{m m}{#1 \Rrightarrow #2}
\NewDocumentCommand{\paredbody}{>{\SplitArgument{1}{~>}} m}{\paredbodyaux #1}
\NewDocumentCommand{\paredaux}{m m}%
{\IfValueTF{#2}%
  {#1 \vdash \paredbody{#2}}%
  {\paredbody{#1}}}
\NewDocumentCommand{\pared}{>{\SplitArgument{1}{|-}} m}%
{\fmath{\paredaux #1}}

\NewDocumentCommand{\extctxaux}{m}{\hastypeinternal{#1},}
\NewDocumentCommand{\extctx}{o >{\SplitList{;}} m}%
{\fmath{\ProcessList{#2}{\extctxaux}\IfValueTF{#1}{#1}{\tctx}}}

\NewDocumentCommand{\dtypeaux}{m m}%
{\IfValueTF{#2}{#1 \vdash #2}{\gctx \vdash #1}}
\NewDocumentCommand{\dtype}{>{\SplitArgument{1}{|-}} m}%
{\fmath{\dtypeaux #1}}

\NewDocumentCommand{\interpE}{O{n} m}{\fmath{\mathcal{E}_{#1}\llbracket#2\rrbracket}}
\NewDocumentCommand{\interpV}{O{n} m}{\fmath{\mathcal{V}_{#1}\llbracket#2\rrbracket}}
\NewDocumentCommand{\interpG}{O{n} m}{\fmath{\mathcal{G}_{#1}\llbracket#2\rrbracket}}
\newcommand{\subst}{\fmath{\sigma}}

\NewDocumentCommand{\mergeableaux}{m m}%
{{\curlywedgeuparrow}#1\IfValueT{#2}{\rhd #2}}
\NewDocumentCommand{\mergeable}{>{\SplitArgument{1}{|>}} m}%
{\fmath{\mergeableaux #1}}

\NewDocumentCommand{\coerciblebody}{m m}{#1 \rightarrowtail #2}
\NewDocumentCommand{\coercibleaux}{>{\SplitArgument{1}{~>}} m m}%
{\coerciblebody #1\IfValueT{#2}{\rhd #2}}
\NewDocumentCommand{\coercible}{>{\SplitArgument{1}{|>}} m}%
{\fmath{\coercibleaux #1}}

\newcommand{\liftR}{\type}

\NewDocumentCommand{\ppxaux}{m m}{#1 \rhd #2}
\NewDocumentCommand{\ppx}{>{\SplitArgument{1}{|>}} m}%
{\fmath{\ppxaux #1}}

\NewDocumentCommand{\liftAaux}{m m}{\typeinternal{#1} \mid #2}
\NewDocumentCommand{\liftA}{>{\SplitArgument{1}{~>}} m}%
{\fmath{\liftAaux #1}}

\NewDocumentCommand{\inclsaux}{m m}{#1 \in \eqcls{#2}}
\NewDocumentCommand{\incls}{>{\SplitArgument{1}{~}} m}%
{\fmath{\inclsaux #1}}

\newbool{freshnotfirst}
\NewDocumentCommand{\freshaux}{m}%
{\ifbool{freshnotfirst}{,}{\booltrue{freshnotfirst}}#1}
\NewDocumentCommand{\fresh}{>{\SplitList{;}} m}%
{\fmath{\ProcessList{#1}{\freshaux}\text{ fresh}}}

\NewDocumentCommand{\cstrsataux}{m m}%
{\IfValueTF{#2}%
  {\typectx{#1} \vDash #2}%
  {\subst \vDash #1}}
\NewDocumentCommand{\cstrsat}{>{\SplitArgument{1}{|-}} m}%
{\fmath{\cstrsataux #1}}

\NewDocumentCommand{\lctxvalid}{o m}{\fmath{\vDash\IfValueT{#1}{_#1}#2}}
\NewDocumentCommand{\lctxderiv}{m}{\fmath{\vdash#1}}
\NewDocumentCommand{\erase}{m}{\fmath{\lfloor#1\rfloor}}

\usetikzlibrary{arrows}
\usetikzlibrary{positioning}
\usetikzlibrary{tikzmark}
\usetikzlibrary{calc}


%% file: body.tex
\section{Introduction}
\label{sec:intro}

\emph{Secure multiparty computation} (MPC) techniques allow multiple
parties to jointly compute a function over their private data while
keeping that data secure. A variety of privacy-focused applications
can be formulated as MPC problems, including secure auctions, voting,
and privacy-preserving machine learning~\cite{evans2018, hastings2019,
  laud2015}. MPC solutions typically depend on powerful cryptographic
techniques, e.g., Yao's Garbled Circuits~\cite{yao1982} or secret
sharing~\cite{beimel2011}, to provide strong privacy guarantees. These
cryptographic techniques can be difficult for non-experts to use,
leading to the creation of several high-level languages that help
programmers write MPC applications~\cite{malkhi2004, zhang2013,
  rastogi2014, liu2015, zahur2015, rastogi2019, hastings2019,
  darais2020, acay2021, sweet2023, ye2022}. While raising the level of
abstraction, almost all of these languages intermix privacy and
programmatic concerns, requiring the programmers to explicitly enforce
the high-level privacy policies within the logic of the application
itself, using the secure operations provided by the language. As a
consequence, the entire application must be examined in order to audit
its privacy policy, and an application must be rewritten in order to
change its privacy guarantees. This intermixing of policy enforcement
and application logic thus makes it difficult to read, write, and
reason about MPC applications.

This is particularly true for applications with the sorts of complex
requirements that can occur in practice. Within the United States, for
example, the Health Insurance Portability and Accountability Act
(HIPAA) governs how patient data may be used. HIPAA allows
\textit{either} the personally identifiable information (PII)
\textit{or} medical data to be shared, but not both. Notably, this
policy does not simply specify whether some particular field of a
patient's medical record is private or public; rather it is a
\emph{relation} that dictates how a program can access and manipulate
different parts of every individual record. To conform to this policy,
an MPC application must either pay the (considerable) cryptographic
overhead of conservatively securing all accesses to the fields of a
record, or adopt a more sophisticated strategy for monitoring how data
is accessed. These challenges become more acute when dealing with
structured data, e.g., lists or trees, whose policies are necessarily
more complex. Consider a classifier that takes as input a decision
tree and a medical record, each of which is owned by a different
party: if the owner of the tree stipulates that its depth may be
disclosed, the classification function must use secure operations to
ensure that no other information about the tree is leaked, e.g., its
spine or the attributes it uses.  If the owner of this tree is willing
to share such information, however, this function must either be
rewritten to take advantage of the new, more permissive policy, or
continue to pay the cost of providing stricter privacy guarantees.
Thus, most existing MPC languages require users to write different
implementations of essentially the same program for each distinct
privacy policy.

A notable exception is \taype~\cite{ye2023}, a recently proposed language that
decouples privacy policies from programmatic concerns, allowing users to write
applications over structured data that are agnostic to any particular privacy
policy. To do so, \taype implements a novel form of the \emph{tape semantics}
proposed by~\citet{ye2022}. This semantics allows insecure operations
whose evaluation \emph{could} violate a policy to appear in a program, as long
as the results of these operations are \emph{eventually } protected. Under tape
semantics, such operations are lazily deferred until it is safe to execute them,
effectively \emph{dynamically} ``repairing'' potential leaks at runtime. Using
\taype, programmers can thus build a privacy-preserving version of a standard
functional program by composing it with a policy, specified as a \emph{dependent
type} equipped with security labels, relying on tape semantics to enforce the
policy during execution. Unfortunately, while this enforcement strategy
 disentangles privacy concerns from program logic, it also introduces
considerable overhead for applications that construct or manipulate structured
data with complex privacy requirements. Thus, this strategy does not scale to
the sorts of complex applications that could greatly benefit from this
separation of concerns.

This work presents \taypsi, a policy-agnostic language for writing MPC
applications that eliminates this overhead by instead transforming a
non-secure function into a version that \emph{statically} enforces a
user-provided privacy policy. \taypsi extends \taype with a form of
dependent sums, which we call \tpsi-types, that package together the
public and private components of an algebraic data type (ADT).
Each \tpsi-type is equipped with a set of \emph{\psistructs} which
play an important role in our translation, enabling it to, e.g.,
efficiently combine subcomputations that produce ADTs with different
privacy policies. Our experimental evaluation demonstrates that this
strategy yields considerable performance improvements over the
enforcement strategy used by \taype, yielding exponential improvements
on the most complex benchmarks in our evaluation suite.

To summarize, the contributions of this paper are as follows:
\begin{itemize}
\item We present \taypsi, a version of \taype extended with \tpsi-types, a form
  of dependent sums that enables modular translation of non-secure programs into
  efficient, secure versions. This language is equipped with a security type
  system that offers the same guarantees as \taype: after jointly computing a
  well-typed function, neither party can learn more about the other's private
  data than what can be gleaned from their own data and the output of the
  function.
\item We develop an algorithm that combines a program written in the non-secure
  fragment of \taypsi with a privacy policy to produce a secure private version
  that statically enforces the desired policy. We prove that this algorithm
  generates well-typed (and hence secure) target programs that are additionally
  guaranteed to preserve the semantics of the source programs.
\item We evaluate our approach on a range of case studies and microbenchmarks.
  Our experimental results demonstrate exponential performance improvements over
  the previous state-of-the-art (\taype) on several complicated benchmarks,
  while simultaneously showing no performance regression on the remaining
  benchmarks.
\end{itemize}

\section{Overview}
\label{sec:overview}

\begin{wrapfigure}{r}{.42\textwidth}
\vspace{-.55cm}
\begin{lstlisting}
data list = Nil | Cons int list

fn filter : list -> int -> list = ?xs y =>
  match xs with Nil => Nil
  | Cons x xs' =>
    if x <= y then Cons x (filter xs' y)
    else filter xs' y
\end{lstlisting}
  \caption{Filtering a list}
  \label{fig:list-filter}
\end{wrapfigure}
Before presenting the full details of our approach, we begin with an
overview of \taypsi's strategy for building privacy-preserving
applications.  Consider the simple \lstinline|filter| function
in~\Cref{fig:list-filter}, which drops all the elements in a list
above a certain bound.\footnote{\taypsi supports higher-order
  functions, but our overview will use this specialized version for presentation
  purposes.} Suppose Alice owns some integers, and wants to know which
of those integers are less than some threshold integer belonging to
Bob, but neither party wants to share their data with the other. MPC
protocols allow Alice and Bob to encrypt their data and then jointly
compute \lstinline|filter| using secure operations, without leaking
information about the encrypted data beyond what they can infer from
the final disclosed output. One (insecure) implementation strategy is
to simply encrypt everyone's integers and use a secure version of the
\lstinline!<=! operation to compute the resulting list. Under a
standard semi-honest threat model,\footnote{In this threat model, all
  parties can see the execution traces produced by a small-step
  semantics~\cite{ye2022}.} however, this naive strategy can reveal
private information, via the shape of the input and the intermediate program states.

As an example, assume Alice's input list is %
\lstinline|Cons [2] (Cons [7] (Cons [3] Nil))|, and Bob's input is
\lstinline|[5]|, where square brackets denote secure (encrypted) numbers, i.e.,
only the owner of the integer can observe its value. By observing that Alice's
private data is built from three \lstinline|Cons|s, Bob can already tell Alice
owns exactly $3$ integers, information that Alice may want to keep secret. In
addition, both parties can learn information from the control flow of the
execution of \lstinline|filter|: by observing which branch of \lstinline|if| is
executed, for example, Bob can infer that the second element of Alice's list is
greater than $5$. Thus, even if the integers are secure, both parties can still
glean information about the other's private data.

The particular policy that a secure application enforces can greatly impact the
performance of that application, since the control flow of an application cannot
depend on private data. In the case of our example, this means that the number
of recursive calls to \lstinline!filter! depends on the public information Alice
is willing to share. If Alice only wants to share the maximum length of her
list, for example, its encrypted version must be padded with dummy encrypted
values, and a secure version of \lstinline!filter! must recurse over these dummy
elements, in order to avoid leaking information to Bob through its control flow.
On the other hand, if Alice does not mind sharing the exact number of integers
she owns, the joint computation will not have to go over these values, allowing
a secure version of \lstinline!filter! to be computed more efficiently.

\subsection{Encoding Private Data and Policies}

\begin{wrapfigure}{r}{.32\textwidth}
  \vspace{-.55cm}
\begin{lstlisting}
obliv @list<= (k : nat) =
  if k = 0 then unit
  else unit @+ @int * @list<= (k-1)
\end{lstlisting}
\vspace{-.05cm}
\begin{lstlisting}
obliv @list== (k : nat) =
  if k = 0 then unit
  else @int * @list== (k-1)
\end{lstlisting}
\vspace{-.2cm}
\caption{Oblivious lists with maximum and exact length public views}
  \label{fig:list-oadt}
\end{wrapfigure}
\taypsi allows Alice and Bob to use types to describe what
public information can be shared about their private data (i.e., the
policy governing that data), and its type system guarantees that these
policies are not violated when jointly computing a function over that
data. At a high-level, a privacy policy for a structured
(i.e., algebraic) datatype specifies which of its components are
private and which can be freely shared. We call this publicly shared
information a \emph{public view}, reflecting that it is some
projection of the full data. Formally, policies in \taypsi are encoded
as \emph{oblivious algebraic data types} (OADTs)~\cite{ye2022},
dependent types that take a public view as a parameter. The body of an
OADT is the type of the private components of a data type, which are built
using \emph{oblivious} (i.e., secure) type formers, e.g., oblivious
fixed-width integer (\lstinline!@int!) and oblivious sum
(\lstinline!@+!).\footnote{By convention, we use $\hat{\dot}$ to
  denote the oblivious version of something.} An oblivious sum is
similar to a standard sum, but both its tag and ``payload''
(i.e., component) are obfuscated, so that an attacker cannot
distinguish between a left and right injection.  Essentially, an OADT
is a type-level function that maps the public view of a value to its
\emph{private representation}, i.e., the shape of its private
component.

\Cref{fig:list-oadt} shows two OADTs for the type \lstinline!list!:
\lstinline!@list<=!, whose public view is the maximum length of a list, and
\lstinline!@list==!, whose public view is the exact length. A public view can be
any public data type. We say \lstinline!list! is the \emph{public type} or
public counterpart of the OADTs \lstinline!@list<=! and \lstinline!@list==!. The
key invariant of OADTs is that private values with the same public view are
\emph{indistinguishable} to an attacker, as their private representation is
completely determined by the public view. For example, all private lists of type
\lstinline!@list<= 2! have the same private representation, regardless of the
actual length of the list: %
\fmath{\li!@list<= 2! \equiv \li!unit @+ @int * (unit @+ @int * unit)!}. Thus,
an attacker cannot learn anything about the structure of an OADT, outside of
what is entailed by its public view: an empty list, singleton list, or a list
with two elements of type \lstinline|@list<= 2| all appear the same to an
attacker.

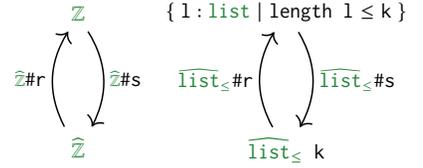
\begin{wrapfigure}{r}{.40\textwidth}
\centering
\vspace{-.48cm}
\begin{tikzcd}[row sep=huge,math mode=false]
  \lstinline!int! \arrow[bend left]{d}{\lstinline!@int\#s!} \\
  \lstinline!@int! \arrow[bend left]{u}{\lstinline!@int\#r!}
\end{tikzcd}
\begin{tikzcd}[row sep=huge,math mode=false]
  \fmath{\Set{\hastype{\lstinline!l! :: \lstinline!list!} | \lstinline!length l! \le \lstinline!k!}} \arrow[bend left]{d}{\lstinline!@list<=\#s!} \\
  \lstinline!@list<= k! \arrow[bend left]{u}{\lstinline!@list<=\#r!}
\end{tikzcd}
\vspace{-.2cm}
  \caption{Public and oblivious types}
  \label{fig:oadt-sr}
\end{wrapfigure}
Conceptually, OADTs generalize the notion of secure fixed-width
integers to secure structured data, as illustrated in
\Cref{fig:oadt-sr}. Every fixed-width integer (of type
\lstinline!int!) can be sent to its secure value in \lstinline!@int!
by ``encryption'', and a secure integer can be converted back to
\lstinline!int! by ``decryption''. In \taypsi, these conversion
functions are called \emph{section} (e.g., \lstinline!@int#s!) and
\emph{retraction} (e.g., \lstinline!@int#r!). The names reflect their
expected semantics: applying retraction to the section of a value
should produce the same value. Importantly, while the oblivious
integer type \lstinline!@int! does not appear to have much structure,
it nonetheless has an implicit policy: the public view of an integer
is its bit width. If we use $32$-bit integers, for example,
\lstinline!int! is the set of all integers whose bit width is $32$,
and \lstinline!@int! is the set of their ``encrypted'' values, related
by a pair of conversion functions. Similarly, \lstinline!@list<= k!
consists of the secure encodings of lists that have at most
\lstinline!k! elements. Like \lstinline!@int!, \lstinline!@list<=! is
equipped with a section function, \lstinline!@list<=#s!, and a
retraction function, \lstinline!@list<=#r!, which convert public
values of \lstinline!list! to their oblivious counterparts and
back. Crucially, just as the oblivious integers in \lstinline!@int!
are indistinguishable, the elements of \lstinline!@list<= k! are also
indistinguishable.

In the implementation of \taypsi, oblivious values are represented
using arrays of secure values. To ensure that attackers cannot learn
anything from the ``memory layout'' of an OADT value, the size of this
array is the same for all values of a particular OADT. As an example,
the encoding of the list %
\lstinline!Cons 10 (Cons 20 Nil)! as an oblivious list of type
\lstinline!@list<= 2! %
is \lstinline!@inr ([10], @inr ([20], ()))!, where \lstinline!@inr!
(\lstinline!@inl!) is the oblivious counterpart of standard sum
injection \lstinline!inr! (\lstinline!inl!). ``Under the hood'' this
oblivious value is represented as an array holding four secure values;
in the remainder of this section, we will informally write this value
as \lstinline![Cons,10,Cons,20]!, where \lstinline![Cons]! is a
synonym of the tag \lstinline![inr]! for readability. As another
example, the empty list \lstinline|Nil| is encoded as %
\lstinline!@inl ()!; it is also represented using an array with four
elements, \lstinline![Nil,-,-,-]!, where the last three elements are
dummy encrypted values (denoted by \lstinline!-!). Our compiler uses
the type of \lstinline!@inl! to automatically pad this array with
these values, in order to ensure that it is indistinguishable from
other private values of \lstinline!@list<= 2!.

\subsection{Enforcing Policies}

Although using OADTs ensures that the representation of private information does
not leak anything, both parties can still learn information by observing the
control flow of a program. In order to protect private data from control flow
channels, \taypsi provides oblivious operations to manipulate private data
safely. One such operation is the atomic conditional
\lstinline|mux|,\footnote{The oblivious version of \lstinline!if! in \taypsi is
called \lstinline!mux!, not \lstinline!@if!, in order to be consistent with the
MPC literature.} a version of \lstinline|if| that fully evaluates \emph{both}
branches before producing its final result. To understand why this evaluation
strategy is necessary, consider the following example of what would happen if we
were to evaluate \lstinline!mux! like a standard \lstinline!if! expression:

{\footnotesize
\setlength{\abovedisplayskip}{0pt}
\setlength{\abovedisplayshortskip}{-.5em}
\[
\begin{array}{l}
  \li!mux [true] ([2] @+ [3]) [4]! \steparrow
  \li![2] @+ [3]! \steparrow
  \li![5]!
\end{array}
\]
}%
Even when all the private data (i.e., the integers in square brackets)
is hidden, an attacker can infer that the private condition is
\lstinline!true! by observing that \li!mux! evaluates to the
expression in its then branch.

With the secure semantics of \lstinline!mux!, however, the following execution
trace does not reveal any private information:

{\footnotesize
\setlength{\abovedisplayskip}{0pt}
\setlength{\abovedisplayshortskip}{-.5em}
\[
\begin{array}{l}
  \li!mux ([0] @<= [1]) ([2] @+ [3]) ([4] @+ [5])! \steparrow
  \li!mux [true] ([2] @+ [3]) ([4] @+ [5])! \steparrow \\
  \li!mux [true] [5] ([4] @+ [5])! \steparrow
  \li!mux [true] [5] [9]! \steparrow
  \li![5]!
\end{array}
\]
}%
Since both branches are evaluated regardless of the private condition,
an attacker cannot infer that condition from this execution trace
(again, all secure values are indistinguishable to an attacker). Thus,
falsifying the condition produces an equivalent trace, modulo the
encrypted data:

{\footnotesize
\setlength{\abovedisplayskip}{0pt}
\setlength{\abovedisplayshortskip}{-.5em}
\[
\begin{array}{l}
  \li!mux ([6] @<= [1]) ([2] @+ [3]) ([4] @+ [5])! \steparrow
  \li!mux [false] ([2] @+ [3]) ([4] @+ [5])! \steparrow \\
  \li!mux [false] [5] ([4] @+ [5])! \steparrow
  \li!mux [false] [5] [9]! \steparrow
  \li![9]!
\end{array}
\]
}

The security-type system of \taypsi ensures all operations on private
data are done in a way that does not reveal any private information,
outside the public information specified by the policies.

\subsection{\emph{Automatically} Enforcing Policies}

Users can directly implement privacy-preserving applications in
\taypsi using OADTs and secure operations, but this requires manually
instrumenting programs so that their control flow only depends on
public information. Under this discipline, the implementation of a
secure function intertwines program logic and privacy policies: the
secure version of \lstinline|filter| requires a different
implementation depending on whether Alice is willing to share the
exact length of her list, or an upper bound on that
length. \taype~\cite{ye2023} decouples these concerns by allowing
programs to include unsafe computations and repairing unsafe
computations at runtime, using a novel form of semantics called
\emph{tape semantics}. As an example of this approach, in \taype, a
secure implementation of \lstinline!filter! that allows Alice to only
share an upper bound on the size of her list can be written as:
\begin{lstlisting}
fn @filter<= : (k : nat) -> @list<= k -> @int -> @list<= k =
  ?k @xs @y => @list<=#s k (filter (@list<=#r k @xs) (@int#r @y))
\end{lstlisting}
The type signature of \lstinline!@filter<=! specifies the policy it
must follow. Intuitively, its implementation first ``decrypts'' the
private inputs, applying the standard \lstinline!filter! function to
those values, and then ``re-encrypts'' the filtered list. In this
example, the retractions of the private inputs \lstinline!@xs! and
\lstinline!@y! are unsafe computations that would violate the desired
policy if they were computed naively. Fortunately, using the tape
semantics prevents this from occurring by deferring these computations
until it is safe to do so. Less fortunately, the runtime overhead of
dynamic policy enforcement makes it hard to scale private applications
manipulating structured data. As one data point, the secure version of
\lstinline!filter! produced by \taype takes more than $5$ seconds to
run with an oblivious list \lstinline!@list<=! with sixteen elements,
and its performance grows exponentially worse as the number of
elements increases.

To understand the source of this slowdown, consider a computation that filters a
private list containing $10$ and $20$ with integer $15$: %
\lstinline!@filter<= 2 [Cons,10,Cons,20] [15]!. The first step in evaluating
this function is to compute %
\lstinline!@list<=#r 2 [Cons,10,Cons,20]!. Completely reducing this
expression leaks information, so tape semantics instead stops
evaluation after producing the following computation:\footnote{We
  refer interested readers to \citet{ye2022, ye2023} for a complete
  accounting of tape semantics.}
\begin{lstlisting}
mux [false] Nil (Cons (@int#r [10]) (mux [false] Nil (Cons (@int#r [20]) Nil)))
\end{lstlisting}
\noindent The two \lstinline![false]!s are the results of securely
checking if the two constructors in the input list are
\lstinline!Nil!. Observe that evaluating either \lstinline!mux! or
\lstinline!@int#r! would reveal private information, so the evaluation
of these operations is deferred. This delayed computation can be
thought of as an ``if-tree'' whose internal nodes are the private
conditions needed to compute the final results, and whose leaves hold
the result of the computation along each corresponding control flow
path. To make progress, tape semantics distributes the context
surrounding a delayed computation, \lstinline!filter! and then
\lstinline!@list<=#s! in this example, into each of its leaves; having
done so, those leaves can be further evaluated. Importantly, in our
example, the leaves of this if-tree are eventually re-encrypted using
\lstinline!@list<=#s!. The tape semantics does so in a secure way, so
that \lstinline!@int#r [10]! becomes \lstinline![10]! again, and each
result list is converted to a secure value of the expected OADT. Once
the branches of a \lstinline|mux| node have been reduced to oblivious
values of the same type, the node itself can be securely reduced
using the secure semantics of \lstinline!mux!. Unfortunately, the
if-tree produced by the tape semantics can grow exponentially large
before its \lstinline|mux| nodes can be reduced. For example, after
applying \lstinline!filter! to the if-tree produced by
\lstinline!@list<=#r!, the resulting if-tree has a leaf corresponding
to every possible list that \lstinline!filter! could produce; the
number of these leaves is exponential in the maximum length of the
input list. As any surrounding computation, i.e., \lstinline!@list<=#s!
in our example, can be distributed to each of these leaves, an
exponential number of computations may need to be performed before the
if-tree can be collapsed.

To remedy these limitations, this paper proposes to instead compile an
insecure program into a secure version that \emph{statically} enforces
a specified policy. To do so, we extend \taype, the secure language of
\citet{ye2023} with \emph{\tpsi-types}, a form of \emph{dependent
  sums} (or dependent pairs) that packs public views and the oblivious
data into a uniform representation. For example, \lstinline!&@list<=!
is the oblivious list \lstinline!@list<=! with its public view: %
\lstinline!&(2, @inr ([10], @inr ([20], ())))&! and %
\lstinline!&(2, @inl ())&! are elements of type \lstinline!&@list<=!,
corresponding to the examples in the previous section. The first component
of this pair-like syntax is a public view and the second component is
an OADT whose public view is exactly the first component. This allows
users to again derive a private filter function from its type
signature:
\begin{lstlisting}
fn @filter<= : &@list<= -> @int -> &@list<= = Mlift filter
\end{lstlisting}
Users no longer need to explicitly provide the public views for either
the output or any intermediate subroutines: both are automatically
inferred. As a result, the policy specification of
\lstinline!@filter<=! more directly corresponds to the type signature
of \lstinline!filter!.  In addition, specifying policies using
\tpsi-types avoids mistakes in the supplied public views: using
\taype, if the programmer mistakenly specifies the return type
\lstinline!@list<= (k-1)! for a secure version of \lstinline|filter|,
for example, the resulting implementation may truncate the last
element of the result list. A keyword \lstinline!Mlift! is used to
translate the standard non-secure function \lstinline!filter! to a
private version that respects the policy specification.

To understand how this translation works, consider a naive approach
where each algebraic data type (ADT) is thought of as an abstract
interface, whose operations correspond to the \emph{introduction} and
\emph{elimination} forms of the algebraic data type\footnote{As
  \taypsi already supports general recursion, we use pattern matching
  instead of recursion schemes as our elimination forms.}. An ADT,
e.g., \lstinline!list!, as well as any corresponding \tpsi-type, e.g.,
\lstinline!&@list<=! and \lstinline!&@list==!, are implementations or
\emph{instances} of this interface. For example, an interface for list
operations is:
\begin{lstlisting}
ListLike t = {Nil : unit -> t; Cons : @int * t -> t; Imatch : t -> (unit -> alpha) -> (@int * t -> alpha) -> alpha}
\end{lstlisting}

As long as \lstinline!&@list<=! and \lstinline!&@list==! implement this interface,
we could straightforwardly translate \lstinline!filter! to a secure version:
{
\vspace{-1em}
\begin{center}
\begin{tabular}{w{c}{.47\textwidth}|w{c}{.48\textwidth}}
\begin{lstlisting}
fn @filter<= : &@list<= -> @int -> &@list<= = ?xs y =>
  @list<=#Imatch xs (?_ => @list<=#Nil ())
    (?(x, xs') =>
      mux (x @<= y) (@list<=#Cons x (@filter<= xs' y))
          (@filter<= xs' y))
\end{lstlisting}&
\begin{lstlisting}
fn @filter== : &@list== -> @int -> &@list== = ?xs y =>
  @list==#Imatch xs (?_ => @list==#Nil ())
    (?(x, xs') =>
      mux (x @<= y) (@list==#Cons x (@filter== xs' y))
          (@filter== xs' y))
\end{lstlisting}
\end{tabular}
\end{center}
}

This strategy does not rely on unsafe retractions like
\lstinline!@list<=#r!, as private data always remains in its secure
form, eliminating the need to defer unsafe computations, which is the
source of exponential slowdowns in \taype. Unfortunately, there are
several obstacles to directly implementing this strategy. First, an
ADT and an OADT may not agree on the type signatures of the abstract
interface. \lstinline!ListLike! fixes the argument types of operations
like \lstinline!Cons! and \lstinline!Imatch!, meaning that
\lstinline!list! is not an instance of this abstract interface,
despite \lstinline!list! being a very reasonable (albeit very
permissive) policy! In general, different OADTs may only be
able to implement operations with specific signatures. Second, a
private function may involve a mixture of oblivious types. Thus, some
functions may need to coerce from one type to a ``more'' secure
version. For example, if the policy of \lstinline!@filter<=! is %
\lstinline!&@list<= -> int -> &@list<=!, %
its second argument \lstinline!y! will need to be converted to %
\lstinline!@int! %
in order to evaluate %
\lstinline!x @<= y!. A secure list that discloses its exact length may
similarly need to be converted to one disclosing its maximum
length. Third, this naive translation results in ill-typed programs,
because the branches of a \lstinline!mux! may have mismatched public
views. In \lstinline!@filter<=!, for example, the branches of
\lstinline|mux| may evaluate to %
\lstinline!&(2, [Cons,10,Cons,20])&! and \lstinline!&(1, [Cons,20])&!,
respectively. Thus, \taypsi's secure type system will (rightly) reject
\lstinline!@filter<=! as leaky. Lastly, the signatures that should be
ascribed to any subsidiary function calls may not be obvious.
Consider the following client of \lstinline!filter!:
\begin{lstlisting}
fn filter5 : list -> list = ?xs => filter xs 5
\end{lstlisting}
If \lstinline|filter5| is given a signature %
\lstinline!&@list<= -> &@list<=!, we would like to use a secure
version of the \lstinline!filter!  function with the type %
\lstinline!&@list<= -> int -> &@list<=!, as the threshold argument is
publicly known. In general, a function may have many private versions,
and we should infer which version to use at each callsite: a recursive
function may even recursively call a different ``version'' of itself.

To solve these challenges, we generalize the abstract interface described above
into a set of more flexible structures, which we collectively refer to as \psistructs
(\Cref{sec:struct}).  Intuitively, each category of \psistructs solves one of
the challenges described above. Our translation algorithm (\Cref{sec:lift})
generates a set of typing \emph{constraints} for the intermediate expressions in
a program. These constraints are then solved using the set of available
\psistructs, resulting in multiple private versions of the necessary functions
and ruling out the infeasible ones, e.g., \lstinline!@filter==!.

\begin{figure}[t]
\footnotesize
\ruleline{\textsc{OADT-structure}}
\begin{minipage}[t]{.48\textwidth}
\begin{lstlisting}[aboveskip=-.5em]
fn @list<=#s : (k : nat) -> list -> @list<= k =
  ?k xs =>
    if k = 0 then ()
    else match xs with Nil => @inl ()
         | Cons x xs' =>
           @inr (@int#s x, @list<=#s (k-1) xs')
\end{lstlisting}
\begin{lstlisting}[aboveskip=0pt]
fn @list<=#view : list -> nat = length
\end{lstlisting}
\end{minipage}%
\begin{minipage}[t]{.52\textwidth}
\begin{lstlisting}[aboveskip=-.5em]
unsafe fn @list<=#r : (k : nat) -> @list<= k -> list =
  ?k =>
    if k = 0 then ?_ => Nil
    else ?xs =>
      @match xs with @inl _ => Nil
      | @inr (x, xs') =>
        Cons (@int#r x) (@list<=#r (k-1) xs')
\end{lstlisting}
\end{minipage}

\ruleline{\textsc{Intro/elim-structure}}
\begin{minipage}[t]{.48\textwidth}
\begin{lstlisting}[aboveskip=-.5em]
fn @list<=#Nil : unit -> &@list<= = ?_ => &(0, ())&

fn @list<=#Cons : @int * &@list<= -> &@list<= =
  ?(x, &(k, xs)&) =>
    &(k+1, @inr (x, xs))&
\end{lstlisting}
\end{minipage}%
\begin{minipage}[t]{.52\textwidth}
\begin{lstlisting}[aboveskip=-.5em]
fn @list<=#Imatch :
  &@list<= -> (unit -> alpha) -> (@int * &@list<= -> alpha) -> alpha =
  ?&(k, xs)& f1 f2 =>
    (if k = 0 then ?_ => f1 ()
     else ?xs =>
       @match xs with @inl _ => f1 ()
       | @inr (x, xs') => f2 (x, &(k-1, xs')&)
         : @list<= k -> alpha) xs
\end{lstlisting}
\end{minipage}

\ruleline{\textsc{Join-structure}}
\begin{lstlisting}
fn @list<=#join : nat -> nat -> nat = max
\end{lstlisting}
\begin{lstlisting}[aboveskip=0pt]
fn @list<=#reshape : (k : nat) -> (k' : nat) -> @list<= k -> @list<= k' = ?k k' =>
  if k' = 0 then ?_ => ()
  else if k = 0 then ?_ => @inl ()
       else ?xs => @match xs with @inl _ => @inl ()
                   | @inr (x, xs') => @inr (x, @list<=#reshape (k-1) (k'-1) xs')
\end{lstlisting}

\caption{\psistructs of \lstinline!@list<=!}
\label{fig:obliv-list-struct}
\end{figure}

\Cref{fig:obliv-list-struct} presents the methods of each category of
\psistructs of \lstinline!@list<=!. The first two methods, \lstinline!@list<=#s!
and \lstinline!@list<=#r!, are its section and retraction functions, belonging
to the OADT-structure category. Unlike \taype, these two functions are not
directly used to derive secure implementations of functions. In fact, our type
system guarantees that retraction functions are never used in a secure
computation, because \taypsi does not rely on tape semantics to repair unsafe
computation (the \lstinline!unsafe fn! keyword tells our type checker that
\lstinline!@list<=#r! is potentially leaky). Our implementation of \taypsi
exposes section and retraction functions as part of the API of the secure
library it generates, however, so that client programs can conceal their private
input and reveal the output of secure computations. This structure also includes
a \lstinline!view! method, which our translation uses to select the public view
needed to safely convert a \lstinline!list!  into a \lstinline!&@list<=!.
\Cref{fig:obliv-list-struct} does not show coercion methods, but the
programmers can define a coercion from \lstinline!&@list==! to
\lstinline!&@list<=!, for example.

The next set of methods belong to the intro-structure and elim-structure
category. These introduction (\lstinline!@list<=#Nil! and
\lstinline!@list<=#Cons!) and elimination (\lstinline!@list<=#Imatch!) methods
construct and destruct private list, respectively. As we construct and
manipulate data, these methods build the private version, calculate its public
view, and record that view in \tpsi-types.  Their type signatures are specified
by the programmers, as long as the signatures are \emph{compatible} with
\lstinline!int * list! (\Cref{sec:struct}). 

The \lstinline!join! and \lstinline!reshape! methods in the join-structure
category enable translated programs to include private conditionals whose
branches return OADT values with different public views. As an example, consider
the following private conditional whose branches have \tpsi-types:
\begin{lstlisting}
mux [true] &(2, [Cons,10,Cons,20])& &(1, [Cons,20])&
\end{lstlisting}
To build a version of this program that does not reveal
\lstinline![true]!, \taypsi uses \lstinline!join! to calculate a
common public view that ``covers'' both branches. In this example,
\lstinline!@list<=#join! chooses a public view of $2$, as a list with
at most one element also has at most two elements. Our translation
then uses the \lstinline!reshape! method to convert both branches to
use this common public view. In our example, \lstinline![Cons,20]!, an
oblivious list of maximum length $1$, is converted into the list %
\lstinline![Cons,20,Nil,-]!, which has maximum length $2$. Since both
branches in the resulting program have the same public view, it is safe
to evaluate \lstinline|mux|: the resulting list is equivalent to %
\li!&(2, mux [true] [Cons,10,Cons,20] [Cons,20,Nil,-])&!.
As we will see later, not all OADTs admit join structures, e.g.,
\lstinline!@list==!, but our translation generates constraints that take
advantage of any that are available, failing when these constraints cannot be
resolved in a way that guarantees security. Note that these two methods are key
to avoiding the slowdown exhibited by \taype's enforcement strategy: they allow
functions that may return different private representations to be \emph{eagerly}
evaluated, instead of being lazily deferred in a way that requires an
exponential number of subcomputations to resolve.

In summary, to develop a secure application in \taypsi, programmers
first implement its desired functionality, e.g., \lstinline!filter!,
in the public fragment of \taypsi, independently of any particular
privacy policy. Policies are separately defined as oblivious algebraic
data types, e.g., \lstinline!@list<=!, and their \psistructs. Users
can then automatically derive a secure version of their application by
providing the desired policy in the form of a type signature involving
\tpsi-types, relying on \taypsi's compiler to produce a
privacy-preserving implementation. The type system of \taypsi, like
\taype's, provides a strong security guarantee in the form of an
obliviousness theorem (\Cref{thm:obliviousness}). This obliviousness
theorem is a variant of noninterference~\cite{goguen1982}, and ensures
that well-typed programs in \taypsi are \emph{secure by construction}:
no private information can be inferred even by an attacker capable of
observing every state of a program's execution. Our compilation
algorithm is further guaranteed to generate a secure implementation
that preserves the behavior of the original program
(\Cref{thm:correct}).\footnote{\taypsi's formal guarantees
  (\Cref{sec:meta+lift}) do not cover equi-termination of the source
  and target programs: when the public view lacks sufficient
  information to bound the computation of the original program, the
  secure version will not terminate, in order to avoid leaking
  information through its termination behavior.}

The following three sections formally develop the language \taypsi, the
\psistructs, and our translation algorithm.

\section{\taypsi, Formally}
\label{sec:formal}

This section presents \loadtpsi, the core calculus for secure computation that
we will use to explain our translation. This calculus extends the existing
\loadt~\cite{ye2022} calculus with \tpsi-types, and uses ML-style ADTs in lieu
of explicit \lstinline!fold! and \lstinline!unfold! operations.\footnote{For
simplicity, \loadtpsi does not include public sums and oblivious integers, which
are straightforward to add.}

\subsection{Syntax}

\begin{figure}[t]
\footnotesize

\begin{tabular}{r@{\hspace{2pt}}C@{}L@{\hspace{6pt}}l@{\hspace{6pt}}C@{\hspace{2pt}}L@{\hspace{6pt}}l}
\lstinline|e,tau| &\production&&&&& \textsc{Expressions} \\
&\mid& \lstinline!unit! \mid \lstinline!bool! \mid \lstinline!@bool! \mid%
  \lstinline!tau*tau! \mid \lstinline!tau@+tau! & public \& oblivious types%
&\mid& \graylight{\lstinline!&@T!} & \tpsi-type \\
&\mid& \lstinline!Pix:tau,tau! \mid%
  \lstinline!?x:tau=>e! & dependent function%
&\mid& \lstinline!(e,e)! \mid%
  \graylight{\lstinline!&(e,e)&!} & pair \& \tpsi-pair \\
&\mid& \lstinline!()! \mid \lstinline!b! \mid \lstinline!x! \mid \lstinline!T! %
  & literals \& variables%
&\mid& \graylight{\lstinline!pi_1 e!} \mid%
  \graylight{\lstinline!pi_2 e!} & product and \tpsi-type projection \\
&\mid& \lstinline!let x = e in e! & let binding%
&\mid& \lstinline!@inl<tau> e! \mid%
  \lstinline!@inr<tau> e! & oblivious sum injection \\
&\mid& \lstinline!e e! \mid \lstinline|C e| \mid \lstinline|@T e|%
  & applications%
&\mid& \lstinline!@match e with x=>e|x=>e! & oblivious sum elimination \\
&\mid& \lstinline!if e then e else e! & conditional%
&\mid& \lstinline!match e with !\overline{\lstinline!C x=>e!}%
  & ADT elimination \\
&\mid& \lstinline!mux e e e! & oblivious conditional%
&\mid& \lstinline!@bool#s e! & boolean section \\
\end{tabular}%
\vspace{4pt}
\begin{tabular}{r@{\hspace{2pt}}C@{}L@{\hspace{6pt}}lr@{\hspace{2pt}}C@{\hspace{2pt}}L@{\hspace{6pt}}l}
\lstinline|D| &\production&& \textsc{Global Definitions}%
&\lstinline|@w| &\production& \lstinline!unit! \mid \lstinline!@bool! \mid%
  \lstinline!@w*@w! \mid \lstinline!@w@+@w! &
  \textsc{Obliv. Type Values} \\
&\mid& \lstinline!data T = !\overline{\lstinline!C tau!} &%
  algebraic data type%
&\lstinline|@v| &\production& \lstinline!()! \mid \lstinline![b]! \mid%
  \lstinline!(@v,@v)!&%
  \textsc{Obliv. Values} \\
&\mid& \lstinline!fn x:tau = e! & (recursive) function%
&&\mid&\lstinline![inl<@w> @v]! \mid \lstinline![inr<@w> @v]! & \\
&\mid& \lstinline!obliv @T (x:tau) = tau! &%
  (recursive) obliv. type%
&\lstinline|v| &\production& \lstinline|@v| \mid \lstinline!b! \mid%
  \lstinline!(v,v)! \mid \graylight{\lstinline!&(v,v)&!} \mid%
  \lstinline!?x:tau=>e! \mid \lstinline!C v! & \textsc{Values} \\
\end{tabular}

\caption{\loadtpsi syntax with extensions to \loadt highlighted }
\label{fig:syntax}
\end{figure}

\Cref{fig:syntax} presents the syntax of \loadtpsi. Types and
expressions are in the same syntax class, as \loadtpsi is dependently
typed, but we use \lstinline!e! for expressions and \lstinline!tau!
for types when possible. A \loadtpsi program consists of a set of
\emph{global definitions} of data types, functions and oblivious
types. Definitions in each of these classes are allowed to refer to
themselves, permitting recursive types and general recursion in both
function and oblivious type definitions. We use \lstinline!x!  for
variable names, \lstinline!C! for constructor names, \lstinline!T!
for type names, and \lstinline!@T! for oblivious type names. Each
constructor of an ADT definition takes exactly one argument, but this
does not harm expressivity: this argument is \lstinline!unit! for
constructors that take no arguments, e.g., \lstinline!Nil!, and a
tuple of types for constructors that have more than one argument,
e.g., \lstinline!Cons! takes an argument of type %
\lstinline!int * list!.

In addition to standard types and dependent function types
(\lstinline!Pi!), \loadtpsi includes oblivious booleans
(\lstinline!@bool!) and oblivious sum types (\lstinline!@+!). The
elimination forms of these types are oblivious
conditionals \lstinline!mux! and oblivious case analysis
\lstinline!@match!, respectively. The branches of both expressions
must be private and each branch has to be fully evaluated
before the expression can take an atomic step to a final result. Boolean
section \lstinline!@bool#s! is a primitive operation that ``encrypts''
a boolean expression to an oblivious version. Oblivious injection
\lstinline!@inl! and \lstinline!@inr! are the oblivious counterparts
of the standard constructors for sums. Other terms are mostly
standard, although let bindings (\lstinline!let!),
conditionals (\lstinline!if!) and pattern matching (\lstinline!match!)
are allowed to return a type, as \loadtpsi supports type-level
computation.

The key addition over \loadt is the \emph{\tpsi-type},
\lstinline!&@T!. It is constructed from a pair expression
\lstinline!&(DOT,DOT)&! that packs the public view and the oblivious
data together, and has the same eliminators \lstinline!pi_1!  and
\lstinline!pi_2! as normal products. As an example, %
\lstinline!&(3, @list<=#s 3 (Cons 1 Nil))&! creates a \tpsi-pair of type
\lstinline!&@list<=! with public view $3$, using the section function
from \Cref{fig:obliv-list-struct}. Projecting out the second component
of a pair using \lstinline!pi_2! produces a value of type
\lstinline!@list<= 3!. A \tpsi-type is essentially a dependent sum type
(\lstinline!Sigmax:tau,@T x!), with the restriction that
\lstinline!tau!  is the public view of \lstinline!@T!, and that
\lstinline!@T x! is an oblivious type.

Since \loadtpsi has type-level computation, oblivious types have
normal forms; oblivious type values (\lstinline!@w!) are essentially
polynomials formed by primitive oblivious types.  We also have the
oblivious values of oblivious boolean and sum type. Note that these
``boxed'' values only appear at runtime, our semantics use these to
model encrypted booleans and tagged sums.

\subsection{Semantics}

\begin{figure}[t]
\footnotesize
\jbox{\step{\li!e! ~> \li!e'!}}
\begin{mathpar}
  \inferrule[S-Ctx]{
    \step{\li!e! ~> \li!e'!}
  }{
    \step{\li!ectx[e]! ~> \li!ectx[e']!}
  }

  \inferrule[S-Fun]{
    \li!fn x:tau = e! \in \gctx
  }{
    \step{\li!x! ~> \li!e!}
  }

  \inferrule[S-OADT]{
    \li!obliv @T (x:tau) = tau'! \in \gctx
  }{
    \step{\li!@T v! ~> \li![v/x]tau'!}
  }

  \inferrule[S-App]{
  }{
    \step{\li!(?x:tau=>e) v! ~> \li![v/x]e!}
  }

  \inferrule[S-IfTrue]{
  }{
    \step{\li!if true then e_1 else e_2! ~> \li!e_1!}
  }

  \inferrule[S-MuxTrue]{
  }{
    \step{\li!mux [true] v_1 v_2! ~> \li!v_1!}
  }

  \inferrule[S-Match]{
  }{
    \step{\li!match C_i v with\ !\overline{\li!C x=>e!} ~> \li![v/x]e_i!}
  }

  \inferrule[S-Proj$_1$]{
  }{
    \step{\li!pi_1 (v_1,v_2)! ~> \li!v_1!}
  }

  \inferrule[S-Sec]{
  }{
    \step{\li!@bool\#s b! ~> \li![b]!}
  }

  \inferrule[S-OInl]{
  }{
    \step{\li!@inl<@w> @v! ~> \li![inl<@w> @v]!}
  }

  \inferrule[S-PsiProj$_1$]{
  }{
    \step{\li!pi_1 &(v_1,v_2)&! ~> \li!v_1!}
  }

  \inferrule[S-OMatchL]{
    \ovalty{\li!@v_2! <~ \li!@w_2!}
  }{
    \step{\li!@match [inl<@w_1@+@w_2> @v] with x=>e_1|x=>e_2! ~>%
        \li!mux [true] ([@v/x]e_1) ([@v_2/x]e_2)!
    }
  }
\end{mathpar}

\vspace{4pt}
\begin{tabular}{RCLl}
\lstinline|ectx| &\production&\textsc{Evaluation Context}& \\
&\mid& \lstinline![]*tau! \mid \lstinline!@w*[]! \mid%
  \lstinline![]@+tau! \mid \lstinline!@w@+[]! \mid%
  \lstinline!e []! \mid \lstinline![] v! \mid%
  \lstinline!C []! \mid \lstinline!@T []! \mid%
  \lstinline!if [] then e else e! & \\
&\mid& \lstinline!mux [] e e! \mid%
  \lstinline!mux v [] e! \mid \lstinline!mux v v []! & \\
&\mid& \lstinline!([],e)! \mid \lstinline!(v,[])! \mid%
  \lstinline!&([],e)&! \mid \lstinline!&(v,[])&! \mid%
  \lstinline!pi_1 []! \mid%
  \lstinline!match [] with !\overline{\lstinline!C x => e!} & \\
&\mid& \lstinline!@inl<[]> e! \mid \lstinline!@inl<@w> []! \mid%
  \lstinline!@match [] with x=>e|x=>e! \mid%
  \lstinline!@bool#s []! \mid \ldots &
\end{tabular}

\caption{Selected small-step semantics rules of \loadtpsi}
\label{fig:semantics}

\end{figure}

\Cref{fig:semantics} shows a selection of the small-step semantics rules of
\loadtpsi (the full rules are included in the appendix), with judgment
\step{\gctx |- \lstinline!e! ~> \lstinline!e'!}. The global context \gctx{} is a
map from names to a global definition, which is elided for brevity as it is
fixed in these rules. The semantics of \loadtpsi is similar to \loadt, with the
addition of \textsc{S-PsiProj$_1$} (and \textsc{S-PsiProj$_2$}) to handle the
projection of dependent pairs, which is simply the same as normal projection.
\textsc{S-Ctx} reduces subterms according to the evaluation contexts defined in
\Cref{fig:semantics}. The first few contexts take care of the type-level
reduction of product and oblivious sum type. The type annotation of oblivious
injection \lstinline!@inl! and \lstinline!@inr! is reduced to a type value
first, before reducing the payload. The evaluation contexts for \lstinline!mux!
capture the intuition that all components of a private conditional have to be
normalized to values first to avoid leaking the private condition through
control flow channels.

\textsc{S-OMatchL} (and \textsc{S-OMatchR}) evaluates a pattern
matching expression for oblivious sums. Similar to \lstinline|mux|,
oblivious pattern matching needs to ensure the reduction does not
reveal private information about the discriminee, e.g., whether it is
the left injection or right injection. To do so, we reduce a
\lstinline|@match| to an oblivious conditional that uses the private
tag. The pattern variable in the ``correct'' branch is of course
instantiated by the payload in the discriminee, while the pattern
variable in the ``wrong'' branch is an arbitrary value of the
corresponding type, synthesized from the judgment \ovalty{\li!@v! <~
  \li!@w!}, whose definition is in appendix. When evaluating a
\lstinline|@match| statement whose discriminee is %
\lstinline![inl<@bool@+@bool*@bool> [true]]!, the pattern variable in
the second branch can be substituted by %
\lstinline!([true],[true])!, \lstinline!([false],[true])!, or any
other pair of oblivious booleans.

\subsection{Type System}
Similar to \loadt, types in \loadtpsi are classified by \emph{kinds}
which specify how protected a type is, in addition to ensuring the
types are well-formed. For example, an oblivious type, e.g.,
\lstinline!@bool!, kinded by \lstinline!*@O!, can be used as branches
of an oblivious conditional, but not as a public view, which can only
be kinded by \lstinline!*@P!. A mixed kind \lstinline!*@M! is used to
classify function types and types that consist of both public and
oblivious components, e.g., \lstinline!bool*@bool!. A type with a
mixed kind cannot be used as a public view or in private context.

\begin{figure}[t]
\footnotesize
\jbox{\type{\li!e! :: \li!tau!}}
\begin{mathpar}
  \inferrule[T-Conv]{
    \type{\li!e! :: \li!tau!} \\
    \typequiv{\li!tau! == \li!tau'!} \\
    \kind{\li!tau'! :: \li!*@*!}
  }{
    \type{\li!e! :: \li!tau'!}
  }

  \inferrule[T-Abs]{
    \type{\extctx{\li!x! :: \li!tau_1!} |- \li!e! :: \li!tau_2!} \\
    \kind{\li!tau_1! :: \li!*@*!}
  }{
    \type{\li!?x:tau_1=>e! :: \li!Pix:tau_1,tau_2!}
  }

  \inferrule[T-App]{
    \type{\li!e_2! :: \li!Pix:tau_1,tau_2!} \\
    \type{\li!e_1! :: \li!tau_1!}
  }{
    \type{\li!e_2 e_1! :: \li![e_1/x]tau_2!}
  }

  \inferrule[T-If]{
    \type{\li!e_0! :: \li!bool!} \\\\
    \type{\li!e_1! :: \li![true/y]tau!} \\
    \type{\li!e_2! :: \li![false/y]tau!}
  }{
    \type{\li!if e_0 then e_1 else e_2! :: \li![e_0/y]tau!}
  }

  \inferrule[T-Mux]{
    \type{\li!e_0! :: \li!@bool!} \\
    \kind{\li!tau! :: \li!*@O!} \\\\
    \type{\li!e_1! :: \li!tau!} \\
    \type{\li!e_2! :: \li!tau!}
  }{
    \type{\li!mux e_0 e_1 e_2! :: \li!tau!}
  }

  \inferrule[T-PsiPair]{
    \li!obliv @T (x:tau) = tau'! \in \gctx \\\\
    \type{\li!e_1! :: \li!tau!} \\
    \type{\li!e_2! :: \li!@T e_1!}
  }{
    \type{\li!&(e_1,e_2)&! :: \li!&@T!}
  }

  \inferrule[T-PsiProj$_{1}$]{
    \li!obliv @T (x:tau) = tau'! \in \gctx \\
    \type{\li!e! :: \li!&@T!}
  }{
    \type{\li!pi_1 e! :: \li!tau!}
  }

  \inferrule[T-PsiProj$_{2}$]{
    \li!obliv @T (x:tau) = tau'! \in \gctx \\
    \type{\li!e! :: \li!&@T!}
  }{
    \type{\li!pi_2 e! :: \li!@T (pi_1 e)!}
  }
\end{mathpar}

\footnotesize
\jbox{\kind{\li!tau! :: \li!kappa!}}
\begin{mathpar}
  \inferrule[K-Sub]{
    \kind{\li!tau! :: \li!kappa!} \\
    \li!kappa! \sqsubseteq \li!kappa'!
  }{
    \kind{\li!tau! :: \li!kappa'!}
  }

  \inferrule[K-OADT]{
    \li!obliv @T (x:tau) = tau'! \in \gctx \\
    \type{\li!e! :: \li!tau!}
  }{
    \kind{\li!@T e! :: \li!*@O!}
  }

  \inferrule[K-Pi]{
    \kind{\li!tau_1! :: \li!*@*!} \\
    \kind{\extctx{\li!x! :: \li!tau_1!} |- \li!tau_2! :: \li!*@*!}
  }{
    \kind{\li!Pix:tau_1,tau_2! :: \li!*@M!}
  }

  \inferrule[K-OSum]{
    \kind{\li!tau_1! :: \li!*@O!} \\
    \kind{\li!tau_2! :: \li!*@O!}
  }{
    \kind{\li!tau_1@+tau_2! :: \li!*@O!}
  }

  \inferrule[K-Psi]{
    \li!obliv @T (x:tau) = tau'! \in \gctx
  }{
    \kind{\li!&@T! :: \li!*@M!}
  }

  \inferrule[K-If]{
    \type{\li!e_0! :: \li!bool!} \\
    \kind{\li!tau_1! :: \li!*@O!} \\
    \kind{\li!tau_2! :: \li!*@O!}
  }{
    \kind{\li!if e_0 then tau_1 else tau_2! :: \li!*@O!}
  }
\end{mathpar}
  \caption{Selected typing and kinding rules of \loadtpsi}
  \label{fig:kinding}
\end{figure}

\begin{figure}[t]
\footnotesize
\jbox{\dtype{\li!D!}}
\begin{mathpar}
  \inferrule[DT-Fun]{
    \kind{\empctx |- \li!tau! :: \li!*@*!} \\
    \type{\empctx |- \li!e! :: \li!tau!}
  }{
    \dtype{\li!fn x:tau = e!}
  }

  \inferrule[DT-ADT]{
    \forall i.\; \kind{\empctx |- \li!tau_i! :: \li!*@P!}
  }{
    \dtype{\li!data T =\ !\overline{\li!C tau!}}
  }

  \inferrule[DT-OADT]{
    \kind{\empctx |- \li!tau! :: \li!*@P!} \\
    \kind{\hastype{\li!x! :: \li!tau!} |- \li!tau'! :: \li!*@O!}
  }{
    \dtype{\li!obliv @T (x:tau) = tau'!}
  }
\end{mathpar}
  \caption{\loadtpsi global definitions typing}
  \label{fig:gctx-typing}
\end{figure}

The type system of \loadtpsi is defined by a pair of typing and kinding
judgments, \type{\gctx; \tctx |- \li!e! :: \li!tau!} and \kind{\gctx; \tctx |-
\li!tau! :: \li!kappa!}, with global context \gctx{} (which is again elided for
brevity) and the standard typing context \tctx{}. \Cref{fig:kinding} presents a
subset of our typing and kinding rules; the full rules are in appendix.

The security type system~\cite{sabelfeld2003} of \loadtpsi enforces a few key
invariants. First, oblivious types can only be constructed from oblivious types,
which is enforced by the kinding rules, such as \textsc{K-OSum}. Otherwise, the
attacker could infer the private tag of an oblivious sum, e.g.,
\lstinline!bool@+unit!, by observing its public payload. Second, oblivious
operations, e.g., \lstinline!mux!, require their subterms to be oblivious, to
avoid leaking private information via control flow channels. \textsc{T-Mux}, for
example, requires both branches to be typed by an oblivious type, otherwise an
attacker may infer the private condition by observing the result, as in %
\lstinline!mux [true] true false!. Third, type-level computation is only defined
for oblivious types and cannot depend on private information. Thus,
\textsc{K-If} requires both branches to have oblivious kinds, and the condition
to be a public boolean. The type %
\lstinline!mux [true] unit @bool! is ill-typed, since the ``shape'' of the data
reveals the private condition.

The typing rules for \tpsi-types are defined similarly to the rules of
standard dependent sums. \textsc{T-PsiPair} introduces a dependent
pair, where the type of the second component depends on the first
component. In contrast to standard dependent sum type, \tpsi-type has
the restriction that the first component must be public, and the second
component must be oblivious. This condition is implicitly enforced by
the side condition that \lstinline!@T! is an OADT with public view
type \lstinline!tau!. \Cref{fig:gctx-typing} shows the typing rules
for global definitions; \textsc{DT-OADT} prescribes exactly this
restriction. The rules for the first and second projection of
\tpsi-type, \textsc{T-PsiProj$_1$} and \textsc{T-PsiProj$_2$}, are
very similar to the corresponding rules for standard dependent sum
types. Observe that a \tpsi-type always has mixed kind, as in
\textsc{K-Psi}, because it consists of both public and oblivious
components.

\textsc{T-Conv} allows conversion between equivalent types, such as
\lstinline!if true then @bool else unit! and \lstinline!@bool!. The equivalence
judgment \typequiv{\li!tau! == \li!tau'!} is defined by a set of \emph{parallel
reduction} rules, which we elide here. The converted type is nonetheless
required to be well-kinded.

Note that these rules cannot be used to type check retraction functions, e.g.,
\lstinline!@list<=#r! from \Cref{fig:obliv-list-struct}, and for good reason:
these functions reveal private information. Nevertheless, we still want to check
that these sorts of ``leaky'' functions have standard type safety properties,
i.e., progress and preservation. To do so, we use a version of these rules that
simply omit some security-related side-conditions about oblivious kinding:
removing \kind{\li!tau! :: \li!*@O!} from \textsc{T-Mux} allows the branches of
a \li!mux! to disclose the private condition, for example. The implementation
of \taypsi's type checker uses a ``mode'' flag to indicate whether
security-related side-conditions should be checked. Our implementation ensures
that secure functions never use any leaky functions.

\subsection{Metatheory}
\label{sec:formal:meta}

With our addition of \tpsi-types, \loadtpsi enjoys the standard type safety
properties (i.e., progress and preservation), and, more importantly, the same
security guarantees as \loadt:

\begin{theorem}[Obliviousness]
  \label{thm:obliviousness}
  If\/ \indist{\li!e_1! ~ \li!e_2!} and\/ \type{\empctx |- \li!e_1! ::
  \li!tau_1!} and\/ \type{\empctx |- \li!e_2! :: \li!tau_2!}, then
  \begin{enumerate}
    \item \step[n]{\li!e_1! ~> \li!e_1'!} if and only if\/ \step[n]{\li!e_2! ~>
    \li!e_2'!} for some\/ \lstinline|e_2'|.
    \item if\/ \step[n]{\li!e_1! ~> \li!e_1'!} and\/ \step[n]{\li!e_2! ~>
    \li!e_2'!} , then\/ \indist{\li!e_1'! ~ \li!e_2'!}.
  \end{enumerate}
\end{theorem}

Here, \indist{\li!e! ~ \li!e'!} means the two expressions are
indistinguishable, i.e., they only differ in their unobservable
oblivious values, and \step[n]{\li!e! ~> \li!e'!} means \lstinline!e!
reduces to \lstinline!e'! in exactly $n$ steps. This obliviousness
theorem provides a strong security guarantee: well-typed programs that
are indistinguishable produce traces that are pairwise
indistinguishable. In other words, an attacker cannot infer any
private information even by observing the execution trace of a
program. All these results are mechanized in the Coq theorem prover,
including the formalization of the core calculus and the proofs of
soundness and obliviousness theorems.

\section{\psistructs and Declarative Lifting}
\label{sec:struct}

While our secure language makes it possible to encode structured data
and privacy policies, and use them in a secure way, it does not quite
achieve our main goal yet, i.e., to decouple privacy policies and
programmatic concerns. To do so, we allow the programmers to implement
the functionality of their secure application in a conventional way,
that is using only the public, nondependent fragment of \taypsi. We
make this fragment explicit by requiring such programs to have
\emph{simple types}, denoted by \lstinline!eta!, defined in
\Cref{fig:types-erasure}. For example, \lstinline!filter! has simple
type \lstinline!list -> int -> list!. Programs of simple types are the
source programs to our lifting process that translates them to a
private version against a policy, which stipulates the public
information allowed to disclose in the program input and output. This
policy on private functionality is specified by a \emph{specification
  type}, denoted by \lstinline!theta!, defined also in
\Cref{fig:types-erasure}. For example, \lstinline!@filter<=! has
specification type \lstinline!&@list<= -> @int -> &@list<=!. Note that
dependent types are not directly allowed in specifications,
they are instead encapsulated in \tpsi-types.  Simple types and
specification types are additionally required to be well-kinded under
empty local context, i.e., all ADTs and OADTs appear in them are
defined.

However, not all specification types are valid with respect to a
simple type. It is nonsensical to give \lstinline!@filter! the
specification type \lstinline!@int -> @bool!, for example. The
specification types should still correspond to the simple types in
some way: the specification type corresponding to \lstinline!list!
should at least be ``list-like''. This correspondence is formally
captured in the erasure function in \Cref{fig:types-erasure}, which
maps a specification type to the ``underlying'' simple type. For
example, \lstinline!&@list<=! is erased to \lstinline!list!. This
function clearly induces an equivalence relation: the erasure
\lstinline|&[theta]&| is the representative of the equivalence
class. We call this equivalence class a \emph{compatibility class},
and say two types are \emph{compatible} if they belong to the same
compatibility class. For example, \lstinline!list!, \lstinline!&@list<=!
and \lstinline!&@list==!  are in the same compatibility class
\lstinline![list]!. This erasure operation is straightforwardly
extended to typing contexts, \erase{\tctx}, by erasing every
specification type in \tctx{} and leaving other types untouched.

Our translation transforms source programs with simple types into
target programs with the desired (compatible) specification types. As
mentioned in \Cref{sec:overview}, this \emph{lifting} process depends
on a set of \emph{$\psi$-structures} which explain how to translate
certain operations associated with an OADT.

\begin{figure}[t]
\footnotesize

\begin{minipage}{.4\textwidth}
\raggedright
\textsc{Simple types}

\begin{tabular}{RCLl}
\lstinline|eta| &\production& \lstinline!unit! \mid \lstinline!bool! \mid \lstinline!T! \mid \lstinline!eta*eta! \mid \lstinline!eta->eta! & \\
\end{tabular}

\vspace{4pt}

\textsc{specification types}

\begin{tabular}{RCLl}
\lstinline|theta| &\production& \lstinline!unit! \mid \lstinline!bool! \mid \lstinline!@bool! \mid \lstinline!T! \mid \lstinline!&@T! \mid \lstinline!theta*theta! \mid \lstinline!theta->theta! & \\
\end{tabular}
\end{minipage}%
\begin{minipage}{.6\textwidth}
\jbox{\lstinline!&[theta]&!}
\begin{mathpar}
  \li!&[unit]&! = \li!unit! \and%
  \li!&[bool]&! = \li!&[@bool]&! = \li!bool! \and%
  \li!&[T]&! = \li!T! \quad\text{where \li!T! is an ADT} \\
  \li!&[&@T]&! = \li!T! \quad\text{where \li!@T! is an OADT for \li!T!} \\
  \li!&[theta*theta]&! = \li!&[theta]&*&[theta]&! \and%
  \li!&[theta->theta]&! = \li!&[theta]&->&[theta]&!
\end{mathpar}
\end{minipage}

  \caption{Simple types, specification types and erasure}
  \label{fig:types-erasure}
\end{figure}

\subsection{OADT Structures}

Every global OADT definition \lstinline|@T| must be equipped with an
\emph{OADT-structure}, defined below.

\begin{definition}[OADT-structure]
  \label{def:oadt-struct}

  An OADT-structure of an OADT \lstinline|@T|, with public view
  type \lstinline|tau|, consists of the following (\taypsi) type and functions:
  \begin{itemize}
    \item A public type \haskind{\li!T! :: \li!*@P!}, which is the public
      counterpart of \lstinline|@T|. We say \lstinline|@T| is an OADT
      for \lstinline|T|.
    \item A section function \hastype{\li!s! :: \li!Pik:tau,T->@T k!}, which
      converts a public type to its oblivious counterpart.
    \item A retraction function \hastype{\li!r! :: \li!Pik:tau,@T k->T!}, which
      converts an oblivious type to its public version.
    \item A public view function \hastype{\li!VIEW! :: \li!T->tau!}, which creates
      a valid view of the public type.
    \item A binary relation \hasview{} over values of types \lstinline!T! and
    \lstinline!tau!; \fmath{\li!v! \hasview \li!k!} reads as \lstinline|v| has
    public view \lstinline|k|, or \lstinline|k| is a valid public view of
    \lstinline|v|.
  \end{itemize}

  These operations are required to satisfy the following axioms:
  \begin{itemize}
    \item (\textsc{A-O$_{1}$}) \lstinline|s| and \lstinline|r| are
      a valid section and retraction, i.e., \lstinline|r| is a left-inverse for
      \lstinline|s|, given a valid public view: for any values \hastype{\li!v! ::
      \li!T!}, \hastype{\li!k! :: \li!tau!} and \hastype{\li!@v! :: \li!@T k!},
      if \fmath{\li!v! \hasview \li!k!} and \step*{\li!s k v! ~> \li!@v!}, then
      \step*{\li!r k @v! ~> \li!v!}.
    \item (\textsc{A-O$_{2}$}) the result of \lstinline|r| always has valid
      public view: \step*{\li!r k @v! ~> \li!v!} implies \fmath{\li!v! \hasview
      \li!k!} for all values \hastype{\li!k! :: \li!tau!}, \hastype{\li!@v!
      :: \li!@T k!} and \hastype{\li!v! :: \li!T!}.
    \item (\textsc{A-O$_{3}$}) \lstinline|VIEW| produces a valid public
      view: \step*{\li!VIEW v! ~> \li!k!} implies \fmath{\li!v! \hasview
      \li!k!}, given any values \hastype{\li!v! :: \li!T!} and \hastype{\li!k!
      :: \li!tau!}.
  \end{itemize}
\end{definition}

For example, \lstinline!@list<=! is equipped with the OADT-structure
with the public type \lstinline!list!, section function
\lstinline!@list<=#s!, retraction function \lstinline!@list<=#r! and view
function \lstinline!@list<=#view!, all of which are
shown in \Cref{fig:obliv-list-struct}. \taypsi users do not need to
explicitly give the public type of an OADT-structure, as it can be
inferred from the types of the other functions. The binary relation
\hasview{} is only used in the proof of correctness of our translation,
so \taypsi users can also elide it. In the case of \lstinline!@list<=!,
\hasview{} simply states the length of the list is no larger than the
public view.

\subsection{Join Structures}

In order for \tpsi-types to be flexibly used in the branches of secure
control flow structures, our translation must be able to find a common
public view for both branches, and to convert an OADT to use this
view. To do so, an OADT can \emph{optionally} be equipped with a
\emph{join-structure}.

\begin{definition}[join-structure]
  \label{def:join-struct}

  A join-structure of an OADT \lstinline|@T| for \lstinline|T|, with public view
  type \lstinline|tau|, consists of the following operations:
  \begin{itemize}
  \item A binary relation \ple{} on \lstinline|tau|, used to compare
    two public views.
  \item A join function \hastype{\li!JOIN! :: \li!tau->tau->tau!},
    which computes an upper bound of two public views\footnote{It is a
      bit misleading to call the operation \li!JOIN! ``join'', as it
      only computes an upper bound, not necessarily the lowest
      one. However, it \emph{should} compute a supremum for
      performance reasons: intuitively, larger public view means more
      padding.}.
  \item A reshape function %
    \hastype{\li!RESHAPE! :: \li!Pik:tau,Pik':tau,@T k->@T k'!}, which
    converts an OADT to one with a different public view.
  \end{itemize}
  such that:
  \begin{itemize}
    \item (\textsc{A-R$_{1}$}) \ple{} is a partial order on \lstinline|tau|.
    \item (\textsc{A-R$_{2}$}) the join function produces an upper bound: given
      values \lstinline!k_1!, \lstinline!k_2! and \lstinline!k! of
      type \lstinline!tau!, if \step*{\li!k_1JOINk_2! ~> \li!k!}, then
      \fmath{\li!k_1! \ple \li!k!} and \fmath{\li!k_2! \ple \li!k!}.
    \item (\textsc{A-R$_{3}$}) the validity of public views is monotone with
      respect to the binary relation \ple: for any values \hastype{\li!v! ::
      \li!T!}, \hastype{\li!k! :: \li!tau!} and \hastype{\li!k'! :: \li!tau!},
      if \fmath{\li!v! \hasview \li!k!} and \fmath{\li!k! \ple \li!k'!}, then
      \fmath{\li!v! \hasview \li!k'!}.
    \item (\textsc{A-R$_{4}$}) the reshape function produces equivalent value,
      as long as the new public view is valid: for any values \hastype{\li!v! ::
      \li!T!}, \hastype{\li!k! :: \li!tau!}, \hastype{\li!k'! :: \li!tau!},
      \hastype{\li!@v! :: \li!@T k!} and \hastype{\li!@v'! :: \li!@T k'!}, if
      \step*{\li!r k @v! ~> \li!v!} and \fmath{\li!v! \hasview \li!k'!} and
      \step*{\li!RESHAPE k k' @v! ~> \li!@v'!}, then \step*{\li!r k' @v'! ~>
      \li!v!}.
  \end{itemize}
\end{definition}

\Cref{fig:obliv-list-struct} defines the join and reshape functions
\lstinline!@list<=#join! and \lstinline!@list<=#reshape!. The partial
order for this join structure is simply the total order on integers,
and the join is simply the maximum of the two numbers. Not all OADTs
have a sensible join-structure: oblivious lists using their exact
length as a public view cannot be combined if they have different
lengths. If such lists are the branches of an oblivious conditional,
lifting will either fail or coerce both to an OADT with a
join-structure.

\begin{figure}[t]
\footnotesize
\jbox{\mergeable{\li!theta! |> \li!@ite!}}
\begin{mathpar}
  \inferrule{
    \li!theta! \in \set{\li!unit!, \li!@bool!}
  }{
    \mergeable{\li!theta! |> \li!?@b x y=>mux @b x y!}
  }

  \inferrule{
    \mergeable{\li!theta_1! |> \li!@ite_1!} \\
    \mergeable{\li!theta_2! |> \li!@ite_2!}
  }{
    \mergeable{\li!theta_1*theta_2! |>%
      \li!?@b x y=>(@ite_1 @b (pi_1 x) (pi_1 y),@ite_2 @b (pi_2 x) (pi_2 y))!}
  }

  \inferrule{
    \mergeable{\li!theta_2! |> \li!@ite_2!}
  }{
    \mergeable{\li!theta_1->theta_2! |> \li!?@b x y=>?z=>@ite_2 @b (x z) (y z)!}
  }

  \inferrule{
    (\li!@T!, \li!JOIN!, \li!RESHAPE!) \in \sctxjoin
  }{
    \mergeable{\li!&@T! |>%
      \begin{tabular}{@{}l}
        \li!?@b x y=>let k = pi_1 x JOIN pi_1 y in! \\
        \li!\ \ \ \ \ \ \ \ \&(k,mux @b (RESHAPE (pi_1 x) k (pi_2 x))! \\
        \li!\ \ \ \ \ \ \ \ \ \ \ \ \ \ \ \ \ (RESHAPE (pi_1 y) k (pi_2 y)))\&!
      \end{tabular}
    }
  }
\end{mathpar}
  \caption{Mergeability}
  \label{fig:mergeable}
\end{figure}

Join structures induce a \emph{mergeability} relation, defined in
\Cref{fig:mergeable}, that can be used to decide if a specification
type can be used in oblivious conditionals. We say \lstinline|theta|
is \emph{mergeable} if \mergeable{\li!theta! |> \li!@ite!}, with
witness \lstinline|@ite| of type
\lstinline|@bool->theta->theta->theta|. We will write
\mergeable{\li!theta!} when we do not care about the witness. This
witness can be thought of as a generalized, drop-in replacement of
\lstinline!mux!: we simply translate \lstinline!mux! to the derived
\lstinline!@ite! if the result type is mergeable. The case of
\tpsi-type captures this intuition: we first join the public views,
and reshape all branches to this common public view, before we select
the correct one privately using \lstinline!mux!. This rule looks up the necessary methods from the context of join structures \sctxjoin. Other cases are
straightforward: we simply fall back to \lstinline!mux! for primitive
types, and the derivation for product and function types are done
congruently.

\subsection{Introduction and Elimination Structures}

An ADT is manipulated by its introduction and elimination forms. To successfully
lift a public program using ADTs, we need structures to explain how the
primitive operations of its ADTs are handled in their OADT counterparts. Thus, an
OADT \lstinline|@T| can \emph{optionally} be equipped with an
\emph{introduction-structure} (intro-structure) and an
\emph{elimination-structure} (elim-structure), defined below. These structures
are optional because some programs only consume ADTs, without constructing any
new ADT values (and vice versa): a function that checks membership in a list
only requires an elim-structure on lists, for example. Intuitively, the axioms
of these structures require the introduction and elimination methods of an OADT
to behave like those of the corresponding ADT. This is formalized using a pair
of logical refinement relations on values (\interpV{\li!DOT!}) and expressions
(\interpE{\li!DOT!}); these relations are formally defined in \Cref{sec:logrel}.

\begin{definition}[intro-structure]
  \label{def:intro-struct}

  An intro-structure of an OADT \lstinline|@T| for ADT \lstinline!T!, with
  global definition \lstinline!data T = !$\overline{\li!C eta!}$, consists of a
  set of functions \lstinline|@C_i|, each corresponding to a constructor
  \lstinline|C_i|. The type of \lstinline|@C_i| is \lstinline|theta_i->&@T|,
  where \fmath{\li!&[theta_i]&! = \li!eta_i!} (note that \textsc{DT-ADT}
  guarantees that \lstinline|eta_i| is a simple type). The particular
  \lstinline|theta_i| an intro-structure uses is determined by the author of
  that structure.
  Each \lstinline|@C_i| is required to logically refine the corresponding
  constructor (\textsc{A-I$_1$}): given any values \hastype{\li!v! ::
  \li!&[theta]&!} and \hastype{\li!v'! :: \li!theta!}, if \fmath{(\li!v!,
  \li!v'!) \in \interpV{\li!theta!}}, then \fmath{(\li!C_i v!, \li!@C_i v'!) \in
  \interpE{\li!&@T!}}.
\end{definition}

\begin{definition}[elim-structure]
  \label{def:elim-struct}

  An elim-structure of an OADT \lstinline|@T| for ADT \lstinline!T!, with global
  definition \lstinline!data T = !$\overline{\li!C eta!}$, consists of a family
  of functions \lstinline|@Imatch_alpha|, indexed by the possible return types.
  The type of \lstinline|@Imatch_alpha| is
  \lstinline|&@T->|$\overline{\li!(theta->alpha)!}$\lstinline|->alpha|, where
  \fmath{\li!&[theta_i]&! = \li!eta_i!} for each \lstinline|theta_i| in the
  function arguments corresponding to alternatives.
  Each \lstinline|@Imatch_alpha| is required to logically refine the
  pattern matching expression, specialized with ADT \lstinline|T| and
  return type \lstinline|alpha|. The sole axiom of this structure
  (\textsc{A-E$_1$}) only considers return type \lstinline|alpha|
  being a specification type: given values \hastype{\li!v_i! ::
    \li!eta_i!}, %
  \hastype{\li!&(k,@v)&! :: \li!@T k!}, %
  \hastype{\li!?x=>e_i! :: \li!&[theta_i]&->&[alpha]&!} and %
  \hastype{\li!?x=>e_i'! :: \li!theta_i->alpha!}, %
  if \step*{\li!r k @v! ~> \li!C_i v_i!} and %
  \fmath{(\li!?x=>e_i!, \li!?x=>e_i'!) \in
    \interpV{\li!theta_i->alpha!}} then %
  \fmath{(\li![v_i/x]e_i!,%
    \li!@Imatch &(k,@v)&\ !\overline{\li!(?x=>e')!}) \in
    \interpE{\li!alpha!}}.
\end{definition}

The types of the oblivious introduction and elimination forms in these
structures are only required to be compatible with the public
counterparts. The programmers can choose which specific OADTs to use
according to their desired privacy
policy. \Cref{fig:obliv-list-struct} shows the constructors and
pattern matching functions for \lstinline!@list<=!.

The elim-structure of an OADT consists of a family of destructors, whose return
type \lstinline!alpha! does not necessarily range over all types. For example,
\lstinline!@Imatch_alpha! of \lstinline!@list<=!, \lstinline!@list<=#Imatch! in
\Cref{fig:obliv-list-struct}, requires \lstinline!alpha! to be a mergeable type,
due to the use of \lstinline!@match!, which imposes a restriction similarly to
\lstinline!mux!. Such constraints on \lstinline!alpha! are automatically
inferred and enforced.

\subsection{Coercion Structures}

As discussed in \Cref{sec:overview}, we may need to convert an
oblivious type to another, either due to a mismatch from input to
output, or due to its lack of certain structures. For example,
\lstinline!@list==! does not have join structure, so if the branches of
an oblivious conditional has type \lstinline!&@list==!, they should be
coerced to \lstinline!&@list<=!, when such a coercion is available.

Two compatible OADTs may form a \emph{coercion-structure}, shown below.

\begin{definition}[coercion-structure]
  \label{def:coer-struct}

  A coercion-structure of a pair of compatible OADTs \lstinline|@T|
  and \lstinline|@T'| for \lstinline|T|, with public view type \lstinline|tau|
  and \lstinline|tau'| respectively, consists of a coercion
  function \lstinline|COER| of type \lstinline|&@T->&@T'|. The coercion should
  produce an equivalent value (\textsc{A-C$_1$}): given values \hastype{\li!v!
    :: \li!T!}, \hastype{\li!&(k,@v)&! :: \li!&@T!} and \hastype{\li!&(k',@v')&!
    :: \li!&@T'!}, if \step*{\li!r k @v! ~> \li!v!} and \step*{\li!COER&(k,@v)&!
    ~> \li!&(k',@v')&!}, then \step*{\li!r k' @v'! ~> \li!v!}.
\end{definition}

\begin{figure}[t]
\jbox{\coercible{\li!theta! ~> \li!theta'! |> \li!COER!}}
\footnotesize
\begin{mathpar}
  \inferrule{
  }{
    \coercible{\li!theta! ~> \li!theta! |> \li!?x=>x!}
  }

  \inferrule{
  }{
    \coercible{\li!bool! ~> \li!@bool! |> \li!?x=>@bool\#s x!}
  }

  \inferrule{
    (\li!@T!, \li!T!, \li!s!, \li!r!, \li!VIEW!, \hasview) \in \sctxoadt
  }{
    \coercible{\li!T! ~> \li!&@T! |> \li!?x=>&(VIEW x,s (VIEW x) x)&!}
  }

  \inferrule{
    \hastype{\li!COER! :: \li!&@T->&@T'!} \in \sctxcoer
  }{
    \coercible{\li!&@T! ~> \li!&@T'! |> \li!COER!}
  }

  \inferrule{
    \coercible{\li!theta_1! ~> \li!theta_1'! |> \li!COER_1!} \\
    \coercible{\li!theta_2! ~> \li!theta_2'! |> \li!COER_2!}
  }{
    \coercible{\li!theta_1*theta_2! ~> \li!theta_1'*theta_2'! |>%
      \li!?x=>(COER_1(pi_1 x),COER_2(pi_2 x))!}
  }

  \inferrule{
    \coercible{\li!theta_1'! ~> \li!theta_1! |> \li!COER_1!} \\
    \coercible{\li!theta_2! ~> \li!theta_2'! |> \li!COER_2!}
  }{
    \coercible{\li!theta_1->theta_2! ~> \li!theta_1'->theta_2'! |>%
      \li!?x=>?y=>COER_2(x (COER_1y))!}
  }
\end{mathpar}
  \caption{Coercion}
  \label{fig:coercion}
\end{figure}

This structure only defines the coercion between two \tpsi-types.
\Cref{fig:coercion} generalizes the coercion relation to any (compatible)
specification types. We say \lstinline|theta| is \emph{coercible} to
\lstinline|theta'| if \coercible{\li!theta! ~> \li!theta'! |> \li!COER!}, with
witness \lstinline|COER| of type \lstinline|theta->theta'|. We may write
\coercible{\li!theta! ~> \li!theta'!} when we do not care about the witness. The
rules of this relation are straightforward. The context of coercion structures
\sctxcoer{} and the context of OADT structures \sctxoadt{} are used to look up the
necessary methods in the corresponding rules. The rule for coercing a function
type is contravariant. Note that we can always coerce a public type to an OADT
by running the section function, and the public view can be selected by the view
function in the OADT structure.

\subsection{Declarative Lifting}

\begin{figure}[t]
\footnotesize
\jbox{\liftR{\li!e! :: \li!theta! |> \li!Oe!}}
\begin{mathpar}
  \inferrule[L-Lit]{
  }{
    \liftR{\li!b! :: \li!bool! |> \li!b!}
  }

  \inferrule[L-Var]{
    \hastype{\li!x! :: \li!theta!} \in \tctx
  }{
    \liftR{\li!x! :: \li!theta! |> \li!x!}
  }

  \inferrule[L-Fun]{
    \hastype{\li!x! :: \li!theta! |> \li!Ox!} \in \lctx
  }{
    \liftR{\li!x! :: \li!theta! |> \li!Ox!}
  }

  \inferrule[L-Abs]{
    \liftR{\extctx{\li!x! :: \li!theta_1!} |- \li!e! :: \li!theta_2! |> \li!Oe!}
  }{
    \liftR{\li!?x:&[theta_1]&=>e! :: \li!theta_1->theta_2! |>%
      \li!?x:theta_1=>Oe!}
  }

  \inferrule[L-App]{
    \liftR{\li!e_2! :: \li!theta_1->theta_2! |> \li!Oe_2!} \\
    \liftR{\li!e_1! :: \li!theta_1! |> \li!Oe_1!}
  }{
    \liftR{\li!e_2 e_1! :: \li!theta_2! |> \li!Oe_2 Oe_1!}
  }

  \inferrule[L-Let]{
    \liftR{\li!e_1! :: \li!theta_1! |> \li!Oe_1!} \\
    \liftR{\extctx{\li!x! :: \li!theta_1!} |-%
      \li!e_2! :: \li!theta_2! |> \li!Oe_2!}
  }{
    \liftR{\li!let x = e_1 in e_2! :: \li!theta_2! |> \li!let x = Oe_1 in Oe_2!}
  }

  \inferrule[L-If$_1$]{
    \liftR{\li!e_0! :: \li!bool! |> \li!Oe_0!} \\\\
    \liftR{\li!e_1! :: \li!theta! |> \li!Oe_1!} \\
    \liftR{\li!e_2! :: \li!theta! |> \li!Oe_2!}
  }{
    \liftR{\li!if e_0 then e_1 else e_2! :: \li!theta! |>%
      \li!if Oe_0 then Oe_1 else Oe_2!}
  }

  \inferrule[L-If$_2$]{
    \liftR{\li!e_0! :: \li!@bool! |> \li!Oe_0!} \\
    \mergeable{\li!theta! |> \li!@ite!} \\\\
    \liftR{\li!e_1! :: \li!theta! |> \li!Oe_1!} \\
    \liftR{\li!e_2! :: \li!theta! |> \li!Oe_2!} \\
  }{
    \liftR{\li!if e_0 then e_1 else e_2! :: \li!theta! |>%
      \li!@ite Oe_0 Oe_1 Oe_2!}
  }

  \inferrule[L-Ctor$_1$]{
    \li!data T =\ !\overline{\li!C eta!} \in \gctx \\\\
    \liftR{\li!e! :: \li!eta_i! |> \li!Oe!}
  }{
    \liftR{\li!C_i e! :: \li!T! |> \li!C_i Oe!}
  }

  \inferrule[L-Ctor$_2$]{
    \hastype{\li!@C_i! :: \li!theta_i->&@T!} \in \sctxintro \\\\
    \liftR{\li!e! :: \li!theta_i! |> \li!Oe!}
  }{
    \liftR{\li!C_i e! :: \li!&@T! |> \li!@C_i Oe!}
  }

  \inferrule[L-Match$_1$]{
    \li!data T =\ !\overline{\li!C eta!} \in \gctx \\\\
    \liftR{\li!e_0! :: \li!T! |> \li!Oe_0!} \\
    \forall i.\; \liftR{\extctx{\li!x! :: \li!eta_i!} |-%
      \li!e_i! :: \li!theta'! |> \li!Oe_i!}
  }{
    \liftR{\li!match e_0 with\ !\overline{\li!C x=>e!} :: \li!theta'! |>%
      \li!match Oe_0 with\ !\overline{\li!C x=>Oe!}}
  }

  \inferrule[L-Match$_2$]{
    \hastype{\li!@Imatch! ::%
      \li!&@T->!\overline{\li!(theta->theta')!}\li!->theta'!} \in \sctxelim \\
    \liftR{\li!e_0! :: \li!&@T! |> \li!Oe_0!} \\
    \forall i.\; \liftR{\extctx{\li!x! :: \li!theta_i!} |-%
      \li!e_i! :: \li!theta'! |> \li!Oe_i!}
  }{
    \liftR{\li!match e_0 with\ !\overline{\li!C x=>e!} :: \li!theta'! |>%
      \li!@Imatch Oe_0\ !\overline{\li!(?x:theta=>Oe)!}}
  }

  \inferrule[L-Coerce]{
    \liftR{\li!e! :: \li!theta! |> \li!Oe!} \\
    \coercible{\li!theta! ~> \li!theta'! |> \li!COER!}
  }{
    \liftR{\li!e! :: \li!theta'! |> \li!COEROe!}
  }
\end{mathpar}
  \caption{Selected declarative lifting rules}
  \label{fig:liftingR}
\vspace{-\baselineskip}
\end{figure}

With these \psistructs, we define a declarative lifting relation, which
describes what the lifting procedure is allowed to derive at a high level. This
lifting relation is given by the judgment \liftR{\sctx; \lctx; \gctx; \tctx |-
\li!e! :: \li!theta! |> \li!Oe!}. It is read as the expression \lstinline|e| of
type \lstinline|&[theta]&| is lifted to the expression \lstinline|Oe| of target
type \lstinline|theta|, under various contexts. The \psistruct context \sctx{}
consists of the set of OADT-structures (\sctxoadt{}), join-structures
(\sctxjoin{}), intro-structures (\sctxintro{}), elim-structures (\sctxelim{})
and coercion-structures (\sctxcoer{}), respectively. The global definition
context \gctx{} is the same as the one used in the typing relation. The local
context \tctx{} is also similar to the one in the typing relation, but it keeps
track of the target types of local variables instead of source types. Finally,
the lifting context \lctx{} consists of entries of the form \hastype{\li!x! ::
\li!theta! |> \li!Ox!}, which associates the global function \lstinline|x| of
type \lstinline|&[theta]&| with a generated function \lstinline|Ox| of the
target type \lstinline|theta|. A single global function may have multiple target
types, i.e., multiple private versions, either specified by the users or by the
callsites. For example, \lctx{} may contain \hastype{\li!filter! ::
\lstinline!&@list<= -> @int -> &@list<=! |> \lstinline!@filter_1!} and
\hastype{\li!filter! :: %
\lstinline!&@list== -> @int -> &@list<=! |> \lstinline!@filter_2!}.

\Cref{fig:liftingR} shows a selection of rules of the declarative
lifting relation (the full rules are in appendix). We elide most
contexts as they are fixed, and simply write \liftR{\li!e! ::
  \li!theta! |> \li!Oe!} for brevity. Most rules are simply
congruences and similar to typing rules. \textsc{L-Fun} outsources the
lifting of a function call to the lifting context. \textsc{L-If$_2$}
handles the case when the condition is lifted to an oblivious boolean
by delegating the translation to the mergeability relation. Similarly,
\textsc{L-Ctor$_2$} and \textsc{L-Match$_2$} query the contexts of the
intro-structures and elim-structures, and use the corresponding
instances as the drop-in replacement, when we are constructing or
destructing \tpsi-types. Lastly, \textsc{L-Coerce} coerces an
expression nondeterministically using the coercion relation.

This lifting relation in \Cref{fig:liftingR} only considers one expression. In
practice, the users specify a set of functions and their target types to lift.
The result of our lifting procedure is a lifting context \lctx{} which maps
these functions and target types to the corresponding generated functions, as
well as any other functions and the inferred target types that these functions
depend on. The global context \gctx{} is also extended with the definitions of
the generated functions. To make this more clear, we say a lifting context is
\emph{derivable}, denoted by \lctxderiv{\lctx}, if and only if, for any
\fmath{\hastype{\li!x! :: \li!theta! |> \li!Ox!} \in \lctx}, %
\fmath{\li!fn x:&[theta]& = e! \in \gctx} and %
\fmath{\li!fn Ox:theta = Oe! \in \gctx} for some \lstinline|e| and
\lstinline|Oe|, such that \liftR{\sctx; \lctx; \gctx; \empctx |- \li!e! ::
\li!theta! |> \li!Oe!}. In other words, any definitions of the lifted functions
in \lctx{} can be derived from the lifting relation in \Cref{fig:liftingR}. Note
that the derivation of a function definition is under a lifting context with
possibly an entry of this function itself. This is similar to the role of global
context in type checking, as \taypsi supports mutually recursive functions. The
goal of our algorithm (\Cref{sec:lift}) is then to find such a derivable lifting
context that includes the user-specified liftings.

\subsection{Logical Refinement}
\label{sec:logrel}

The correctness of the lifting procedure is framed as a \emph{logical
  refinement} between expressions of specification types and those of
simple types; this relationship is defined as a step-indexed logical
relation~\cite{ahmed2006}. As is common, this relation is defined via
a pair of set-valued type denotations: a value interpretation
\interpV{\li!theta!} and an expression interpretation
\interpE{\li!theta!}. We say an expression \lstinline|e'| of type
\lstinline|theta| refines \lstinline|e| of type \lstinline|&[theta]&|
(within $n$ steps) if \fmath{(\li!e!, \li!e'!) \in
  \interpE{\li!theta!}}. In other words, \lstinline|e'| preserves the
behavior of \lstinline|e|, in that if \lstinline|e'| terminates at a
value, \lstinline|e| must terminate at an equivalent value. The
equivalence between values is dictated by \interpV{\li!theta!}.

\Cref{fig:log-rel} shows the complete definition of the logical relation. All
pairs in the relations must be closed and well-typed, i.e., their interpretations
have the forms:

{\footnotesize
\setlength{\abovedisplayskip}{0pt}
\setlength{\abovedisplayshortskip}{-.5em}
  \begin{gather*}
    \interpV{\li!theta!} = \Set{(\li!v!, \li!v'!) |%
      \type{\empctx |- \li!v! :: \li!&[theta]&!} \land%
      \type{\empctx |- \li!v'! :: \li!theta!} \land%
      \ldots} \\
    \interpE{\li!theta!} = \Set{(\li!e!, \li!e'!) |%
      \type{\empctx |- \li!e! :: \li!&[theta]&!} \land%
      \type{\empctx |- \li!e'! :: \li!theta!} \land%
      \ldots}
  \end{gather*}
}%
For brevity, we leave this requirement implicit in
\Cref{fig:log-rel}.

\begin{figure}[t]
\footnotesize
\jbox{\interpV{\li!theta!}}
\begin{mathpar}
  \interpV{\li!unit!} =%
  \interpV{\li!bool!} =%
  \interpV{\li!T!} = \Set{(\li!v!, \li!v'!) |%
    0 < n \implies \li!v! = \li!v'!} \and
  \interpV{\li!@bool!} =%
  \Set{(\li!b!, \li![b']!) |%
    0 < n \implies \li!b! = \li!b'!} \and
  \interpV{\li!&@T!} =%
  \Set{(\li!v!, \li!&(k,@v)&!) |%
    0 < n \implies \step*{\li!r k @v! ~> \li!v!}} \and
  \interpV{\li!theta_1*theta_2!} =%
  \Set{(\li!(v_1,v_2)!, \li!(v_1',v_2')!) |%
    (\li!v_1!, \li!v_1'!) \in \interpV{\li!theta_1!} \land%
    (\li!v_2!, \li!v_2'!) \in \interpV{\li!theta_2!}} \and
  \interpV{\li!theta_1->theta_2!} =%
  \Set{(\li!?x:&[theta_1]&=>e!, \li!?x:theta_1=>e'!) |%
    \forall i < n.%
    \forall (\li!v!, \li!v'!) \in \interpV[i]{\li!theta_1!}.%
      (\li![v/x]e!, \li![v'/x]e'!) \in \interpE[i]{\li!theta_2!}}
\end{mathpar}

\jbox{\interpE{\li!theta!}}
\begin{mathpar}
  \interpE{\li!theta!} = \Set{(\li!e!, \li!e'!) |%
    \forall i < n.%
    \forall \li!v'!.\;%
    \step[i]{\li!e'! ~> \li!v'!} \implies%
    \exists \li!v!.\; \step*{\li!e! ~> \li!v!} \land%
    (\li!v!, \li!v'!) \in \interpV[n-i]{\li!theta!}}
\end{mathpar}

\caption{A logical relation for refinement}
\label{fig:log-rel}
\end{figure}

The definitions are mostly standard. The most interesting case is the
value interpretation of \tpsi-type: we say the pair of a public view
and an oblivious value of an OADT is equivalent to a public value of
the corresponding ADT when the oblivious value can be retracted to the
public value. Intuitively, an encrypted value is equivalent to the
value it decrypts to. The base cases of the value interpretation
are also guarded by the condition that we still have steps left,
i.e., greater than $0$. This requirement maintains the pleasant
property that the interpretations \interpV[0]{\li!theta!} and
\interpE[0]{\li!theta!} are total relations on closed values and
expressions, respectively, of type \lstinline|theta|. The proof also
uses a straightforward interpretation of typing context,
\interpG{\tctx}, whose definition is in appendix.

This relation also gives rise to a semantic characterization of the lifting
context. We say a lifting context is $n$-valid, denoted by \lctxvalid[n]{\lctx},
if and only if, for any \fmath{\hastype{\li!x! :: \li!theta! |> \li!Ox!} \in
\lctx}, \fmath{(\li!x!, \li!Ox!) \in \interpE{\li!theta!}}. If
\lctxvalid[n]{\lctx} for any $n$, we say \lctx{} is valid, denoted by
\lctxvalid{\lctx}. The validity is essentially a semantic correctness of
\lctx{}.

\subsection{Metatheory of Lifting}
\label{sec:meta+lift}
The first key property of the lifting relation is well-typedness,
which guarantees the security of translated programs, thanks to
\Cref{thm:obliviousness}.

\begin{theorem}[Regularity of declarative lifting]
  \label{thm:liftingR-reg}
  Suppose\/ \lctx\/ is well-typed and\/ \liftR{\sctx; \lctx; \gctx; \tctx |-
    \li!e! :: \li!theta! |> \li!Oe!}. We have\/ \type{\gctx; \erase{\tctx} |-
    \li!e! :: \li!&[theta]&!} and\/ \type{\gctx; \tctx |- \li!Oe! ::
    \li!theta!}.
\end{theorem}

Our lifting relation ensures that lifted expressions refine source
expressions in fewer than $n$ steps, as long as every lifted program
in \lctx{} is also semantically correct in fewer than $n$ steps.  As
is common in logical relation proofs, this proof requires a more
general theorem about open terms.

\begin{theorem}[Correctness of declarative lifting of closed terms]
  \label{thm:correct-closed-exp}
  Suppose\/ \liftR{\sctx; \lctx; \gctx; \empctx |- \li!e! :: \li!theta! |>
    \li!Oe!} and\/ \lctxvalid[n]{\lctx}. We have \fmath{(\li!e!, \li!Oe!) \in
    \interpE{\li!theta!}}.
\end{theorem}

Finally, \Cref{thm:correct} provides a strong result of the
correctness of our translation. Any lifting context that is derived
using the rules of \Cref{fig:liftingR} is semantically correct. In other
words, if every pair of source program and lifted program in \lctx{}
are in our lifting relation, they also satisfy our refinement
criteria.

\begin{theorem}[Correctness of declarative lifting]
  \label{thm:correct}
  \lctxderiv{\lctx} implies \lctxvalid{\lctx}.
\end{theorem}

Our notion of logical refinement only provides partial correctness
guarantees, as can be seen in the definition of \interpE{\cdot}. As a
result, the lifting relation does not guarantee equi-termination: it
is possible that a lifted program will diverge when the source program
terminates. This can occur when an \lstinline!if! is replaced by a
\lstinline!mux!: since the latter fully executes both branches, this
effectively changes the semantics of a conditional from a lazy
evaluation strategy to an eager strategy. Using a public value to
bound the recursion depth in order to guarantee termination is a
common practice in data-oblivious computation, for the reasons
discussed in \Cref{sec:overview}.  While the public view of an OADT
naturally serves as a measure in many cases, including all of the case
studies and benchmarks in our evaluation, in theory it is possible for
a user to provide a policy to a function that results in a
nonterminating lifted version. In this situation, users must either
specify a different policy, or rewrite the functions to recurse on a
different argument, e.g., a fuel value.

\section{Algorithmic Lifting}
\label{sec:lift}

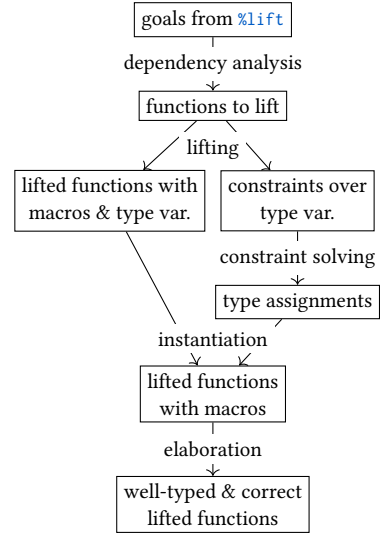
\begin{wrapfigure}{r}{.36\textwidth}
  \vspace{-.3cm}
  \footnotesize
  \tikzstyle{cat} = [text centered, align=center]
  \tikzstyle{block} = [text centered, align=center, rectangle, draw]
  \tikzstyle{label} = [fill=white, align=center, anchor=center]
  \tikzstyle{line} = [->]
  \begin{tikzpicture}[auto, node distance=20pt and 0pt]
    \node[block](g) {goals from \lstinline[basicstyle=\ttfamily\scriptsize]!Mlift!};
    \node[block, below=of g](f) {functions to lift};
    \node[below=27pt of f](p0) {};
    \node[block, left=of p0](l0) {lifted functions with\\macros \& type var.};
    \node[block, right=of p0](c) {constraints over\\type var.};
    \node[block, below=of c](a) {type assignments};
    \node[below=33pt of p0](p1) {};
    \node[block, below=of p1](l1) {lifted functions\\with macros};
    \node[block, below=of l1](l2) {well-typed \& correct\\lifted functions};

    \draw[line] (g) -- node[label] {dependency analysis} (f);
    \draw[line] (f) -- (l0);
    \draw[line] (f) -- (c);
    \path (f) -- node[label,yshift=3pt] {lifting} (p0);
    \draw[line] (c) -- node[label] {constraint solving} (a);
    \draw[line] (a) -- (l1);
    \draw[line] (l0) -- (l1);
    \path (p1) -- node[label] {instantiation} (l1);
    \draw[line] (l1) -- node[label] {elaboration} (l2);
  \end{tikzpicture}
  \caption{Translation pipeline}
  \label{fig:pipeline}
\end{wrapfigure}
\Cref{fig:pipeline} presents the overall workflow of our lifting
algorithm. This algorithm starts with a set of \emph{goals},
i.e., pairs of source functions tagged with the \lstinline!Mlift!
keywords and their desired specification types. We then run our
lifting algorithm on all the functions in these goals, as well as any
functions they depend on, transforming each of these functions to an
oblivious version parameterized by \emph{typed macros} and type
variables, along with a set of constraints over these type variables.
After solving the constraints, we obtain a set of type assignments for
each function. Note that a single function may have multiple type
assignments, one for each occurence in a goal and
callsite. For example, \lstinline!filter!  may have the type
assignment for the goal %
\lstinline!&@list<= -> @int -> &@list<=! generated by \lstinline!Mlift!, and
the assignment for \lstinline!&@list<= -> int -> &@list<=! generated
by the call in \lstinline!filter5! from \Cref{sec:overview}. Finally, we
generate the private versions of all the lifted functions by
instantiating their type variables and expanding away any macros. The
lifting context from the last section is simply these lifted functions
and their generated private versions.

The lifting algorithm is defined using the judgment \liftA{\gctx;
  \tctx |- \li!e! :: \li!eta! ~ \li!X! |> \li!Oe! ~> \cstrs}. It reads
as the source expression \lstinline|e| of type \lstinline|eta| is
lifted to the target expression \lstinline|Oe| whose type is a
\emph{type variable} \lstinline|X| as a placeholder for the
specification type, and generates constraints \cstrs. The source
expression \lstinline|e| is required to be in \emph{administrative
  normal form} (ANF)~\cite{flanagan1993}, which is guaranteed by our
type checker. In particular, type annotations are added to
\lstinline|let|-bindings, and the body of every \lstinline|let| is
either another \lstinline|let| or a variable. Importantly, this means
the last expression of a sequence of \lstinline|let| must be a
variable. The output of this algorithm is an expression \lstinline|Oe|
containing macros (which will be discussed shortly), and the
constraints \cstrs. Unlike the declarative rules, this algorithm keeps
track of the source type \lstinline|eta|, which is used to restrict
the range of the type variables. Consequently, every entry of the
typing context \tctx{} has the form \hastype{\li!x! :: \li!eta! ~
  \li!X!}, meaning that local variable \lstinline|x| has type
\lstinline|eta| in the source program and type \lstinline|X| in the target
program. For example, after the lifting algorithm has processed the
function arguments of \lstinline!filter! in \Cref{fig:list-filter},
the typing context contains entries \hastype{\li!xs! :: \li!list! ~
  \li!X_1!} and \hastype{\li!y! :: \li!int! ~ \li!X_2!}, with freshly
generated type variables \lstinline!X_1! and \lstinline!X_2!.

The typed macros, defined in \Cref{fig:ppx}, are an essential part of
the output of the lifting algorithm, and permit a form of ad-hoc
polymorphism, that allows the algorithm to cleanly separate constraint
solving from program generation. These macros take types as parameters
and elaborate to expressions, under the contexts \sctx, \lctx{} and
\gctx{} implicitly. These macros are effectively thin ``wrappers'' of
their corresponding language constructs and the previously defined
relations. The conditional macro \lstinline|Mite|, for example,
corresponds to the \lstinline|if| expression, but the condition may be
oblivious. The constructor macro \lstinline|MC| is a ``smart''
constructor that may construct a \tpsi-type. The pattern matching
macro \lstinline|Mmatch| is similar to \lstinline|MC| but for
eliminating a type compatible with an ADT. Lastly, \lstinline|MCOER|
and \lstinline|Mx| is simply a direct wrapper of the mergeable
relation and the lifting context \lctx, respectively. Note that the
derivation of these macro are completely determined by the type
parameters.

\Cref{fig:constraints} defines the constraints used in the algorithm,
where \lstinline|thetaV| is the specification types extended with type
variables.  The constraint \incls{\li!X! ~ \li!eta!} means type
variable \lstinline|X| belongs to the compatibility class of
\lstinline|eta|. In other words, \fmath{\li!&[X]&!  = \li!eta!}. Each
macro is accompanied by a constraint on its type parameters. These
constraints mean that the corresponding macros are resolvable. More
formally, this means they can elaborate to some expressions according
to the rules in \Cref{fig:ppx} for any expression arguments. As a
result, after solving all constraints and concretizing the type
variables, all macros in the lifted expression \lstinline|Oe| can be
fully elaborated away.

\begin{figure}[t]
\footnotesize
\jbox{\ppx{\li!Mite(theta_0,theta;e_0,e_1,e_2)! |> \li!Oe!}}
\begin{mathpar}
  \inferrule{
  }{
    \ppx{\li!Mite(bool,theta;e_0,e_1,e_2)! |> \li!if e_0 then e_1 else e_2!}
  }

  \inferrule{
    \mergeable{\li!theta! |> \li!@ite!}
  }{
    \ppx{\li!Mite(@bool,theta;e_0,e_1,e_2)! |> \li!@ite e_0  e_1 e_2!}
  }
\end{mathpar}

\jbox{\ppx{\li!MC(theta,theta';e)! |> \li!Oe!}}
\begin{mathpar}
  \inferrule{
    \li!data T =\ !\overline{\li!C eta!} \in \gctx
  }{
    \ppx{\li!MC_i(eta_i,T;e)! |> \li!C_i e!}
  }

  \inferrule{
    \hastype{\li!@C_i! :: \li!theta_i->&@T!} \in \sctxintro
  }{
    \ppx{\li!MC_i(theta_i,&@T;e)! |> \li!@C_i e!}
  }
\end{mathpar}

\jbox{\ppx{\li!Mmatch(theta_0,!\overline{\li!theta!}\li!,theta';%
    e_0,!\overline{\li!e!}\li!)! |> \li!Oe!}}
\begin{mathpar}
  \inferrule{
    \li!data T =\ !\overline{\li!C eta!} \in \gctx
  }{
    \ppx{\li!Mmatch(T,!\overline{\li!eta!}\li!,theta';%
      e_0,!\overline{\li!e!}\li!)! |>%
      \li!match e_0 with\ !\overline{\li!C x=>e!}}
  }

  \inferrule{
    \hastype{\li!@Imatch! ::%
      \li!&@T->!\overline{\li!(theta->theta')!}\li!->theta'!} \in \sctxelim
  }{
    \ppx{\li!Mmatch(&@T,!\overline{\li!theta!}\li!,theta';%
      e_0,!\overline{\li!e!}\li!)! |>%
      \li!@Imatch e_0\ !\overline{\li!(?x:theta=>e)!}}
  }
\end{mathpar}

\begin{minipage}{.5\textwidth}
\jbox{\ppx{\li!MCOER(theta,theta';e)! |> \li!Oe!}}
\begin{mathpar}
  \inferrule{
    \coercible{\li!theta! ~> \li!theta'! |> \li!COER!}
  }{
    \ppx{\li!MCOER(theta,theta';e)! |> \li!COERe!}
  }
\end{mathpar}
\end{minipage}%
\begin{minipage}{.5\textwidth}
\jbox{\ppx{\li!Mx(theta)! |> \li!Oe!}}
\begin{mathpar}
  \inferrule{
    \hastype{\li!x! :: \li!theta! |> \li!Ox!} \in \lctx
  }{
    \ppx{\li!Mx(theta)! |> \li!Ox!}
  }
\end{mathpar}
\end{minipage}
  \caption{Typed macros}
  \label{fig:ppx}
\end{figure}

\Cref{fig:liftingA} shows a selection of lifting algorithm
rules. Coercions only happen when we lift variables, as in
\textsc{A-Var}. This works because the source program is in ANF, so
each expression is bound to a variable which has the opportunity to
get coerced. For example, the argument to a function or constructor,
in \textsc{A-App} and \textsc{A-Ctor}, is always a variable in ANF,
and recursively lifting it allows the application of
\textsc{A-Var}. On the other hand, the top-level program is always in
\lstinline!let!-binding form, whose last expression is always a
variable too, allowing coercion of the whole program. However, not all
variables are subject to coercions: the function \lstinline!x_2! in
\textsc{A-App}, the condition \lstinline!x_0! in \textsc{A-If} and the
discriminee \lstinline!x_0! in \textsc{A-Match} are kept as they are,
for example. Coercing these variables would be unnecessary and
undesirable. For example, coercing the condition in a conditional only
makes the generated program more expensive: there is no reason to
coerce from \lstinline!bool! to \lstinline!@bool!, and use
\lstinline!mux! instead of \lstinline!if!. Another key invariant we
enforce in our algorithmic rules is that every fresh variable is
``guarded'' by a compatibility class constraint.  For example, in
\textsc{A-Abs}, the freshly generated variables \lstinline!X_1!  and
\lstinline!X_2! belong to the classes \lstinline!eta_1! and
\lstinline!eta_2!, respectively. This constraint ensures that every
type variable can be finitely enumerated, as every compatibility class
is a finite set, bounded by the number of available OADTs. As a
result, constraint solving in our context is decidable. Finally, if an
expression is translated to a macro, a corresponding constraint is
added to ensure this macro is resolvable.

\begin{figure}[t]
\footnotesize
\raggedright
\textsc{Constraints}

\begin{tabular}{RC>{\footnotesize\(}l<{\)}l}
\cstr &\production& \incls{\li!X! ~ \li!eta!} \mid%
  \li!thetaV! = \li!thetaV! \mid \li!Mite(thetaV,thetaV)! \mid%
  \li!MC(thetaV,thetaV)! \mid%
  \li!Mmatch(thetaV,!\overline{\li!thetaV!}\li!,thetaV)! \mid%
  \li!MCOER(thetaV,thetaV)! \mid \li!Mx(thetaV)! & \\
\end{tabular}

  \caption{Constraints}
  \label{fig:constraints}
\end{figure}

\begin{figure}[t]
\footnotesize
\jbox{\liftA{\li!e! :: \li!eta! ~ \li!X! |> \li!Oe! ~> \cstrs}}
\begin{mathpar}
  \inferrule[A-Lit]{
  }{
    \liftA{\li!b! :: \li!bool! ~ \li!X! |> \li!b! ~> \li!X! = \li!bool!}
  }

  \inferrule[A-Var]{
    \hastype{\li!x! :: \li!eta! ~ \li!X!} \in \tctx
  }{
    \liftA{\li!x! :: \li!eta! ~ \li!X'! |> \li!MCOER(X,X';x)! ~> \li!MCOER(X,X')!}
  }

  \inferrule[A-Fun]{
    \li!fn x:eta = e! \in \gctx
  }{
    \liftA{\li!x! :: \li!eta! ~ \li!X! |> \li!Mx(X)! ~> \li!Mx(X)!}
  }

  \inferrule[A-Abs]{
    \fresh{\li!X_1!; \li!X_2!} \\
    \liftA{\extctx{\li!x! :: \li!eta_1! ~ \li!X_1!} |-%
      \li!e! :: \li!eta_2! ~ \li!X_2! |> \li!Oe! ~> \cstrs}
  }{
    \liftA{\li!?x:eta_1=>e! :: \li!eta_1->eta_2! ~ \li!X! |>%
      \li!?x:X_1=>Oe! ~>%
      \incls{\li!X_1! ~ \li!eta_1!}, \incls{\li!X_2! ~ \li!eta_2!},%
      \li!X! = \li!X_1->X_2!, \cstrs}
  }

  \inferrule[A-App]{
    \fresh{\li!X_1!} \\
    \hastype{\li!x_2! :: \li!eta_1->eta_2! ~ \li!X!} \in \tctx \\
    \liftA{\li!x_1! :: \li!eta_1! ~ \li!X_1! |> \li!Oe_1! ~> \cstrs}
  }{
    \liftA{\li!x_2 x_1! :: \li!eta_2! ~ \li!X_2! |> \li!x_2 Oe_1! ~>%
      \incls{\li!X_1! ~ \li!eta_1!}, \li!X! = \li!X_1->X_2!, \cstrs}
  }

  \inferrule[A-Let]{
    \fresh{\li!X_1!} \\
    \liftA{\li!e_1! :: \li!eta_1! ~ \li!X_1! |> \li!Oe_1! ~> \cstrs_1} \\
    \liftA{\extctx{\li!x! :: \li!eta_1! ~ \li!X_1!} |-%
      \li!e_2! :: \li!eta_2! ~ \li!X_2! |> \li!Oe_2! ~> \cstrs_2}
  }{
    \liftA{\li!let x:eta_1 = e_1 in e_2! :: \li!eta_2! ~ \li!X_2! |>%
      \li!let x:X_1 = Oe_1 in Oe_2! ~>%
      \incls{\li!X_1! ~ \li!eta_1!}, \cstrs_1, \cstrs_2}
  }

  \inferrule[A-If]{
    \hastype{\li!x_0! :: \li!bool! ~ \li!X_0!} \in \tctx \\
    \liftA{\li!e_1! :: \li!eta! ~ \li!X! |> \li!Oe_1! ~> \cstrs_1} \\
    \liftA{\li!e_2! :: \li!eta! ~ \li!X! |> \li!Oe_2! ~> \cstrs_2} \\
  }{
    \liftA{\li!if x_0 then e_1 else e_2! :: \li!eta! ~ \li!X! |>%
      \li!Mite(X_0,X;x_0,Oe_1,Oe_2)! ~> \li!Mite(X_0,X)!, \cstrs_1, \cstrs_2}
  }

  \inferrule[A-Ctor]{
    \li!data T =\ !\overline{\li!C eta!} \in \gctx \\
    \fresh{\li!X_i!} \\
    \liftA{\li!x! :: \li!eta_i! ~ \li!X_i! |> \li!Oe! ~> \cstrs}
  }{
    \liftA{\li!C_i x! :: \li!T! ~ \li!X! |> \li!MC_i(X_i,X;Oe)! ~>%
      \incls{\li!X_i! ~ \li!eta_i!}, \li!MC_i(X_i,X)!, \cstrs}
  }

  \inferrule[A-Match]{
    \li!data T =\ !\overline{\li!C eta!} \in \gctx \\
    \fresh{\overline{\li!X!}} \\
    \hastype{\li!x_0! :: \li!T! ~ \li!X_0!} \in \tctx \\
    \forall i.\; \liftA{\extctx{\li!x! :: \li!eta_i! ~ \li!X_i!} |-%
      \li!e_i! :: \li!eta'! ~ \li!X'! |> \li!Oe_i! ~> \cstrs_i}
  }{
    \liftA{\li!match x_0 with\ !\overline{\li!C x=>e!} :: \li!eta'! ~ \li!X'! |>%
      \li!Mmatch(X_0,!\overline{\li!X!}\li!,X';%
      x_0,!\overline{\li!Oe!}\li!)! ~>%
      \overline{\incls{\li!X! ~ \li!eta!}},%
      \li!Mmatch(X_0,!\overline{\li!X!}\li!,X')!,%
      \overline{\cstrs}}
  }
\end{mathpar}

  \caption{Selected algorithmic lifting rules}
  \label{fig:liftingA}
\end{figure}

We use the judgment \cstrsat{\sctx; \lctx; \gctx; \subst |- \cstrs} to
mean the assignment \subst{} satisfies a set of constraints \cstrs{},
under the context of \psistruct, lifting context and global definition
context.  The constraints generated by our lifting algorithm use type
variables \lstinline!X! as placeholders for the target type of the
function being lifted. To solve a goal with a particular target type
\lstinline!theta!, we add a constraint to \cstrs{} that equates the
placeholder with the stipulated type, i.e., \lstinline!X = theta!. Our
constraint solver then attempts to find type assignments that satisfy
the constraints in \cstrs{}; the resulting assignment is used to
generate private versions of all the functions in the set of goals, as
well as the accompanying lifting context.

At a high level,\footnote{The full details of our constraint solver are given in
the appendix.} our solver reduces all constraints, except for function call
constraints (\lstinline!Mx!), to quantifier-free formulas in a finite domain
theory, which can be efficiently solved using an off-the-shelf solver. Function
call constraints are recursively solved once their type arguments have been
concretized by discharging the other constraints. When a function call
constraint is unsatisfiable, we add a new refutation constraint and invoke the
solver again to find a new instantiation of type parameters. As an example of
this process, in order to ascribe \lstinline!filter! the type %
\lstinline!&@list== -> @int -> &@list<=!, we first add the
constraint %
\lstinline!X = &@list== -> @int -> &@list<=! to the constraints generated by the
lifting algorithm \liftA{\empctx |- \li!...! :: \li!list -> int -> list! ~ \li!X!
|> \li!Oe! ~> \cstrs}. Solving the other constraints may concretize the type
variable of function call constraint \lstinline!Mfilter(X)!, i.e., the type of
the recursive call to %
\lstinline!filter!, to %
\lstinline!Mfilter(&@list== -> @int -> &@list<=)!.  Recursively
solving this subgoal assuming the original goal is solved, i.e.,
extending the lifting context with the original goal, results in
immediate success, as the subgoal is simply in the lifting context. On
the other hand, if the type of the recursive call is instantiated as %
\lstinline!Mfilter(&@list== -> @int -> &@list==)!, the same
constraints generated by lifting \lstinline!filter! are solved, with
an additional constraint %
\lstinline!X = &@list== -> @int -> &@list==!. However, this set of
constraints is unsatisfiable, as \lstinline!@list==! has no join
structure, so we add a refutation constraint to the context that
forces the solver to not generate this assignment again. In general,
the type of the recursive call to \lstinline!filter! may be
concretized to any types compatible with %
\lstinline!list -> int -> list!. The number of such compatible types
is bounded, as the number of arguments of this function and the number
of OADTs are themselves bounded. The function \lstinline!filter! has
$3 \times 2 \times 3 = 18$ possible type assignments. In the worst
case scenario, the algorithm eventually terminates after exhausting
all $18$ combinations.

The lifting algorithm enjoys a soundness theorem with respect to the
declarative lifting relation. As a result, our algorithm inherits the
well-typedness and correctness properties of the declarative
version. The statement of this theorem follows how the algorithm is
used: if the generated constraints, equating the function type
variable with the specification type, are satisfiable by the type
assignment \subst{}, instantiating the lifted expression with \subst{}
and elaborating the macros results in a target expression that is
valid under the declarative lifting relation:

\begin{theorem}[Soundness of algorithmic lifting]
  \label{thm:sound-algo}
  Suppose\/ \liftA{\gctx; \empctx |- \li!e! :: \li!eta! ~ \li!X! |> \li!Oe! ~>
    \cstrs}. Given a specification type \li!theta!, if\/ \cstrsat{\sctx; \lctx;
    \gctx; \subst |- \li!X! = \li!theta!, \cstrs}, then\/
  \fmath{\subst(\li!Oe!)} elaborates to an expression \li!Oe'!, such that\/
  \liftR{\sctx; \lctx; \gctx; \empctx |- \li!e! :: \li!theta! |> \li!Oe'!}.
\end{theorem}

The proof of this theorem is available in the appendix.

\section{Implementation and Evaluation}
\label{sec:eval}

Our compilation pipeline takes as input a source program, including any OADTs,
\psistructs, and macros (e.g., \lstinline!Mlift!), in the public fragment of
\taypsi and privacy policies (i.e., security-type signatures) for all target
functions. After typing the source program using a bidirectional type checker,
our lifting pass generates secure versions of the specified functions and their
dependencies, using Z3~\cite{demoura2008} as its constraint solver. The
resulting \taypsi functions are translated into \oil~\cite{ye2023}, an ML-style
functional language equipped with oblivious arrays and secure array operations:
OADTs are converted to serialized versions which are stored in secure arrays,
and all oblivious operations are translated into secure array operations. After
applying some optimizations, our pipeline outputs an OCaml library providing
secure implementations of all the specified functions, including section and
retraction functions for encrypting private data and decrypting the results of a
joint computation. After linking this library to a \emph{driver} that provides
the necessary cryptographic primitives (i.e., secure integer arithmetic),
programmers can build secure MPC applications on top of this API.  The following
evaluation uses a driver implemented using the popular open-source EMP
toolkit~\cite{emp-toolkit}.

\paragraph{Optimizations} Our implementation of \taypsi implements three
optimizations which further improve the performance of the programs it
generates.\footnote{Our appendix describes each of these optimizations in more
detail.} The \emph{reshape guard} optimization instruments reshape instances to
first check if the public views of two private values are identical, omitting
the reshape operation if so. The \emph{memoization} optimization caches the
sizes of the private representation of data in order to avoid recalculating this
information, which is needed to create and slice oblivious arrays. The final,
\emph{smart array} optimization supports zero-cost array slicing and
concatenation, and eliminates redundant operations over the serialized
representation of oblivious data. One observation underlying this optimization
is that evaluating a \lstinline!mux! whose branches are encrypted versions of
publicly-known values is unnecessary: %
\lstinline!mux [b] (@bool#s true) (@bool#s false)!  is equivalent to
\lstinline![b]!, for example. This situation frequently occurs in map-like
functions, where the constructor used in each branch of a function is publicly
known. Under the hood, the serialized encoding of the result of \lstinline!map!
uses a boolean tag to indicate which constructor was used to build it, i.e.,
\lstinline!Nil! or \lstinline!Cons!; this boolean is determined by the tag of
the input list, e.g., \lstinline!mux [tag] [true] [false]!. Of course, the tag
used in each branch is publicly known: \lstinline!map! always returns
\lstinline!Nil! if the input list is empty, and returns a \lstinline!Cons!
otherwise. Thus, we can safely reuse the \lstinline![tag]! of the input list to
label the result of \lstinline!map!, for similar reasons as the previous
example. The smart array optimization exploits this observation by marking when
section functions are applied to public values instead of, for example,
immediately evaluating \lstinline!@bool#s true! to the encrypted value
\lstinline![true]!. Then, when performing a \lstinline!mux!, the smart array
first checks if both branches are ``fake'' private values, safely reducing the
\lstinline!mux! to its private condition if so, without actually performing any
cryptographic operations.

Our evaluation considers the following research questions:
\begin{description}
\item[RQ1] How does the performance of \taypsi's transformation-based
approach compare to the dynamic enforcement strategy of \taype?
\item[RQ2] What is the compilation overhead of \taypsi's translation strategy?
\end{description}

\subsection{Microbenchmark Performance}

\begin{figure}[t]
  \centering
  \footnotesize
  \input{list-tree-bench.tex}

  \caption{Running times for each benchmark in milliseconds. The \taypsi column
    also reports the percentage of running time relative to \taype and
    \taype-SA. A \textcolor{red}{\bf failed} entry indicates the benchmark
    either timed out after $5$ minutes or exceeded the memory bound of
    \qty{8}{\giga\byte}. List and tree benchmarks appear above and below the
    double line, respectively.}
  \label{fig:bench}
\end{figure}

To answer \textbf{RQ1}, we have evaluated the performance of a set of
microbenchmarks compiled with both \taypsi and \taype. Both
approaches are equipped with optimizations that are unique to their
enforcement strategies: \taypsi's reshape guard optimization is not
applicable to \taype, and \taype features an \textit{early tape}
optimization that does not make sense for \taypsi.\footnote{\taype
  also implements a tupling optimization, but this is analogous to
  \taypsi's memoization optimization.} Our evaluation also includes
a version of \taype that implements \taypsi's smart array optimization
(\taype-SA), in order to provide a comparison of the two approaches at
their full potential.

Our benchmarks are a superset of the benchmarks from~\citet{ye2023}.
\Cref{fig:bench} presents the experimental results.\footnote{All
  results are averaged across $5$ runs, on an M1 MacBook Pro with
  \qty{16}{\giga\byte} memory. All parties run on the same host with local network
  communication.} These experiments fix the public views of private
lists and trees to be their maximum length and maximum depth,
respectively; the suffix of each benchmark name indicates the public
view used. The benchmarks annotated with $\dag$ simply traverse the
data type in order to produce a primitive value, e.g., an integer;
these include membership (\verb|elem|), hamming distance
(\verb|hamming|), minimum euclidean distance (\verb|euclidean|), dot
product (\verb|dot_prod|), secure index look up (\verb|nth|) and
computing the probability of an event given a probability tree diagram
(\verb|prob|). The programs generated by \taype, \taype-SA and \taypsi
all exhibit similar performance on these benchmarks.
The remaining benchmarks all construct structured data values,
i.e., the sort of application on which \taypsi is expected to
shine. In addition to standard list operations, the list benchmarks
include insertion into a sorted list (\verb|insert|) and insertion of
a list of elements into a sorted list (\verb|insert_list|) (both lists
have public view $100$). The tree examples include a filter function
that removes all nodes (including any subtrees) greater than a given
private integer (\verb|filter|), swapping subtrees if the node matches
a private integer (\verb|swap|), computing a subtree reached following
a list of ``going left'' and ``going right'' directions (\verb|path|),
insertion into a binary search tree (\verb|insert|), replacing the
leaves of a tree with a given tree (\verb|bind|), and collecting all
nodes smaller than a private integer into a list (\verb|collect|).

Dynamic policy enforcement either fails to finish within $5$ minutes or exceeds
an \qty{8}{\giga\byte} memory bound on almost half of the last set of
benchmarks, due to the exponential blowup discussed in \Cref{sec:overview}. For
those benchmarks that do finish, \taypsi's enforcement strategy results in a
fraction of the total execution time compared to \taype. Compared to the version
of \taype using smart arrays, \taypsi still performs comparably or better,
although the gap is somewhat narrowed: functions like \lstinline!map! do not
suffer from exponential blowup, so these benchmarks benefit mostly from the
smart array optimization. In summary, these results demonstrate that a static
enforcement strategy performs considerably better than a dynamic one on many
benchmarks, and works roughly as well on the remainder.

\subsection{Impact of Optimization}

\begin{figure}[t]
  \centering
  \footnotesize
  \input{list-tree-opt.tex}

  \caption{Impact of turning off the smart array (No Smart Array), reshape guard
    (No Reshape Guard), and public view memoization (No Memoization)
    optimizations. Each column presents running time in milliseconds and the
    slowdown relative to that of the fully optimized version reported in
    \Cref{fig:bench}.}
  \label{fig:ablation}
\end{figure}

To evaluate the performance impact of \taypsi's three optimizations, we
conducted an ablation study on their effect. The results, shown in
\Cref{fig:ablation}, indicate that our smart array optimization is the most
important, providing up to almost $800$x speedup in the best case. As suggested
by \Cref{fig:bench}, this optimization also helps significantly with the
performance of \taype, although not enough to outweigh the exponential blowup
innate in its dynamic approach. The other optimizations also improve
performance, albeit not as significantly. As our memoization pass caches public
views of arbitrary type, we have also conducted an ablation study for these list
and tree examples using ADT public views instead, e.g., using Peano number to
encode the maximum length of a list. In this study, we observe up to $9$ times
speed up in list examples, with minimal regression in tree examples. The full
results of this study are included in the appendix.

\subsection{Compilation Overhead}

To measure the overhead of \taypsi's use of an external solver to
resolve constraints, we have profiled the compilation of a set of
larger programs drawn from \taype's benchmark suite, plus an
additional secure dating application.\footnote{The full details of the
  additional case study can be found in the appendix.}
The first two benchmark suites (List and Tree) in
\Cref{fig:bench-solver} include all the microbenchmarks from previous
section. The next benchmark, List (stress), consists of the same
microbenchmarks as List with $5$ additional list OADTs. The purpose of
this synthetic suite is to examine the impact of the number of OADTs
on the search space. The remaining benchmarks represent larger, more
realistic applications which demonstrate the expressivity and
usability of \taypsi.

\begin{wrapfigure}{r}{.56\textwidth}
  \vspace{-.3cm}
  \centering
  \footnotesize
  \input{compile-stats.tex}

  \caption{Impact of constraint solving on compilation.}
  \label{fig:bench-solver}
\end{wrapfigure}
The last three columns of \Cref{fig:bench-solver} report the results
of these experiments: total compilation time (Tot), time spent on
constraint solving (Slv) and the number of solver queries (\#Qu). The
group of columns in the middle of the table describes features that
can impact the performance of our constraint-based approach: the
number of functions (\#Fn) being translated, the number of atomic
types (\#Ty), and the total number of atomic types used in function
types (\#At).  For example, the \textsf{List} benchmark features $7$
atomic types: public and oblivious booleans, integers and lists, as
well as an unsigned integer type (i.e. natural numbers). The number of
atomic types in the function
\lstinline!filter : list -> int -> list! is $3$. In the worst case
scenario, our constraint solving algorithm will explore every
combination of types that are compatible with this signature,
resulting in the constraints associated with \lstinline!filter! being
solved $2*2*2 = 8$ times. Exactly how many compatible types the
constraint solving algorithm explores depends on many factors: the
user-specified policies, the complexity of the functions, the calls to
other functions and so on. We chose these $3$ metrics as a coarse
approximation of the solution space. Our results show that the solver
overhead is quite minimal for most benchmarks, and in general solving
time per query is low thanks to our encoding of constraints in an
efficiently decidable logic.

\section{Related Work}

The problem of secure computation was first formally introduced by
\citet{yao1982}, who simultaneously proposed garbled circuits as a
solution.  Subsequently, a number of other solutions have been
proposed~\cite{evans2018,hazay2010}. Solutions categorized as
multiparty computation are usually based on cryptographic protocols,
e.g.,
secret-sharing~\cite{beimel2011,goldreich1987,maurer2006}. Outsourced
computation is another type of secure computation that includes both
cryptographic solutions, e.g., fully homomorphic
encryption~\cite{gentry2009,acar2018}, and solutions based on
virtualization~\cite{barthe2014c,barthe2019a} or secure
processors~\cite{hoekstra2015}.

\taypsi features a \emph{security-type system}~\cite{sabelfeld2003,
zdancewic2002} based on the type system of \loadt. While most security-type
systems tag types with labels classifying the sensitivity of data, our dependent
type system tags kinds instead. The obliviousness guarantee provided by
\Cref{thm:obliviousness} is a form of \emph{noninterference}~\cite{goguen1982}
that generalizes \emph{memory trace obliviousness} (MTO)~\cite{liu2013}. MTO
considers traces of memory accesses, while the traces in
\Cref{thm:obliviousness} include \emph{every} intermediate program state under a
small-step operational semantics. This reflects MPC's stronger threat model, in
which all parties can observe the complete execution of a program, including
each instruction executed. As a consequence, our type system also protects
against timing channels, similar to other \emph{constant-time}
languages~\cite{cauligi2019}.

Numerous high-level programming languages for writing secure
multiparty computation applications have been
proposed~\cite{hastings2019}. Most prior languages either do not
support structured data, or require all structural information to be
public, e.g., Obliv-C~\cite{zahur2015} and ObliVM~\cite{liu2015}. To
the best of our knowledge \taype is the only existing language for MPC applications that natively supports decoupling
privacy policies from program logic. On the other hand, there are many
aspects of MPC tackled by prior languages that we do not consider
here. Wysteria and Wys$^\ast$~\cite{rastogi2014, rastogi2019}, for
example, focus on \emph{mixed-mode computation} which allows certain
computation to be executed locally. Symphony~\cite{sweet2023}, a
successor of Wysteria, supports first-class shares and first-class
party sets for coordinating many parties, enabling more reactive
applications.  \citet{darais2020} developed $\lambda_{\text{obliv}}$,
a probabilistic functional language for oblivious computations that
can be used to safely implement a variety of cryptography algorithms,
including Oblivious RAM.

Several prior works have considered how to compile secure programs into more
efficient secure versions. Viaduct~\cite{acay2021,acay2024} is a compiler that
transforms high-level programs into secure distributed versions by intelligently
selecting an efficient combination of protocols for subcomputations. The HyCC
toolchain~\cite{buscher2018} similarly transforms a C program into a version
that combines different MPC protocols to optimize performance. The HACCLE
toolchain~\cite{bao2021} uses staging to generate efficient garbled circuits
from a high-level language. Compiler techniques, e.g., vectorization, have been
studied for optimizing fully homomorphic encryption (FHE)
applications~\cite{dathathri2020,cowan2021,malik2021,malik2023,viand2023}.

Jeeves~\cite{yang2012} and \taypsi have a shared goal of decoupling
security policies from program logic. While they both employ a similar
high-level strategy of relying on the language to automatically
enforce policies, their different settings result in very different
solutions. In Jeeves' programming model, each piece of data is
equipped with a pair of high- and low-level views: a username, for
example, may have a high confidentiality view of ``Alice'', but a low
view of ``Anonymous''. The language then uses the view stipulated by
the privacy policy and current execution context, ensuring that
information is only visible to observers with the proper authority.
In the MPC setting, however, no party is allowed to observe the
private data of other parties. Thus, no party can view all the data
necessary for the computation, making it impossible to compute a
correct result by simply replacing data with some predetermined value,
like ``Anonymous''.

\section{Conclusion}
Secure multiparty computation allows joint computation over private
data from multiple parties, while keeping that data secure. Previous
work has considered how to make languages for MPC more accessible by
allowing privacy requirements to be decoupled from functionality,
relying on dynamic enforcement of polices. Unfortunately, the
resulting overhead of this strategy made it difficult to scale
applications manipulating structured data. This work presents \taypsi,
a policy-agnostic language for oblivious computation that transforms
programs to instead \emph{statically} enforce a user-provided privacy
policy. The resulting programs are guaranteed to be both well-typed,
and hence secure, and equivalent to the source program. Our
experimental results show this strategy yields considerable
performance improvements over prior approaches, while maintaining a
clean separation between privacy and programmatic concerns.

%% file: list-tree-bench.tex
\begin{tabular}{llll}
Benchmark & \taype (ms) & \taype-SA (ms) & \taypsi (ms) \\
\midrule
\verb|elem_1000|$^\dag$ & 8.15 & 8.11 & 8.02 \hfill(98.47\%, 98.89\%) \\
\verb|hamming_1000|$^\dag$ & 15.09 & 15.21 & 14.46 \hfill(95.79\%, 95.04\%) \\
\verb|euclidean_1000|$^\dag$ & 67.43 & 67.55 & 67.32 \hfill(99.84\%, 99.66\%) \\
\verb|dot_prod_1000|$^\dag$ & 66.12 & 66.19 & 66.41 \hfill(100.43\%, 100.33\%) \\
\verb|nth_1000|$^\dag$ & 11.98 & 12.05 & 12.04 \hfill(100.54\%, 99.93\%) \\
\verb|map_1000| & 2139.55 & 5.07 & 5.14 \hfill(0.24\%, 101.44\%) \\
\verb|filter_200| & \textcolor{red}{\bf failed} & \textcolor{red}{\bf failed} & 86.86 \hfill(N/A, N/A) \\
\verb|insert_200| & 5796.69 & 88.92 & 88.07 \hfill(1.52\%, 99.04\%) \\
\verb|insert_list_100| & \textcolor{red}{\bf failed} & \textcolor{red}{\bf failed} & 4667.66 \hfill(N/A, N/A) \\
\verb|append_100| & 4274.7 & 45.09 & 44.18 \hfill(1.03\%, 97.99\%) \\
\verb|take_200| & 169.07 & 3.05 & 3.09 \hfill(1.83\%, 101.15\%) \\
\verb|flat_map_200| & \textcolor{red}{\bf failed} & \textcolor{red}{\bf failed} & 7.3 \hfill(N/A, N/A) \\
\verb|span_200| & 13529.34 & 124.79 & 91.22 \hfill(0.67\%, 73.09\%) \\
\verb|partition_200| & \textcolor{red}{\bf failed} & \textcolor{red}{\bf failed} & 176.49 \hfill(N/A, N/A) \\
\midrule
\midrule
\verb|elem_16|$^\dag$ & 446.81 & 459.1 & 404.9 \hfill(90.62\%, 88.19\%) \\
\verb|prob_16|$^\dag$ & 13082.52 & 12761.7 & 12735.16 \hfill(97.34\%, 99.79\%) \\
\verb|map_16| & 4414.69 & 262.14 & 215.67 \hfill(4.89\%, 82.27\%) \\
\verb|filter_16| & 8644.14 & 452.04 & 433.7 \hfill(5.02\%, 95.94\%) \\
\verb|swap_16| & \textcolor{red}{\bf failed} & \textcolor{red}{\bf failed} & 4251.36 \hfill(N/A, N/A) \\
\verb|path_16| & \textcolor{red}{\bf failed} & 6657.07 & 894.88 \hfill(N/A, 13.44\%) \\
\verb|insert_16| & 83135.81 & 8093.81 & 1438.87 \hfill(1.73\%, 17.78\%) \\
\verb|bind_8| & 21885.65 & 494.98 & 532.86 \hfill(2.43\%, 107.65\%) \\
\verb|collect_8| & \textcolor{red}{\bf failed} & \textcolor{red}{\bf failed} & 143.38 \hfill(N/A, N/A) \\
\end{tabular}

%% file: list-tree-opt.tex
\begin{tabular}{llll}
Benchmark & No Smart Array (ms) & No Reshape Guard (ms) & No Memoization (ms) \\
\midrule
\verb|elem_1000| & 18.37 \hfill(2.29x) & 8.06 \hfill(1.0x) & 17.76 \hfill(2.21x) \\
\verb|hamming_1000| & 51.73 \hfill(3.58x) & 14.53 \hfill(1.01x) & 35.5 \hfill(2.46x) \\
\verb|euclidean_1000| & 79.07 \hfill(1.17x) & 67.31 \hfill(1.0x) & 76.36 \hfill(1.13x) \\
\verb|dot_prod_1000| & 87.77 \hfill(1.32x) & 66.15 \hfill(1.0x) & 77.33 \hfill(1.16x) \\
\verb|nth_1000| & 22.69 \hfill(1.88x) & 12.18 \hfill(1.01x) & 20.53 \hfill(1.7x) \\
\verb|map_1000| & 2106.43 \hfill(409.89x) & 139.91 \hfill(27.23x) & 37.71 \hfill(7.34x) \\
\verb|filter_200| & 5757.28 \hfill(66.29x) & 93.93 \hfill(1.08x) & 114.7 \hfill(1.32x) \\
\verb|insert_200| & 255.43 \hfill(2.9x) & 94.61 \hfill(1.07x) & 89.32 \hfill(1.01x) \\
\verb|insert_list_100| & 22806.87 \hfill(4.89x) & 5186.07 \hfill(1.11x) & 4771.28 \hfill(1.02x) \\
\verb|append_100| & 4226.32 \hfill(95.66x) & 50.79 \hfill(1.15x) & 61.77 \hfill(1.4x) \\
\verb|take_200| & 169.45 \hfill(54.91x) & 12.92 \hfill(4.19x) & 4.68 \hfill(1.52x) \\
\verb|flat_map_200| & 5762.63 \hfill(789.08x) & 16.99 \hfill(2.33x) & 60.03 \hfill(8.22x) \\
\verb|span_200| & 5924.1 \hfill(64.95x) & 99.83 \hfill(1.09x) & 120.09 \hfill(1.32x) \\
\verb|partition_200| & 11528.0 \hfill(65.32x) & 185.16 \hfill(1.05x) & 231.06 \hfill(1.31x) \\
\midrule
\midrule
\verb|elem_16| & 433.73 \hfill(1.07x) & 404.05 \hfill(1.0x) & 402.15 \hfill(0.99x) \\
\verb|prob_16| & 13019.56 \hfill(1.02x) & 12746.24 \hfill(1.0x) & 12731.89 \hfill(1.0x) \\
\verb|map_16| & 4410.84 \hfill(20.45x) & 635.18 \hfill(2.95x) & 213.96 \hfill(0.99x) \\
\verb|filter_16| & 8674.71 \hfill(20.0x) & 1131.02 \hfill(2.61x) & 440.16 \hfill(1.01x) \\
\verb|swap_16| & 8671.52 \hfill(2.04x) & 5471.4 \hfill(1.29x) & 4246.39 \hfill(1.0x) \\
\verb|path_16| & 9108.54 \hfill(10.18x) & 1083.21 \hfill(1.21x) & 888.95 \hfill(0.99x) \\
\verb|insert_16| & 19101.36 \hfill(13.28x) & 2151.83 \hfill(1.5x) & 1432.92 \hfill(1.0x) \\
\verb|bind_8| & 19647.83 \hfill(36.87x) & 870.93 \hfill(1.63x) & 534.3 \hfill(1.0x) \\
\verb|collect_8| & 11830.6 \hfill(82.51x) & 152.29 \hfill(1.06x) & 186.92 \hfill(1.3x) \\
\end{tabular}

%% file: compile-stats.tex
\begin{tabular}{l|rrr|rrr}
Suite & \#Fn & \#Ty & \#At & \#Qu & Tot (s) & Slv (s) \\
\midrule
List & 20 & 7 & 70 & 84 & 0.47 & 0.081 \\
Tree & 14 & 9 & 44 & 31 & 0.47 & 0.024 \\
List (stress) & 20 & 12 & 70 & 295 & 3.45 & 2.8 \\
Dating & 4 & 13 & 16 & 10 & 0.58 & 0.019 \\
Medical Records & 20 & 19 & 58 & 51 & 0.48 & 0.072 \\
Secure Calculator & 2 & 9 & 6 & 5 & 1.34 & 0.013 \\
Decision Tree & 2 & 13 & 6 & 16 & 0.28 & 0.016 \\
K-means & 16 & 11 & 68 & 86 & 1.62 & 0.95 \\
Miscellaneous & 11 & 7 & 42 & 47 & 0.26 & 0.065 \\
\end{tabular}

%% file: appendix.tex
\section{Syntax}

\Cref{ap:fig:syntax} defines the syntax of \loadtpsi.

By convention, We use \lstinline!x!  for variable names, \lstinline!C! for
constructor names, \lstinline!T! for type names, and \lstinline!@T! for
oblivious type names. We use a boolean subscript to indicate if it is the left
injection (projection) or right injection (projection). We may write the more
conventional \lstinline|@inl| (\lstinline|pi_1|) and \lstinline|@inr|
(\lstinline|pi_2|), instead of \lstinline|@inj_true| (\lstinline|pi_true|) and
\lstinline|@inj_false| (\lstinline|pi_false|).

\paragraph{Remark} Our Coq formalization distinguishes between public product
type and oblivious product type. Here we reuse a single product type for both,
to simplify the syntax for better presentation. A product type (pair) is
considered oblivious product (oblivious pair) if both components are oblivious.

\begin{figure}[h]
\raggedright

\textsc{Expressions}

\begin{tabular}{RCLl}
\lstinline|e, tau| &\production&& \\
&\mid& \lstinline!unit! \mid \lstinline!bool! \mid \lstinline!@bool! \mid%
  \lstinline!tau*tau! \mid \lstinline!tau@+tau! & simple types \\
&\mid& \lstinline!Pix:tau,tau! & dependent function type \\
&\mid& \lstinline!&@T! & \tpsi-type \\
&\mid& \lstinline!x! \mid \lstinline!T! & variable and type names \\
&\mid& \lstinline!()! \mid \lstinline!b! & unit and boolean constants \\
&\mid& \lstinline!?x:tau=>e! & function abstraction \\
&\mid& \lstinline!let x = e in e! & let binding \\
&\mid& \lstinline!e e! \mid \lstinline|C e| \mid \lstinline|@T e|%
  & function, constructor and type application \\
&\mid& \lstinline!if e then e else e! & conditional \\
&\mid& \lstinline!mux e e e! & oblivious conditional \\
&\mid& \lstinline!(e,e)! & pair \\
&\mid& \lstinline!&(e,e)&! & dependent pair (\tpsi-pair) \\
&\mid& \lstinline!pi_b e! & product and \tpsi-type projection \\
&\mid& \lstinline!@inj_b<tau> e! & oblivious sum injection \\
&\mid& \lstinline!@match e with x=>e|x=>e! & oblivious sum elimination \\
&\mid& \lstinline!match e with !\overline{\lstinline!C x=>e!} & ADT elimination \\
&\mid& \lstinline!@bool#s e! & boolean section \\
&\mid& \lstinline![b]! \mid \lstinline![inj_b<@w> @v]! & runtime boxed values \\
\end{tabular}

\textsc{Global Definitions}

\begin{tabular}{RCLl}
\lstinline|D| &\production& & \\
&\mid& \lstinline!data T = !\overline{\lstinline!C tau!} &%
  algebraic data type definition \\
&\mid& \lstinline!fn x:tau = e! & (recursive) function definition \\
&\mid& \lstinline!obliv @T (x:tau) = tau! & (recursive) oblivious type definition \\
\end{tabular}

\textsc{Oblivious Type Values}

\begin{tabular}{RCLl}
\lstinline|@w| &\production& \lstinline!unit! \mid \lstinline!@bool! \mid%
  \lstinline!@w*@w! \mid \lstinline!@w@+@w! & \\
\end{tabular}

\textsc{Oblivious Values}

\begin{tabular}{RCLl}
\lstinline|@v| &\production& \lstinline!()! \mid \lstinline![b]! \mid%
  \lstinline!(@v,@v)! \mid \lstinline![inj_b<@w> @v]! & \\
\end{tabular}

\textsc{Values}

\begin{tabular}{RCLl}
\lstinline|v| &\production& \lstinline|@v| \mid \lstinline!b! \mid%
  \lstinline!(v,v)! \mid \lstinline!&(v,v)&! \mid%
  \lstinline!?x:tau=>e! \mid \lstinline!C v! \\
\end{tabular}

\caption{\loadtpsi syntax}
\label{ap:fig:syntax}
\end{figure}

\FloatBarrier

\section{Semantics}

\Cref{ap:fig:ectx} and \Cref{ap:fig:ovalty} present the auxiliary definitions,
evaluation context and oblivious value synthesis. \Cref{ap:fig:semantics}
defines the small-step operational semantics of \loadtpsi. The judgment
\step{\gctx |- \li!e! ~> \li!e'!} means \lstinline|e| steps to \lstinline|e'|
under the global context \gctx. We often abbreviate this judgment as
\step{\li!e! ~> \li!e'!}, as the global context is fixed.

\begin{figure}[h]
\raggedright

\textsc{Evaluation Context}

\begin{tabular}{RCLl}
\lstinline|ectx| &\production&& \\
&\mid& \lstinline![]*tau! \mid \lstinline!@w*[]! \mid%
  \lstinline![]@+tau! \mid \lstinline!@w@+[]! & \\
&\mid& \lstinline!let x = [] in e! \mid%
  \lstinline!e []! \mid \lstinline![] v! \mid%
  \lstinline!C []! \mid \lstinline!@T []! & \\
&\mid& \lstinline!if [] then e else e! \mid%
  \lstinline!mux [] e e! \mid%
  \lstinline!mux v [] e! \mid \lstinline!mux v v []! & \\
&\mid& \lstinline!([],e)! \mid \lstinline!(v,[])! \mid%
  \lstinline!&([],e)&! \mid \lstinline!&(v,[])&! \mid \lstinline!pi_b []! & \\
&\mid& \lstinline!@inj_b<[]> e! \mid \lstinline!@inj_b<@w> []! \mid%
  \lstinline!@match [] with x=>e|x=>e! & \\
&\mid& \lstinline!match [] with !\overline{\lstinline!C x => e!} \mid%
  \lstinline!@bool#s []! & \\
\end{tabular}

  \caption{\loadtpsi evaluation context}
  \label{ap:fig:ectx}
\end{figure}

\begin{figure}[h]
\footnotesize
\jbox{\ovalty{\li!@v! <~ \li!@w!}}
\begin{mathpar}
  \inferrule[OT-Unit]{
  }{
    \ovalty{\li!()! <~ \li!unit!}
  }

  \inferrule[OT-OBool]{
  }{
    \ovalty{\li![b]! <~ \li!@bool!}
  }

  \inferrule[OT-Prod]{
    \ovalty{\li!@v_1! <~ \li!@w_1!} \\
    \ovalty{\li!@v_2! <~ \li!@w_2!}
  }{
    \ovalty{\li!(@v_1,@v_2)! <~ \li!@w_1*@w_2!}
  }

  \inferrule[OT-OSum]{
    \ovalty{\li!@v! <~ \li!ite(b,@w_1,@w_2)!}
  }{
    \ovalty{\li![inj_b<@w_1@+@w_2> @v]! <~ \li!@w_1@+@w_2!}
  }
\end{mathpar}
  \caption{\loadtpsi oblivious value synthesis}
  \label{ap:fig:ovalty}
\end{figure}

\begin{figure}[h]
\footnotesize
\jbox{\step{\li!e! ~> \li!e'!}}
\begin{mathpar}
  \inferrule[S-Ctx]{
    \step{\li!e! ~> \li!e'!}
  }{
    \step{\li!ectx[e]! ~> \li!ectx[e']!}
  }

  \inferrule[S-Fun]{
    \li!fn x:tau = e! \in \gctx
  }{
    \step{\li!x! ~> \li!e!}
  }

  \inferrule[S-OADT]{
    \li!obliv @T (x:tau) = tau'! \in \gctx
  }{
    \step{\li!@T v! ~> \li![v/x]tau'!}
  }

  \inferrule[S-App]{
  }{
    \step{\li!(?x:tau=>e) v! ~> \li![v/x]e!}
  }

  \inferrule[S-Let]{
  }{
    \step{\li!let x = v in e! ~> \li![v/x]e!}
  }

  \inferrule[S-If]{
  }{
    \step{\li!if b then e_1 else e_2! ~> \li!ite(b,e_1,e_2)!}
  }

  \inferrule[S-Mux]{
  }{
    \step{\li!mux [b] v_1 v_2! ~> \li!ite(b,v_1,v_2)!}
  }

  \inferrule[S-Match]{
  }{
    \step{\li!match C_i v with\ !\overline{\li!C x=>e!} ~> \li![v/x]e_i!}
  }

  \inferrule[S-Proj]{
  }{
    \step{\li!pi_b (v_1,v_2)! ~> \li!ite(b,v_1,v_2)!}
  }

  \inferrule[S-PsiProj]{
  }{
    \step{\li!pi_b &(v_1,v_2)&! ~> \li!ite(b,v_1,v_2)!}
  }

  \inferrule[S-Sec]{
  }{
    \step{\li!@bool\#s b! ~> \li![b]!}
  }

  \inferrule[S-OInj]{
  }{
    \step{\li!@inj_b<@w> @v! ~> \li![inj_b<@w> @v]!}
  }

  \inferrule[S-OMatch]{
    \ovalty{\li!@v_1! <~ \li!@w_1!} \\
    \ovalty{\li!@v_2! <~ \li!@w_2!}
  }{
    \step{\li!@match [inj_b<@w_1@+@w_2> @v] with x=>e_1|x=>e_2! ~>%
      \begin{tabular}{@{}l}
        \li!mux [b] ite(b,[@v/x]e_1,[@v_1/x]e_1)! \\
        \li!\ \ \ \ \ \ \ \ ite(b,[@v_2/x]e_2,[@v/x]e_2)!
      \end{tabular}
    }
  }
\end{mathpar}
  \caption{\loadtpsi operational semantics}
  \label{ap:fig:semantics}
\end{figure}

\FloatBarrier

\section{Typing}

\Cref{ap:fig:kind} defines \loadtpsi kinds and the semi-lattice that they form.
The type system of \loadtpsi consists of a pair of kinding and typing judgments,
\kind{\gctx; \tctx |- \li!tau! :: \li!kappa!} and \type{\gctx; \tctx |- \li!e!
:: \li!tau!}, respectively. \Cref{ap:fig:kinding} and \Cref{ap:fig:typing}
define these rules. We elide \gctx{} from the rules, as they both assume a fixed
global context. Some kinding side conditions in the typing rules are elided for
brevity; the complete rules can be found in the Coq development.

The rule \textsc{T-Conv} in \Cref{ap:fig:typing} relies on a type equivalence
relation \typequiv{\gctx |- \li!tau! == \li!tau'!}(or simply
\typequiv{\li!tau! == \li!tau'!}). This relation is defined in terms of parallel
reduction \pared{\gctx |- \li!tau! ~> \li!tau'!} which is elided as its
definition is mostly standard. Interested readers may consult our Coq
formalization.

\Cref{ap:fig:gctx-typing} gives the rules for typing global definitions.

\begin{figure}[h]
\raggedright

\begin{minipage}{.5\textwidth}
\textsc{Kinds}

\begin{tabular}{rCLl}
\lstinline|kappa| &\production&& \\
&\mid& \lstinline!*@A! & Any \\
&\mid& \lstinline!*@P! & Public \\
&\mid& \lstinline!*@O! & Oblivious \\
&\mid& \lstinline!*@M! & Mixed \\
\end{tabular}
\end{minipage}%
\begin{minipage}{.5\textwidth}
\begin{tikzpicture}[->,-stealth]
  \node (M) {M};
  \node (P) [below left of=M] {P};
  \node (O) [below right of=M] {O};
  \node (A) [below right of=P] {A};

  \draw (A) edge (P)
        (A) edge (O)
        (P) edge (M)
        (O) edge (M);
\end{tikzpicture}
\end{minipage}
  \caption{\loadtpsi kinds}
  \label{ap:fig:kind}
\end{figure}

\begin{figure}[h]
\footnotesize
\jbox{\kind{\li!tau! :: \li!kappa!}}
\begin{mathpar}
  \inferrule[K-Sub]{
    \kind{\li!tau! :: \li!kappa!} \\
    \li!kappa! \sqsubseteq \li!kappa'!
  }{
    \kind{\li!tau! :: \li!kappa'!}
  }

  \inferrule[K-Unit]{
  }{
    \kind{\li!unit! :: \li!*@A!}
  }

  \inferrule[K-Bool]{
  }{
    \kind{\li!bool! :: \li!*@P!}
  }

  \inferrule[K-OBool]{
  }{
    \kind{\li!@bool! :: \li!*@O!}
  }

  \inferrule[K-ADT]{
    \li!data T =\ !\overline{\li!C tau!} \in \gctx
  }{
    \kind{\li!T! :: \li!*@P!}
  }

  \inferrule[K-OADT]{
    \li!obliv @T (x:tau) = tau'! \in \gctx \\
    \type{\li!e! :: \li!tau!}
  }{
    \kind{\li!@T e! :: \li!*@O!}
  }

  \inferrule[K-Pi]{
    \kind{\li!tau_1! :: \li!*@*!} \\
    \kind{\extctx{\li!x! :: \li!tau_1!} |- \li!tau_2! :: \li!*@*!}
  }{
    \kind{\li!Pix:tau_1,tau_2! :: \li!*@M!}
  }

  \inferrule[K-Prod]{
    \kind{\li!tau_1! :: \li!kappa!} \\
    \kind{\li!tau_2! :: \li!kappa!}
  }{
    \kind{\li!tau_1*tau_2! :: \li!kappa!}
  }

  \inferrule[K-OSum]{
    \kind{\li!tau_1! :: \li!*@O!} \\
    \kind{\li!tau_2! :: \li!*@O!}
  }{
    \kind{\li!tau_1@+tau_2! :: \li!*@O!}
  }

  \inferrule[K-Psi]{
    \li!obliv @T (x:tau) = tau'! \in \gctx
  }{
    \kind{\li!&@T! :: \li!*@M!}
  }

  \inferrule[K-Let]{
    \type{\li!e! :: \li!tau!} \\
    \kind{\extctx{\li!x! :: \li!tau!} |- \li!tau'! :: \li!*@O!}
  }{
    \kind{\li!let x = e in tau'! :: \li!*@O!}
  }

  \inferrule[K-If]{
    \type{\li!e_0! :: \li!bool!} \\\\
    \kind{\li!tau_1! :: \li!*@O!} \\
    \kind{\li!tau_2! :: \li!*@O!}
  }{
    \kind{\li!if e_0 then tau_1 else tau_2! :: \li!*@O!}
  }

  \inferrule[K-Match]{
    \li!data T =\ !\overline{\li!C tau!} \in \gctx \\
    \type{\li!e_0! :: \li!T!} \\\\
    \forall i.\; \kind{\extctx{\li!x! :: \li!tau_i!} |- \li!tau_i'! :: \li!*@O!}
  }{
    \kind{\li!match e_0 with\ !\overline{\li!C x=>tau'!} :: \li!*@O!}
  }
\end{mathpar}
  \caption{\loadtpsi kinding rules}
  \label{ap:fig:kinding}
\end{figure}

\begin{figure}[h]
\footnotesize
\jbox{\type{\li!e! :: \li!tau!}}
\begin{mathpar}
  \inferrule[T-Conv]{
    \type{\li!e! :: \li!tau!} \\
    \typequiv{\li!tau! == \li!tau'!} \\
    \kind{\li!tau'! :: \li!*@*!}
  }{
    \type{\li!e! :: \li!tau'!}
  }

  \inferrule[T-Unit]{
  }{
    \type{\li!()! :: \li!unit!}
  }

  \inferrule[T-Lit]{
  }{
    \type{\li!b! :: \li!bool!}
  }

  \inferrule[T-Var]{
    \hastype{\li!x! :: \li!tau!} \in \tctx
  }{
    \type{\li!x! :: \li!tau!}
  }

  \inferrule[T-Fun]{
    \li!fn x:tau = e! \in \gctx
  }{
    \type{\li!x! :: \li!tau!}
  }

  \inferrule[T-Abs]{
    \type{\extctx{\li!x! :: \li!tau_1!} |- \li!e! :: \li!tau_2!} \\
    \kind{\li!tau_1! :: \li!*@*!}
  }{
    \type{\li!?x:tau_1=>e! :: \li!Pix:tau_1,tau_2!}
  }

  \inferrule[T-App]{
    \type{\li!e_2! :: \li!Pix:tau_1,tau_2!} \\
    \type{\li!e_1! :: \li!tau_1!}
  }{
    \type{\li!e_2 e_1! :: \li![e_1/x]tau_2!}
  }

  \inferrule[T-Let]{
    \type{\li!e_1! :: \li!tau_1!} \\
    \type{\extctx{\li!x! :: \li!tau_1!} |- \li!e_2! :: \li!tau_2!}
  }{
    \type{\li!let x = e_1 in e_2! :: \li![e_1/x]tau_2!}
  }

  \inferrule[T-Pair]{
    \type{\li!e_1! :: \li!tau_1!} \\
    \type{\li!e_2! :: \li!tau_2!}
  }{
    \type{\li!(e_1,e_2)! :: \li!tau_1*tau_2!}
  }

  \inferrule[T-Proj]{
    \type{\li!e! :: \li!tau_1*tau_2!}
  }{
    \type{\li!pi_b e! :: \li!ite(b,tau_1,tau_2)!}
  }

  \inferrule[T-If]{
    \type{\li!e_0! :: \li!bool!} \\\\
    \type{\li!e_1! :: \li![true/y]tau!} \\
    \type{\li!e_2! :: \li![false/y]tau!}
  }{
    \type{\li!if e_0 then e_1 else e_2! :: \li![e_0/y]tau!}
  }

  \inferrule[T-Ctor]{
    \li!data T =\ !\overline{\li!C tau!} \in \gctx \\\\
    \type{\li!e! :: \li!tau_i!}
  }{
    \type{\li!C_i e! :: \li!T!}
  }

  \inferrule[T-Match]{
    \li!data T =\ !\overline{\li!C tau!} \in \gctx \\
    \type{\li!e_0! :: \li!T!} \\\\
    \forall i.\;
      \type{\extctx{\li!x! :: \li!tau_i!} |- \li!e_i! :: \li![C_i x/y]tau'!}
  }{
    \type{\li!match e_0 with\ !\overline{\li!C x=>e!} :: \li![e_0/y]tau'!}
  }

  \inferrule[T-PsiPair]{
    \li!obliv @T (x:tau) = tau'! \in \gctx \\\\
    \type{\li!e_1! :: \li!tau!} \\
    \type{\li!e_2! :: \li!@T e_1!}
  }{
    \type{\li!&(e_1,e_2)&! :: \li!&@T!}
  }

  \inferrule[T-PsiProj$_{1}$]{
    \li!obliv @T (x:tau) = tau'! \in \gctx \\\\
    \type{\li!e! :: \li!&@T!}
  }{
    \type{\li!pi_1 e! :: \li!tau!}
  }

  \inferrule[T-PsiProj$_{2}$]{
    \li!obliv @T (x:tau) = tau'! \in \gctx \\\\
    \type{\li!e! :: \li!&@T!}
  }{
    \type{\li!pi_2 e! :: \li!@T (pi_1 e)!}
  }

  \inferrule[T-Mux]{
    \type{\li!e_0! :: \li!@bool!} \\
    \kind{\li!tau! :: \li!*@O!} \\\\
    \type{\li!e_1! :: \li!tau!} \\
    \type{\li!e_2! :: \li!tau!}
  }{
    \type{\li!mux e_0 e_1 e_2! :: \li!tau!}
  }

  \inferrule[T-OInj]{
    \type{\li!e! :: \li!ite(b,tau_1,tau_2)!} \\
    \type{\li!tau_1@+tau_2! :: \li!*@O!}
  }{
    \type{\li!@inj_b<tau_1@+tau_2> e! :: \li!tau_1@+tau_2!}
  }

  \inferrule[T-OMatch]{
    \type{\li!e_0! :: \li!tau_1@+tau_2!} \\
    \kind{\li!tau! :: \li!*@O!} \\\\
    \type{\extctx{\li!x! :: \li!tau_1!} |- \li!e_1! :: \li!tau!} \\
    \type{\extctx{\li!x! :: \li!tau_2!} |- \li!e_2! :: \li!tau!}
  }{
    \type{\li!@match e_0 with x=>e_1|x=>e_2! :: \li!tau!}
  }

  \\

  \inferrule[T-Sec]{
    \type{\li!e! :: \li!bool!}
  }{
    \type{\li!@int\#s e! :: \li!@bool!}
  }

  \inferrule[T-BoxedLit]{
  }{
    \type{\li![b]! :: \li!@bool!}
  }

  \inferrule[T-BoxedInj]{
    \ovalty{\li![inj_b<@w> @v]! <~ \li!@w!}
  }{
    \type{\li![inj_b<@w> @v]! :: \li!@w!}
  }
\end{mathpar}
  \caption{\loadtpsi typing rules}
  \label{ap:fig:typing}
\end{figure}

\begin{figure}[h]
\footnotesize
\jbox{\dtype{\li!D!}}
\begin{mathpar}
  \inferrule[DT-Fun]{
    \kind{\empctx |- \li!tau! :: \li!*@*!} \\
    \type{\empctx |- \li!e! :: \li!tau!}
  }{
    \dtype{\li!fn x:tau = e!}
  }

  \inferrule[DT-ADT]{
    \forall i.\; \kind{\empctx |- \li!tau_i! :: \li!*@P!}
  }{
    \dtype{\li!data T =\ !\overline{\li!C tau!}}
  }

  \inferrule[DT-OADT]{
    \kind{\empctx |- \li!tau! :: \li!*@P!} \\
    \kind{\hastype{\li!x! :: \li!tau!} |- \li!tau'! :: \li!*@O!}
  }{
    \dtype{\li!obliv @T (x:tau) = tau'!}
  }
\end{mathpar}
  \caption{\loadtpsi global definitions typing}
  \label{ap:fig:gctx-typing}
\end{figure}

\FloatBarrier

\section{\psistructs and Logical Refinement}

\subsection{OADT structure}

Every OADT (defined with keyword \lstinline|obliv|) \lstinline|@T| must be
equipped with an \emph{OADT-structure}, defined in \Cref{ap:def:oadt-struct}. We
denote the context of OADT-structures by \sctxoadt.

\begin{definition}[OADT-structure]
  \label{ap:def:oadt-struct}

  An OADT-structure of an OADT \lstinline|@T|, with public view
  type \lstinline|tau|, consists of the following (\taypsi) type and functions:
  \begin{itemize}
    \item A public type \haskind{\li!T! :: \li!*@P!}, which is the public
      counterpart of \lstinline|@T|. We say \lstinline|@T| is an OADT
      for \lstinline|T|.
    \item A section function \hastype{\li!s! :: \li!Pik:tau,T->@T k!}, which
      converts a public type to its oblivious counterpart.
    \item A retraction function \hastype{\li!r! :: \li!Pik:tau,@T k->T!}, which
      converts an oblivious type to its public version.
    \item A public view function \hastype{\li!VIEW! :: \li!T->tau!}, which creates
      a valid view of the public type.
    \item A binary relation \hasview{} over values of types \lstinline!T! and
    \lstinline!tau!; \fmath{\li!v! \hasview \li!k!} reads as \lstinline|v| has
    public view \lstinline|k|, or \lstinline|k| is a valid public view of
    \lstinline|v|.
  \end{itemize}

  These operations are required to satisfy the following axioms:
  \begin{itemize}
    \item (\textsc{A-O$_{1}$}) \lstinline|s| and \lstinline|r| are
      a valid section and retraction, i.e., \lstinline|r| is a left-inverse for
      \lstinline|s|, given a valid public view: for any values \hastype{\li!v! ::
      \li!T!}, \hastype{\li!k! :: \li!tau!} and \hastype{\li!@v! :: \li!@T k!},
      if \fmath{\li!v! \hasview \li!k!} and \step*{\li!s k v! ~> \li!@v!}, then
      \step*{\li!r k @v! ~> \li!v!}.
    \item (\textsc{A-O$_{2}$}) the result of \lstinline|r| always has valid
      public view: \step*{\li!r k @v! ~> \li!v!} implies \fmath{\li!v! \hasview
      \li!k!} for all values \hastype{\li!k! :: \li!tau!}, \hastype{\li!@v!
      :: \li!@T k!} and \hastype{\li!v! :: \li!T!}.
    \item (\textsc{A-O$_{3}$}) \lstinline|VIEW| produces a valid public
      view: \step*{\li!VIEW v! ~> \li!k!} implies \fmath{\li!v! \hasview
      \li!k!}, given any values \hastype{\li!v! :: \li!T!} and \hastype{\li!k!
      :: \li!tau!}.
  \end{itemize}
\end{definition}

Note that we do not require these operations to be total as long as they satisfy
the axioms, although it is a desirable property.

\subsection{Logical Refinement}

\Cref{ap:fig:types} defines \emph{simple types}, used for implementing
functionality, and \emph{specification types}, used for specifying policy for
the secure versions. They are required to be well-kinded under empty local
context, i.e., all ADTs and OADTs appear in them are defined. We call a simple
type or specification type \emph{atomic} if it is not formed by the polynomial
type formers, i.e., product and arrow.

\begin{figure}[h]
\raggedright

\textsc{Simple types}

\begin{tabular}{RCLl}
\lstinline|eta| &\production& \lstinline!unit! \mid \lstinline!bool! \mid \lstinline!T! \mid \lstinline!eta*eta! \mid \lstinline!eta->eta! & \\
\end{tabular}

\textsc{specification types}

\begin{tabular}{RCLl}
\lstinline|theta| &\production& \lstinline!unit! \mid \lstinline!bool! \mid \lstinline!@bool! \mid \lstinline!T! \mid \lstinline!&@T! \mid \lstinline!theta*theta! \mid \lstinline!theta->theta! & \\
\end{tabular}

  \caption{Simple types and specification types}
  \label{ap:fig:types}
\end{figure}

\Cref{ap:fig:erasure} defines an erasure function from specification types to
simple types. This function is defined at all specification types, i.e., it is
total. It induces an equivalence relation: the erasure \lstinline|&[theta]&| is
the representative of the equivalence class \eqcls{\li!theta!}, called
\emph{compatibility class}. Two types are said to be \emph{compatible} if they
belong to the same compatibility class. This erasure operation can be naturally
extended to typing context, \erase{\tctx}, by erasing every specification type
in \tctx{} but leaving other types untouched.

\begin{figure}[h]
\footnotesize
\jbox{\lstinline!&[theta]&!}
\begin{mathpar}
  \li!&[unit]&! = \li!unit! \and%
  \li!&[bool]&! = \li!&[@bool]&! = \li!bool! \and%
  \li!&[T]&! = \li!T! \quad\text{where \li!T! is an ADT} \\
  \li!&[&@T]&! = \li!T! \quad\text{where \li!@T! is an OADT of \li!T!} \\
  \li!&[theta*theta]&! = \li!&[theta]&*&[theta]&! \and%
  \li!&[theta->theta]&! = \li!&[theta]&->&[theta]&!
\end{mathpar}

  \caption{Erase specification types to simple types}
  \label{ap:fig:erasure}
\end{figure}

\Cref{ap:fig:log-rel} defines a logical refinement from expressions of
specification types to those of simple types, as a step-indexed logical
relation. Like standard logical relations, this definition is a pair of
set-valued denotation of types: value interpretation \interpV{\li!theta!} and
expression interpretation \interpE{\li!theta!}. We say an expression
\lstinline|e'| of type \lstinline|theta| refines \lstinline|e| of type
\lstinline|&[theta]&| (within $n$ steps) if \fmath{(\li!e!, \li!e'!) \in
\interpE{\li!theta!}}.

All pairs in the relations must be closed and well-typed, i.e., the
interpretations have the forms:
{\footnotesize
  \begin{gather*}
    \interpV{\li!theta!} = \Set{(\li!v!, \li!v'!) |%
      \type{\empctx |- \li!v! :: \li!&[theta]&!} \land%
      \type{\empctx |- \li!v'! :: \li!theta!} \land%
      \ldots} \\
    \interpE{\li!theta!} = \Set{(\li!e!, \li!e'!) |%
      \type{\empctx |- \li!e! :: \li!&[theta]&!} \land%
      \type{\empctx |- \li!e'! :: \li!theta!} \land%
      \ldots}
  \end{gather*}
}
For brevity, we leave this implicit in \Cref{ap:fig:log-rel}.

\Cref{ap:fig:log-rel} also defines an interpretation of a typing context that
maps names to specification types, \interpG{\tctx}. Its codomain is
substitutions (\subst) of pairs of related values. We write \fmath{\subst_1} and
\fmath{\subst_2} for the substitutions that only use the first or the second
component of the pairs.

\begin{figure}[t]
\footnotesize
\jbox{\interpV{\li!theta!}}
\begin{mathpar}
  \interpV{\li!unit!} =%
  \interpV{\li!bool!} =%
  \interpV{\li!T!} = \Set{(\li!v!, \li!v'!) |%
    0 < n \implies \li!v! = \li!v'!} \and
  \interpV{\li!@bool!} =%
  \Set{(\li!b!, \li![b']!) |%
    0 < n \implies \li!b! = \li!b'!} \and
  \interpV{\li!&@T!} =%
  \Set{(\li!v!, \li!&(k,@v)&!) |%
    0 < n \implies \step*{\li!r k @v! ~> \li!v!}} \and
  \interpV{\li!theta_1*theta_2!} =%
  \Set{(\li!(v_1,v_2)!, \li!(v_1',v_2')!) |%
    (\li!v_1!, \li!v_1'!) \in \interpV{\li!theta_1!} \land%
    (\li!v_2!, \li!v_2'!) \in \interpV{\li!theta_2!}} \and
  \interpV{\li!theta_1->theta_2!} =%
  \Set{(\li!?x:&[theta_1]&=>e!, \li!?x:theta_1=>e'!) |%
    \forall i < n.%
    \forall (\li!v!, \li!v'!) \in \interpV[i]{\li!theta_1!}.%
      (\li![v/x]e!, \li![v'/x]e'!) \in \interpE[i]{\li!theta_2!}}
\end{mathpar}

\jbox{\interpE{\li!theta!}}
\begin{mathpar}
  \interpE{\li!theta!} = \Set{(\li!e!, \li!e'!) |%
    \forall i < n.%
    \forall \li!v'!.\;%
    \step[i]{\li!e'! ~> \li!v'!} \implies%
    \exists \li!v!.\; \step*{\li!e! ~> \li!v!} \land%
    (\li!v!, \li!v'!) \in \interpV[n-i]{\li!theta!}}
\end{mathpar}

\jbox{\interpG{\tctx}}
\begin{mathpar}
  \interpG{\empctx} = \Set{\emptyset} \and
  \interpG{\extctx{\li!x! :: \li!theta!}} =%
  \Set{\subst[\li!x! \mapsto (\li!v!, \li!v'!)] |%
    \subst \in \interpG{\tctx} \land%
    (\li!v!, \li!v'!) \in \interpV{\li!theta!}}
\end{mathpar}

  \caption{Refinement as logical relation}
  \label{ap:fig:log-rel}
\end{figure}

\paragraph{Remark} The interpretations \interpV[0]{\li!theta!} and
\interpE[0]{\li!theta!} are total relations of closed values and expressions,
respectively, of type \lstinline|theta|.

\subsection{Join Structures}

An OADT \lstinline|@T| can optionally be equipped with a \emph{join-structure},
shown in \Cref{ap:def:join-struct}. We denote the context of join-structures by
\sctxjoin.

\begin{definition}[join-structure]
  \label{ap:def:join-struct}

  A join-structure of an OADT \lstinline|@T| for \lstinline|T|, with public view
  type \lstinline|tau|, consists of the following operations:
  \begin{itemize}
  \item A binary relation \ple{} on \lstinline|tau|, used to compare two public
    views.
  \item A join function \hastype{\li!JOIN! :: \li!tau->tau->tau!},
    which computes an upper bound of two public views.
  \item A reshape function %
    \hastype{\li!RESHAPE! :: \li!Pik:tau,Pik':tau,@T k->@T k'!}, which converts
    an OADT to one with a different public view.
  \end{itemize}
  such that:
  \begin{itemize}
    \item (\textsc{A-R$_{1}$}) \ple{} is a partial order on \lstinline|tau|.
    \item (\textsc{A-R$_{2}$}) the join function produces an upper bound: given
      values \lstinline!k_1!, \lstinline!k_2! and \lstinline!k! of type
      \lstinline!tau!, if \step*{\li!k_1JOINk_2! ~> \li!k!}, then
      \fmath{\li!k_1! \ple \li!k!} and \fmath{\li!k_2! \ple \li!k!}.
    \item (\textsc{A-R$_{3}$}) the validity of public views is monotone with
      respect to the binary relation \ple: for any values \hastype{\li!v! ::
      \li!T!}, \hastype{\li!k! :: \li!tau!} and \hastype{\li!k'! :: \li!tau!},
      if \fmath{\li!v! \hasview \li!k!} and \fmath{\li!k! \ple \li!k'!}, then
      \fmath{\li!v! \hasview \li!k'!}.
    \item (\textsc{A-R$_{4}$}) the reshape function produces equivalent value,
      as long as the new public view is valid: for any values \hastype{\li!v! ::
      \li!T!}, \hastype{\li!k! :: \li!tau!}, \hastype{\li!k'! :: \li!tau!},
      \hastype{\li!@v! :: \li!@T k!} and \hastype{\li!@v'! :: \li!@T k'!}, if
      \step*{\li!r k @v! ~> \li!v!} and \fmath{\li!v! \hasview \li!k'!} and
      \step*{\li!RESHAPE k k' @v! ~> \li!@v'!}, then \step*{\li!r k' @v'! ~>
      \li!v!}.
  \end{itemize}
\end{definition}

\paragraph{Remark} It is a bit misleading to call the operation \li!JOIN!
``join'', as it only computes an upper bound, not necessarily the lowest one.
However, it \emph{should} compute a supremum for performance reasons:
intuitively, larger public view means more padding.

\Cref{ap:fig:mergeable} defines a relation on specification types, induced by
join-structures. We say \lstinline|theta| is \emph{mergeable} if
\mergeable{\li!theta! |> \li!@ite!}, with witness \lstinline|@ite| of type
\lstinline|@bool->theta->theta->theta|. We may write \mergeable{\li!theta!} when
we do not care about the witness.

\begin{figure}[t]
\footnotesize
\jbox{\mergeable{\li!theta! |> \li!@ite!}}
\begin{mathpar}
  \inferrule{
    \li!theta! \in \set{\li!unit!, \li!@bool!}
  }{
    \mergeable{\li!theta! |> \li!?@b x y=>mux @b x y!}
  }

  \inferrule{
    (\li!@T!, \li!JOIN!, \li!RESHAPE!) \in \sctxjoin
  }{
    \mergeable{\li!&@T! |>%
      \begin{tabular}{@{}l}
        \li!?@b x y=>let k = pi_1 x JOIN pi_1 y in! \\
        \li!\ \ \ \ \ \ \ \ \&(k,mux @b (RESHAPE (pi_1 x) k (pi_2 x))! \\
        \li!\ \ \ \ \ \ \ \ \ \ \ \ \ \ \ \ \ (RESHAPE (pi_1 y) k (pi_2 y)))\&!
      \end{tabular}
    }
  }

  \inferrule{
    \mergeable{\li!theta_1! |> \li!@ite_1!} \\
    \mergeable{\li!theta_2! |> \li!@ite_2!}
  }{
    \mergeable{\li!theta_1*theta_2! |>%
      \li!?@b x y=>(@ite_1 @b (pi_1 x) (pi_1 y),@ite_2 @b (pi_2 x) (pi_2 y))!}
  }

  \inferrule{
    \mergeable{\li!theta_2! |> \li!@ite_2!}
  }{
    \mergeable{\li!theta_1->theta_2! |> \li!?@b x y=>?z=>@ite_2 @b (x z) (y z)!}
  }
\end{mathpar}
  \caption{Mergeable}
  \label{ap:fig:mergeable}
\end{figure}

\subsection{Introduction and Elimination Structures}

An OADT \lstinline|@T| can optionally be equipped with an
\emph{introduction-structure} (intro-structure) and an
\emph{elimination-structure} (elim-structure), shown in
\Cref{ap:def:intro-struct} and \Cref{ap:def:elim-struct} respectively. We denote
the contexts of these two structures by \sctxintro{} and \sctxelim.

\begin{definition}[intro-structure]
  \label{ap:def:intro-struct}

  An intro-structure of an OADT \lstinline|@T| for ADT \lstinline!T!, with
  global definition \lstinline!data T = !$\overline{\li!C eta!}$, consists of a
  set of functions \lstinline|@C_i|, each corresponding to a constructor
  \lstinline|C_i|. The type of \lstinline|@C_i| is \lstinline|theta_i->&@T|,
  where \fmath{\li!&[theta_i]&! = \li!eta_i!} (note that \textsc{DT-ADT}
  guarantees that \lstinline|eta_i| is a simple type). The particular
  \lstinline|theta_i| an intro-structure uses is determined by the author of
  that structure.

  Each \lstinline|@C_i| is required to logically refine the corresponding
  constructor (\textsc{A-I$_1$}): given any values \hastype{\li!v! ::
  \li!&[theta]&!} and \hastype{\li!v'! :: \li!theta!}, if \fmath{(\li!v!,
  \li!v'!) \in \interpV{\li!theta!}}, then \fmath{(\li!C_i v!, \li!@C_i v'!) \in
  \interpE{\li!&@T!}}.
\end{definition}

\begin{definition}[elim-structure]
  \label{ap:def:elim-struct}

  An elim-structure of an OADT \lstinline|@T| for ADT \lstinline!T!, with global
  definition \lstinline!data T = !$\overline{\li!C eta!}$, consists of a family
  of functions \lstinline|@Imatch_alpha|, indexed by the possible return types.
  The type of \lstinline|@Imatch_alpha| is
  \lstinline|&@T->|$\overline{\li!(theta->alpha)!}$\lstinline|->alpha|, where
  \fmath{\li!&[theta_i]&! = \li!eta_i!} for each \lstinline|theta_i| in the
  function arguments corresponding to alternatives.

  Each \lstinline|@Imatch_alpha| is required to logically refine the
  pattern matching expression, specialized with ADT \lstinline|T| and
  return type \lstinline|alpha|. The sole axiom of this structure
  (\textsc{A-E$_1$}) only considers return type \lstinline|alpha|
  being a specification type: given values \hastype{\li!v_i! ::
    \li!eta_i!}, %
  \hastype{\li!&(k,@v)&! :: \li!@T k!}, %
  \hastype{\li!?x=>e_i! :: \li!&[theta_i]&->&[alpha]&!} and %
  \hastype{\li!?x=>e_i'! :: \li!theta_i->alpha!}, %
  if \step*{\li!r k @v! ~> \li!C_i v_i!} and %
  \fmath{(\li!?x=>e_i!, \li!?x=>e_i'!) \in
    \interpV{\li!theta_i->alpha!}} then %
  \fmath{(\li![v_i/x]e_i!,%
    \li!@Imatch &(k,@v)&\ !\overline{\li!(?x=>e')!}) \in
    \interpE{\li!alpha!}}.
\end{definition}

\subsection{Coercion Structure}

Two compatible OADTs may form a \emph{coercion-structure}, shown in
\Cref{ap:def:coer-struct}. We denote the context of this structure by \sctxcoer.

\begin{definition}[coercion-structure]
  \label{ap:def:coer-struct}

  A coercion-structure of a pair of compatible OADTs \lstinline|@T|
  and \lstinline|@T'| for \lstinline|T|, with public view type \lstinline|tau|
  and \lstinline|tau'| respectively, consists of a coercion
  function \lstinline|COER| of type \lstinline|&@T->&@T'|. The coercion should
  produce an equivalent value (\textsc{A-C$_1$}): given values \hastype{\li!v!
    :: \li!T!}, \hastype{\li!&(k,@v)&! :: \li!&@T!} and \hastype{\li!&(k',@v')&!
    :: \li!&@T'!}, if \step*{\li!r k @v! ~> \li!v!} and \step*{\li!COER&(k,@v)&!
    ~> \li!&(k',@v')&!}, then \step*{\li!r k' @v'! ~> \li!v!}.
\end{definition}

\Cref{ap:fig:coercion} defines a relation on specification types, induced by the
coercion-structures. We say \lstinline|theta| is \emph{coercible} to
\lstinline|theta'| if \coercible{\li!theta! ~> \li!theta'! |> \li!COER!}, with
witness \lstinline|COER| of type \lstinline|theta->theta'|. We may write
\coercible{\li!theta! ~> \li!theta'!} when we do not care about the witness.

\begin{figure}[t]
\jbox{\coercible{\li!theta! ~> \li!theta'! |> \li!COER!}}
\footnotesize
\begin{mathpar}
  \inferrule{
  }{
    \coercible{\li!theta! ~> \li!theta! |> \li!?x=>x!}
  }

  \inferrule{
  }{
    \coercible{\li!bool! ~> \li!@bool! |> \li!?x=>@bool\#s x!}
  }

  \inferrule{
    (\li!@T!, \li!T!, \li!s!, \li!r!, \li!VIEW!, \hasview) \in \sctxoadt
  }{
    \coercible{\li!T! ~> \li!&@T! |> \li!?x=>&(VIEW x,s (VIEW x) x)&!}
  }

  \inferrule{
    \hastype{\li!COER! :: \li!&@T->&@T'!} \in \sctxcoer
  }{
    \coercible{\li!&@T! ~> \li!&@T'! |> \li!COER!}
  }

  \inferrule{
    \coercible{\li!theta_1! ~> \li!theta_1'! |> \li!COER_1!} \\
    \coercible{\li!theta_2! ~> \li!theta_2'! |> \li!COER_2!}
  }{
    \coercible{\li!theta_1*theta_2! ~> \li!theta_1'*theta_2'! |>%
      \li!?x=>(COER_1(pi_1 x),COER_2(pi_2 x))!}
  }

  \inferrule{
    \coercible{\li!theta_1'! ~> \li!theta_1! |> \li!COER_1!} \\
    \coercible{\li!theta_2! ~> \li!theta_2'! |> \li!COER_2!}
  }{
    \coercible{\li!theta_1->theta_2! ~> \li!theta_1'->theta_2'! |>%
      \li!?x=>?y=>COER_2(x (COER_1y))!}
  }
\end{mathpar}
  \caption{Coercion}
  \label{ap:fig:coercion}
\end{figure}

\FloatBarrier

\section{Lifting}

\subsection{Declarative Lifting}

\Cref{ap:fig:liftingR} presents the full rules of the declarative lifting
relation. The judgment \liftR{\sctx; \lctx; \gctx; \tctx |- \li!e! :: \li!theta!
|> \li!Oe!} reads as the expression \lstinline|e| of type \lstinline|&[theta]&|
is lifted to the expression \lstinline|Oe| of target type \lstinline|theta|,
under various contexts. The \psistruct context \sctx{} consists of the set of
OADT-structures (\sctxoadt{}), join-structures (\sctxjoin{}), intro-structures
(\sctxintro{}), elim-structures (\sctxelim{}) and coercion-structures
(\sctxcoer{}), respectively. The local typing context \tctx{} maps local
variables to the target type. The lifting context \lctx{} consists of entries of
the form \hastype{\li!x! :: \li!theta! |> \li!Ox!}, which associates the global
function \lstinline|x| of type \lstinline|&[theta]&| with a target function
\lstinline|Ox| of the target type \lstinline|theta|. We elide most contexts as
they are fixed, and simply write \liftR{\li!e! :: \li!theta! |> \li!Oe!} for
brevity.

\begin{figure}[t]
\footnotesize
\jbox{\liftR{\li!e! :: \li!theta! |> \li!Oe!}}
\begin{mathpar}
  \inferrule[L-Unit]{
  }{
    \liftR{\li!()! :: \li!unit! |> \li!()!}
  }

  \inferrule[L-Lit]{
  }{
    \liftR{\li!b! :: \li!bool! |> \li!b!}
  }

  \inferrule[L-Var]{
    \hastype{\li!x! :: \li!theta!} \in \tctx
  }{
    \liftR{\li!x! :: \li!theta! |> \li!x!}
  }

  \inferrule[L-Fun]{
    \hastype{\li!x! :: \li!theta! |> \li!Ox!} \in \lctx
  }{
    \liftR{\li!x! :: \li!theta! |> \li!Ox!}
  }

  \inferrule[L-Abs]{
    \liftR{\extctx{\li!x! :: \li!theta_1!} |- \li!e! :: \li!theta_2! |> \li!Oe!}
  }{
    \liftR{\li!?x:&[theta_1]&=>e! :: \li!theta_1->theta_2! |>%
      \li!?x:theta_1=>Oe!}
  }

  \inferrule[L-App]{
    \liftR{\li!e_2! :: \li!theta_1->theta_2! |> \li!Oe_2!} \\
    \liftR{\li!e_1! :: \li!theta_1! |> \li!Oe_1!}
  }{
    \liftR{\li!e_2 e_1! :: \li!theta_2! |> \li!Oe_2 Oe_1!}
  }

  \inferrule[L-Let]{
    \liftR{\li!e_1! :: \li!theta_1! |> \li!Oe_1!} \\
    \liftR{\extctx{\li!x! :: \li!theta_1!} |-%
      \li!e_2! :: \li!theta_2! |> \li!Oe_2!}
  }{
    \liftR{\li!let x = e_1 in e_2! :: \li!theta_2! |> \li!let x = Oe_1 in Oe_2!}
  }

  \inferrule[L-Pair]{
    \liftR{\li!e_1! :: \li!theta_1! |> \li!Oe_1!} \\
    \liftR{\li!e_2! :: \li!theta_2! |> \li!Oe_2!}
  }{
    \liftR{\li!(e_1,e_2)! :: \li!theta_1*theta_2! |>%
      \li!(Oe_1,Oe_2)!}
  }

  \inferrule[L-Proj]{
    \liftR{\li!e! :: \li!theta_1*theta_2! |> \li!Oe!}
  }{
    \liftR{\li!pi_b e! :: \li!ite(b,theta_1,theta_2)! |> \li!pi_b Oe!}
  }

  \inferrule[L-If$_1$]{
    \liftR{\li!e_0! :: \li!bool! |> \li!Oe_0!} \\\\
    \liftR{\li!e_1! :: \li!theta! |> \li!Oe_1!} \\
    \liftR{\li!e_2! :: \li!theta! |> \li!Oe_2!}
  }{
    \liftR{\li!if e_0 then e_1 else e_2! :: \li!theta! |>%
      \li!if Oe_0 then Oe_1 else Oe_2!}
  }

  \inferrule[L-If$_2$]{
    \liftR{\li!e_0! :: \li!@bool! |> \li!Oe_0!} \\
    \mergeable{\li!theta! |> \li!@ite!} \\\\
    \liftR{\li!e_1! :: \li!theta! |> \li!Oe_1!} \\
    \liftR{\li!e_2! :: \li!theta! |> \li!Oe_2!} \\
  }{
    \liftR{\li!if e_0 then e_1 else e_2! :: \li!theta! |>%
      \li!@ite Oe_0 Oe_1 Oe_2!}
  }

  \inferrule[L-Ctor$_1$]{
    \li!data T =\ !\overline{\li!C eta!} \in \gctx \\
    \liftR{\li!e! :: \li!eta_i! |> \li!Oe!}
  }{
    \liftR{\li!C_i e! :: \li!T! |> \li!C_i Oe!}
  }

  \inferrule[L-Ctor$_2$]{
    \hastype{\li!@C_i! :: \li!theta_i->&@T!} \in \sctxintro \\
    \liftR{\li!e! :: \li!theta_i! |> \li!Oe!}
  }{
    \liftR{\li!C_i e! :: \li!&@T! |> \li!@C_i Oe!}
  }

  \inferrule[L-Match$_1$]{
    \li!data T =\ !\overline{\li!C eta!} \in \gctx \\
    \liftR{\li!e_0! :: \li!T! |> \li!Oe_0!} \\
    \forall i.\; \liftR{\extctx{\li!x! :: \li!eta_i!} |-%
      \li!e_i! :: \li!theta'! |> \li!Oe_i!}
  }{
    \liftR{\li!match e_0 with\ !\overline{\li!C x=>e!} :: \li!theta'! |>%
      \li!match Oe_0 with\ !\overline{\li!C x=>Oe!}}
  }

  \inferrule[L-Match$_2$]{
    \hastype{\li!@Imatch! ::%
      \li!&@T->!\overline{\li!(theta->theta')!}\li!->theta'!} \in \sctxelim \\
    \liftR{\li!e_0! :: \li!&@T! |> \li!Oe_0!} \\
    \forall i.\; \liftR{\extctx{\li!x! :: \li!theta_i!} |-%
      \li!e_i! :: \li!theta'! |> \li!Oe_i!}
  }{
    \liftR{\li!match e_0 with\ !\overline{\li!C x=>e!} :: \li!theta'! |>%
      \li!@Imatch Oe_0\ !\overline{\li!(?x:theta=>Oe)!}}
  }

  \inferrule[L-Coerce]{
    \liftR{\li!e! :: \li!theta! |> \li!Oe!} \\
    \coercible{\li!theta! ~> \li!theta'! |> \li!COER!}
  }{
    \liftR{\li!e! :: \li!theta'! |> \li!COEROe!}
  }
\end{mathpar}
  \caption{Declarative lifting rules}
  \label{ap:fig:liftingR}
\end{figure}

\subsection{Algorithmic Lifting}

\Cref{ap:fig:ppx} defines a set of typed macros that take types as parameters
and elaborate to expressions accordingly, under the contexts \sctx, \lctx{} and
\gctx{} implicitly.

\begin{figure}[t]
\footnotesize
\jbox{\ppx{\li!Mite(theta_0,theta;e_0,e_1,e_2)! |> \li!Oe!}}
\begin{mathpar}
  \inferrule{
  }{
    \ppx{\li!Mite(bool,theta;e_0,e_1,e_2)! |> \li!if e_0 then e_1 else e_2!}
  }

  \inferrule{
    \mergeable{\li!theta! |> \li!@ite!}
  }{
    \ppx{\li!Mite(@bool,theta;e_0,e_1,e_2)! |> \li!@ite e_0  e_1 e_2!}
  }
\end{mathpar}

\jbox{\ppx{\li!MC(theta,theta';e)! |> \li!Oe!}}
\begin{mathpar}
  \inferrule{
    \li!data T =\ !\overline{\li!C eta!} \in \gctx
  }{
    \ppx{\li!MC_i(eta_i,T;e)! |> \li!C_i e!}
  }

  \inferrule{
    \hastype{\li!@C_i! :: \li!theta_i->&@T!} \in \sctxintro
  }{
    \ppx{\li!MC_i(theta_i,&@T;e)! |> \li!@C_i e!}
  }
\end{mathpar}

\jbox{\ppx{\li!Mmatch(theta_0,!\overline{\li!theta!}\li!,theta';%
    e_0,!\overline{\li!e!}\li!)! |> \li!Oe!}}
\begin{mathpar}
  \inferrule{
    \li!data T =\ !\overline{\li!C eta!} \in \gctx
  }{
    \ppx{\li!Mmatch(T,!\overline{\li!eta!}\li!,theta';%
      e_0,!\overline{\li!e!}\li!)! |>%
      \li!match e_0 with\ !\overline{\li!C x=>e!}}
  }

  \inferrule{
    \hastype{\li!@Imatch! ::%
      \li!&@T->!\overline{\li!(theta->theta')!}\li!->theta'!} \in \sctxelim
  }{
    \ppx{\li!Mmatch(&@T,!\overline{\li!theta!}\li!,theta';%
      e_0,!\overline{\li!e!}\li!)! |>%
      \li!@Imatch e_0\ !\overline{\li!(?x:theta=>e)!}}
  }
\end{mathpar}

\jbox{\ppx{\li!MCOER(theta,theta';e)! |> \li!Oe!}}
\begin{mathpar}
  \inferrule{
    \coercible{\li!theta! ~> \li!theta'! |> \li!COER!}
  }{
    \ppx{\li!MCOER(theta,theta';e)! |> \li!COERe!}
  }
\end{mathpar}

\jbox{\ppx{\li!Mx(theta)! |> \li!Oe!}}
\begin{mathpar}
  \inferrule{
    \hastype{\li!x! :: \li!theta! |> \li!Ox!} \in \lctx
  }{
    \ppx{\li!Mx(theta)! |> \li!Ox!}
  }
\end{mathpar}
  \caption{Typed macros}
  \label{ap:fig:ppx}
\end{figure}

\Cref{ap:fig:constraints} defines the constraints generated by the algorithm,
where \lstinline|thetaV| is the specification types extended with type
variables.

\Cref{ap:fig:liftingA} presents the full rules of the lifting algorithm. The
judgment \liftA{\gctx; \tctx |- \li!e! :: \li!eta! ~ \li!X! |> \li!Oe! ~>
\cstrs} reads as the source expression \lstinline|e| of type \lstinline|eta| is
lifted to the target expression \lstinline|Oe| whose type is a \emph{type
variable} \lstinline|X|, and generates constraints \cstrs. The source expression
\lstinline|e| is required to be in \emph{administrative normal form} (ANF). The
typing context \tctx{} has the form \hastype{\li!x! :: \li!eta! ~ \li!X!},
meaning that local variable \lstinline|x| has type \lstinline|eta| in the source
program and \lstinline|X| in the target program.

\begin{figure}[t]
\raggedright

\textsc{Constraints}

\begin{tabular}{RC>{\footnotesize\(}l<{\)}l}
\cstr &\production& \incls{\li!X! ~ \li!eta!} \mid%
  \li!thetaV! = \li!thetaV! \mid \li!Mite(thetaV,thetaV)! \mid%
  \li!MC(thetaV,thetaV)! \mid%
  \li!Mmatch(thetaV,!\overline{\li!thetaV!}\li!,thetaV)! \mid%
  \li!MCOER(thetaV,thetaV)! \mid \li!Mx(thetaV)! & \\
\end{tabular}

  \caption{Constraints}
  \label{ap:fig:constraints}
\end{figure}

\begin{figure}[t]
\footnotesize
\jbox{\liftA{\li!e! :: \li!eta! ~ \li!X! |> \li!Oe! ~> \cstrs}}
\begin{mathpar}
  \inferrule[A-Unit]{
  }{
    \liftA{\li!()! :: \li!unit! ~ \li!X! |> \li!()! ~> \li!X! = \li!unit!}
  }

  \inferrule[A-Lit]{
  }{
    \liftA{\li!b! :: \li!bool! ~ \li!X! |> \li!b! ~> \li!X! = \li!bool!}
  }

  \inferrule[A-Var]{
    \hastype{\li!x! :: \li!eta! ~ \li!X!} \in \tctx
  }{
    \liftA{\li!x! :: \li!eta! ~ \li!X'! |> \li!MCOER(X,X';x)! ~> \li!MCOER(X,X')!}
  }

  \inferrule[A-Fun]{
    \li!fn x:eta = e! \in \gctx
  }{
    \liftA{\li!x! :: \li!eta! ~ \li!X! |> \li!Mx(X)! ~> \li!Mx(X)!}
  }

  \inferrule[A-Abs]{
    \fresh{\li!X_1!; \li!X_2!} \\
    \liftA{\extctx{\li!x! :: \li!eta_1! ~ \li!X_1!} |-%
      \li!e! :: \li!eta_2! ~ \li!X_2! |> \li!Oe! ~> \cstrs}
  }{
    \liftA{\li!?x:eta_1=>e! :: \li!eta_1->eta_2! ~ \li!X! |>%
      \li!?x:X_1=>Oe! ~>%
      \incls{\li!X_1! ~ \li!eta_1!}, \incls{\li!X_2! ~ \li!eta_2!},%
      \li!X! = \li!X_1->X_2!, \cstrs}
  }

  \inferrule[A-App]{
    \fresh{\li!X_1!} \\
    \hastype{\li!x_2! :: \li!eta_1->eta_2! ~ \li!X!} \in \tctx \\
    \liftA{\li!x_1! :: \li!eta_1! ~ \li!X_1! |> \li!Oe_1! ~> \cstrs}
  }{
    \liftA{\li!x_2 x_1! :: \li!eta_2! ~ \li!X_2! |> \li!x_2 Oe_1! ~>%
      \incls{\li!X_1! ~ \li!eta_1!}, \li!X! = \li!X_1->X_2!, \cstrs}
  }

  \inferrule[A-Let]{
    \fresh{\li!X_1!} \\
    \liftA{\li!e_1! :: \li!eta_1! ~ \li!X_1! |> \li!Oe_1! ~> \cstrs_1} \\
    \liftA{\extctx{\li!x! :: \li!eta_1! ~ \li!X_1!} |-%
      \li!e_2! :: \li!eta_2! ~ \li!X_2! |> \li!Oe_2! ~> \cstrs_2}
  }{
    \liftA{\li!let x:eta_1 = e_1 in e_2! :: \li!eta_2! ~ \li!X_2! |>%
      \li!let x:X_1 = Oe_1 in Oe_2! ~>%
      \incls{\li!X_1! ~ \li!eta_1!}, \cstrs_1, \cstrs_2}
  }

  \inferrule[A-Pair]{
    \hastype{\li!x_1! :: \li!eta_1! ~ \li!X_1!} \in \tctx \\
    \hastype{\li!x_2! :: \li!eta_2! ~ \li!X_2!} \in \tctx
  }{
    \liftA{\li!(x_1,x_2)! :: \li!eta_1*eta_2! ~ \li!X! |>%
      \li!(x_1,x_2)! ~> \li!X! = \li!X_1*X_2!}
  }

  \inferrule[A-Proj]{
    \fresh{\li!ite(b,X_2,X_1)!} \\
    \hastype{\li!x! :: \li!eta_1*eta_2! ~ \li!X!} \in \tctx
  }{
    \liftA{\li!pi_b x! :: \li!ite(b,eta_1,eta_2)! ~ \li!ite(b,X_1,X_2)! |>%
      \li!pi_b x! ~>%
      \incls{\li!ite(b,X_2,X_1)! ~ \li!ite(b,eta_2,eta_1)!}, \li!X! = \li!X_1*X_2!}
  }

  \inferrule[A-If]{
    \hastype{\li!x_0! :: \li!bool! ~ \li!X_0!} \in \tctx \\
    \liftA{\li!e_1! :: \li!eta! ~ \li!X! |> \li!Oe_1! ~> \cstrs_1} \\
    \liftA{\li!e_2! :: \li!eta! ~ \li!X! |> \li!Oe_2! ~> \cstrs_2} \\
  }{
    \liftA{\li!if x_0 then e_1 else e_2! :: \li!eta! ~ \li!X! |>%
      \li!Mite(X_0,X;x_0,Oe_1,Oe_2)! ~> \li!Mite(X_0,X)!, \cstrs_1, \cstrs_2}
  }

  \inferrule[A-Ctor]{
    \li!data T =\ !\overline{\li!C eta!} \in \gctx \\
    \fresh{\li!X_i!} \\
    \liftA{\li!x! :: \li!eta_i! ~ \li!X_i! |> \li!Oe! ~> \cstrs}
  }{
    \liftA{\li!C_i x! :: \li!T! ~ \li!X! |> \li!MC_i(X_i,X;Oe)! ~>%
      \incls{\li!X_i! ~ \li!eta_i!}, \li!MC_i(X_i,X)!, \cstrs}
  }

  \inferrule[A-Match]{
    \li!data T =\ !\overline{\li!C eta!} \in \gctx \\
    \fresh{\overline{\li!X!}} \\
    \hastype{\li!x_0! :: \li!T! ~ \li!X_0!} \in \tctx \\
    \forall i.\; \liftA{\extctx{\li!x! :: \li!eta_i! ~ \li!X_i!} |-%
      \li!e_i! :: \li!eta'! ~ \li!X'! |> \li!Oe_i! ~> \cstrs_i}
  }{
    \liftA{\li!match x_0 with\ !\overline{\li!C x=>e!} :: \li!eta'! ~ \li!X'! |>%
      \li!Mmatch(X_0,!\overline{\li!X!}\li!,X';%
      x_0,!\overline{\li!Oe!}\li!)! ~>%
      \overline{\incls{\li!X! ~ \li!eta!}},%
      \li!Mmatch(X_0,!\overline{\li!X!}\li!,X')!,%
      \overline{\cstrs}}
  }
\end{mathpar}

  \caption{Algorithmic lifting rules}
  \label{ap:fig:liftingA}
\end{figure}

\FloatBarrier

\subsection{Metatheory}

\begin{definition}[Well-typedness, derivability and validity of lifting context]
  A lifting context \lctx{} is well-typed (under \gctx) if and only if, for any
  \fmath{\hastype{\li!x! :: \li!theta! |> \li!Ox!} \in \lctx}, \type{\gctx;
    \empctx |- \li!x! :: \li!&[theta]&!} and \type{\gctx; \empctx |- \li!Ox! ::
    \li!theta!}.

  A lifting context is derivable, denoted by \lctxderiv{\lctx}, if and only if,
  for any \fmath{\hastype{\li!x! :: \li!theta! |> \li!Ox!} \in \lctx}, %
  \fmath{\li!fn x:&[theta]& = e! \in \gctx} and %
  \fmath{\li!fn Ox:theta = Oe! \in \gctx} for some \lstinline|e|
  and \lstinline|Oe|, such that \liftR{\sctx; \lctx; \gctx; \empctx |- \li!e! ::
    \li!theta! |> \li!Oe!}.

  A lifting context is $n$-valid, denoted by \lctxvalid[n]{\lctx}, if and only
  if, for any \fmath{\hastype{\li!x! :: \li!theta! |> \li!Ox!} \in \lctx},
  \fmath{(\li!x!, \li!Ox!) \in \interpE{\li!theta!}}. If \lctxvalid[n]{\lctx}
  for any $n$, we say \lctx{} is valid, denoted by \lctxvalid{\lctx}.

  Obviously, derivability or validity implies well-typedness.
\end{definition}

The theorems in this section assume a well-typed global context \gctx. We also
(explicitly or implicitly) apply some standard results about \loadtpsi that are
proved in Coq: weakening lemma (\verb!weakening.v!), substitution lemma
(\verb!preservation.v!), preservation theorem (\verb!preservation.v!) and
canonical forms of values (\verb!progress.v!). We do not state them formally
here, as they are all standard.

\begin{lemma}
  \label{ap:thm:erase-eta-eq}
  \fmath{\li!&[eta]&! = \li!eta!}.
\end{lemma}
\begin{proof}
  By routine induction on \li!eta!.
\end{proof}

\begin{lemma}[Regularity of mergeability]
  \label{ap:thm:mergeable-reg}
  \mergeable{\li!theta! |> \li!@ite!} implies\/ \type{\empctx |- \li!@ite! ::
    \li!@bool->theta->theta->theta!}.
\end{lemma}
\begin{proof}
  By routine induction on the derivation of\/ \mergeable{\li!theta! |>
    \li!@ite!} and applying typing rules as needed.
\end{proof}

\begin{lemma}[Regularity of coercibility]
  \label{ap:thm:coercible-reg}
  \coercible{\li!theta! ~> \li!theta'! |> \li!COER!}\/ implies\/
  \fmath{\li!&[theta]&! = \li!&[theta']&!} and \type{\empctx |- \li!COER! ::
    \li!theta->theta'!}.
\end{lemma}
\begin{proof}
  By routine induction on the derivation of \coercible{\li!theta! ~> \li!theta'!
    |> \li!COER!} and applying typing rules as needed.
\end{proof}

\begin{theorem}[Regularity of declarative lifting]
  \label{ap:thm:liftingR-reg}
  Suppose\/ \lctx\/ is well-typed and\/ \liftR{\sctx; \lctx; \gctx; \tctx |-
    \li!e! :: \li!theta! |> \li!Oe!}. We have\/ \type{\gctx; \erase{\tctx} |-
    \li!e! :: \li!&[theta]&!} and\/ \type{\gctx; \tctx |- \li!Oe! ::
    \li!theta!}.
\end{theorem}
\begin{proof}
  By induction on the derivation of the declarative lifting judgment.

  The cases on \textsc{L-Unit}, \textsc{L-Lit}, \textsc{L-Var} and
  \textsc{L-Fun} are trivial.

  The cases on \textsc{L-Abs}, \textsc{L-App}, \textsc{L-Let}, \textsc{L-Pair},
  \textsc{L-Proj}, \textsc{L-If$_1$}, \textsc{L-Ctor$_1$} and
  \textsc{L-Match$_1$} are straightforward to prove, as they are simply
  congruence cases. \textsc{L-Ctor$_1$} and \textsc{L-Match$_1$} also rely on
  \Cref{ap:thm:erase-eta-eq}. Here we show only the proof on \textsc{L-Abs}.
  Other cases are similar.

  To prove \type{\li!?x:theta_1=>Oe! :: \li!theta_1->theta_2!}, by
  \textsc{T-Abs}, we need to show \type{\extctx{\li!x! :: \li!theta_1!} |-
    \li!Oe! :: \li!theta_2!}, which follows immediately by induction hypothesis.
  The side condition of \li!theta_1! being well-kinded under \tctx{} is
  immediate from the fact that specification types are well-kinded and the
  weakening lemma. The other part of this case, \type{\li!?x:&[theta_1]&=>e! ::
    \li!&[theta_1->theta_2]&! = \li!&[theta_1]&->&[theta_2]&!} proceeds
  similarly.

  Case \textsc{L-If$_2$}: by \Cref{ap:thm:mergeable-reg}, we know \type{\empctx
    |- \li!@ite! :: \li!@bool->theta->theta->theta!}. Applying \textsc{T-App}
  and induction hypothesis, we have \type{\li!@ite Oe_0 Oe_1 Oe_2! ::
    \li!theta!}, as desired. The other half is straightforward.

  Cases \textsc{L-Ctor$_2$} and \textsc{L-Match$_2$} rely on the properties of
  structures \sctxintro{} and \sctxelim{}, but otherwise proceed similarly to
  other cases.

  Case \textsc{L-Coerce}: by induction hypothesis, we have \type{\erase{\tctx}
    |- \li!e! :: \li!&[theta]&!} and \type{\li!Oe! :: \li!theta!}. Since
  \fmath{\li!&[theta]&! = \li!&[theta']&!} by \Cref{ap:thm:coercible-reg},
  \type{\erase{\tctx} |- \li!e! :: \li!&[theta']&!} as required. On the other
  hand, \type{\empctx |- \li!COER! :: \li!theta->theta'!} by
  \Cref{ap:thm:coercible-reg}. It follows immediately that \type{\li!COEROe! ::
    \li!theta'!}.
\end{proof}

The regularity theorem ensures that the lifted expression is well-typed, and
provides the security guarantee, as well-typed programs are oblivious by the
obliviousness theorem.

\begin{lemma}[Multi-substitution]
  \label{ap:thm:msubst}
  Let \fmath{\subst \in \interpG{\tctx}}. If\/ \type{\li!e! :: \li!theta!},
  then\/ \type{\empctx |- \subst_2(\li!e!) :: \li!theta!}, and if\/
  \type{\erase{\tctx} |- \li!e! :: \li!eta!}, then\/ \type{\empctx |-
    \subst_1(\li!e!) :: \li!eta!}.
\end{lemma}
\begin{proof}
  By routine induction on the structure of \tctx, and applying substitution
  lemma when needed. Note that \li!theta! and \li!eta! do not contain any local
  variables, so substitution on these types does nothing.
\end{proof}

\begin{lemma}
  \label{ap:thm:public-type-interpV}
  Suppose\/ \kind{\empctx |- \li!eta! :: \li!*@P!}. We have\/ \fmath{(\li!v!,
    \li!v!) \in \interpV{\li!eta!}} for any\/ \hastype{\li!v! :: \li!eta!}.
  Conversely, if\/ \fmath{(\li!v!, \li!v'!) \in \interpV{\li!eta!}} for $n > 0$,
  then\/ \fmath{\li!v! = \li!v'!}.
\end{lemma}
\begin{proof}
  By routine induction on the structure of \li!eta!.
\end{proof}

\begin{lemma}[Anti-monotonicity]
  \label{ap:thm:anti-mono}
  If\/ $m \le n$, then\/ \fmath{\interpV[n]{\li!theta!} \subseteq
    \interpV[m]{\li!theta!}}, and \fmath{\interpE[n]{\li!theta!} \subseteq
    \interpE[m]{\li!theta!}}, and \fmath{\interpG[n]{\tctx} \subseteq
    \interpG[m]{\tctx}}, and \fmath{\lctxvalid[n]{\lctx} \implies
    \lctxvalid[m]{\lctx}}.
\end{lemma}
\begin{proof}
  To prove the case of value interpretation, proceed by routine induction on
  \li!theta!. The cases of other interpretations are straightforward.
\end{proof}

\begin{lemma}[Correctness of mergeability]
  \label{ap:thm:mergeable-correct}
  Suppose\/ \mergeable{\li!theta! |> \li!@ite!}. If\/ \fmath{(\li!v_1!,
    \li!Ov_1!) \in \interpV{\li!theta!}} and\/ \fmath{(\li!v_2!, \li!Ov_2!) \in
    \interpV{\li!theta!}}, and\/ \step*{\li!@ite [b] Ov_1 Ov_2! ~> \li!Ov!},
  then\/ \fmath{(\li!ite(b,v_1,v_2)!, \li!Ov!) \in \interpV{\li!theta!}}.
\end{lemma}
\begin{proof}
  By induction on the derivation of mergeability. We assume $n > 0$ because
  otherwise it is trivial. The cases when \lstinline|theta| is \lstinline|unit|
  or \lstinline|@bool| is trivial. The case of product type is routine, similar
  to the case of function type.

  Case on function type \lstinline|theta_1->theta2|: by assumption,
  \begin{itemize}
    \item \fmath{\li!v_1! = \li!?z=>e_1!} for some \lstinline|e_1|
    \item \fmath{\li!Ov_1! = \li!?z=>Oe_1!} for some \lstinline|Oe_1|
    \item \fmath{\li!v_2! = \li!?z=>e_2!} for some \lstinline|e_2|
    \item \fmath{\li!Ov_2! = \li!?z=>Oe_2!} for some \lstinline|Oe_2|
  \end{itemize}
  Suppose \step*{\li!@ite [b] Ov_1 Ov_2! ~> \li!Ov!}, hence \fmath{\li!Ov! =%
    \li!?z=>@ite_2 [b] (Ov_1 z) (Ov_2 z)!}. Suppose $i < n$, and \fmath{(\li!u!,
    \li!Ou!) \in \interpV[i]{\li!theta_1!}}. We want to show
  \fmath{(\li!ite(b,[u/z]e_1,[u/z]e_2)!, \li!@ite_2 [b] (Ov_1 Ou) (Ov_2 Ou)!)
    \in \interpE[i]{\li!theta_2!}}. Suppose $j < i$. We have the following
  trace:

  {\footnotesize\setlength{\abovedisplayskip}{0pt}
    \begin{align*}
      \li!@ite_2 [b] (Ov_1 Ou) (Ov_2 Ou)!%
      &\steparrow[1] \li!@ite_2 [b] (Ov_1 Ou) ([Ou/z]Oe_2)! \\
      &\steparrow[j_2] \li!@ite_2 [b] (Ov_1 Ou) Ow_2! \\
      &\steparrow[1] \li!@ite_2 [b] ([Ou/z]Oe_1) Ow_2! \\
      &\steparrow[j_1] \li!@ite_2 [b] Ow_1 Ow_2! \\
      &\steparrow[j_3] \li!Ow!
    \end{align*}
  }%
  where $j = j_1 + j_2 + j_3 + 2$, with \step[j_2]{\li![Ou/z]Oe_2! ~> \li!Ow_2!}
  and \step[j_1]{\li![Ou/z]Oe_1! ~> \li!Ow_1!}. We instantiate the assumptions
  with $i$ and \lstinline|u| and \lstinline|Ou| to get:
  \begin{itemize}
    \item \fmath{(\li![u/z]e_1!, \li![Ou/z]Oe_1!) \in \interpE[i]{\li!theta_2!}}
    \item \fmath{(\li![u/z]e_2!, \li![Ou/z]Oe_2!) \in \interpE[i]{\li!theta_2!}}
  \end{itemize}
  It then follows that:
  \begin{itemize}
    \item \step*{\li![u/z]e_1! ~> \li!w_1!} for some value \lstinline!w_1!
    \item \fmath{(\li!w_1!, \li!Ow_1!) \in \interpV[i-j_1]{\li!theta_2!}}
    \item \step*{\li![u/z]e_2! ~> \li!w_2!} for some value \lstinline!w_2!
    \item \fmath{(\li!w_2!, \li!Ow_2!) \in \interpV[i-j_2]{\li!theta_2!}}
  \end{itemize}
  Hence \step*{\li!ite(b,[u/z]e_1,[u/z]e_2)! ~> \li!ite(b,w_1,w_2)!}. It remains
  to show \fmath{(\li!ite(b,w_1,w_2)!, \li!Ow!) \in
    \interpV[i-j]{\li!theta_2!}}, but it follows immediately by induction
  hypothesis and \Cref{ap:thm:anti-mono}.

  Case on \tpsi-type \lstinline|&@T|: by assumption,
  \begin{itemize}
    \item \fmath{\li!Ov_1! = \li!&(k_1,@v_1)&!} for some values \lstinline|k_1|
      and \lstinline|@v_1|
    \item \step*{\li!r k_1 @v_1! ~> \li!v_1!}
    \item \fmath{\li!Ov_2! = \li!&(k_2,@v_2)&!} for some values \lstinline|k_2|
      and \lstinline|@v_2|
    \item \step*{\li!r k_2 @v_2! ~> \li!v_2!}
  \end{itemize}
  Suppose \step*{\li!@ite [b] Ov_1 Ov_2! ~> \li!Ov!}. We get the following trace:

  {\footnotesize\setlength{\abovedisplayskip}{0pt}
    \begin{align*}
      \li!@ite [b] Ov_1 Ov_2!%
      &\steparrow*%
        \begin{tabular}{@{}l}
          \li!let k = pi_1 Ov_1 JOIN pi_1 Ov_2 in! \\
          \li!\&(k,mux [b] (RESHAPE (pi_1 Ov_1) k (pi_2 Ov_1))! \\
          \li!\ \ \ \ \ \ \ \ \ \ \ (RESHAPE (pi_1 Ov_2) k (pi_2 Ov_2)))\&!
        \end{tabular} \\
      &\steparrow* \li!let k = k_1JOINk_2 in ...! \\
      &\steparrow*%
          \li!\&(k,mux [b] (RESHAPE (pi_1 Ov_1) k (pi_2 Ov_1)) (...))\&! \\
      &\steparrow*%
        \li!\&(k,mux [b] (RESHAPE k_1 k @v_1) (...))\&! \\
      &\steparrow*%
        \li!\&(k,mux [b] @v_1' (RESHAPE (pi_1 Ov_2) k (pi_2 Ov_2)))\&! \\
      &\steparrow* \li!\&(k,mux [b] @v_1' (RESHAPE k_2 k @v_2))\&! \\
      &\steparrow* \li!\&(k,mux [b] @v_1' @v_2')\&! \\
      &\steparrow* \li!\&(k,ite(b,@v_1',@v_2'))\&! = \li!Ov!
    \end{align*}
  }%
  with \step*{\li!k_1JOINk_2! ~> \li!k!}, \step*{\li!RESHAPE k_1 k @v_1! ~>
    \li!@v_1'!} and \step*{\li!RESHAPE k_2 k @v_2! ~> \li!@v_2'!}. Because
  \fmath{\li!v_1! \hasview \li!k_1!} by \textsc{A-O$_2$}, and \fmath{\li!k_1!
    \ple \li!k!} by \textsc{A-R$_2$}, we have \fmath{\li!v_1! \hasview \li!k!}
  by \textsc{A-R$_3$}. It follows that \step*{\li!r k @v_1'! ~> \li!v_1!} by
  \textsc{A-R$_4$}. Similarly, we get \step*{\li!r k @v_2'! ~> \li!v_2!}. Hence,
  \step*{\li!r k ite(b,@v_1',@v_2')! ~> \li!ite(b,v_1,v_2)!}. That is to say
  \fmath{(\li!ite(b,v_1,v_2)!, \li!Ov!) \in \interpV{\li!&@T!}}, as desired.
\end{proof}

\begin{lemma}[Correctness of coercibility]
  \label{ap:thm:coercible-correct}
  Suppose\/ \coercible{\li!theta! ~> \li!theta'! |> \li!COER!}. If\/
  \fmath{(\li!v!, \li!Ov!) \in \interpV{\li!theta!}} and\/ \step*{\li!COEROv! ~>
    \li!Ov'!}, then\/ \fmath{(\li!v!, \li!Ov'!) \in \interpV{\li!theta'!}}.
\end{lemma}
\begin{proof}
  By induction on the derivation of coercibility. We only consider $n > 0$ since
  it is otherwise trivial. The cases on identity coercion and boolean coercion
  are trivial. The case on \coercible{\li|&@T| ~> \li|&@T'|} is immediate from
  \textsc{A-C$_1$}. The case on product type is routine, similar to the case of
  function type.

  Case on \coercible{\li!T! ~> \li!&@T!}: suppose \fmath{(\li!v!, \li!Ov!) \in
    \interpV{\li!T!}}, i.e., \fmath{\li!v! = \li!Ov!} by definition, and
  \step*{\li!COEROv! = \li!COERv! ~> \li!Ov'!}. We have the following trace:

  {\footnotesize\setlength{\abovedisplayskip}{0pt}
    \begin{align*}
      \li!COERv!%
      &\steparrow \li!\&(VIEW v,s (VIEW v) v)\&! \\
      &\steparrow* \li!\&(k,s k v)\&! \\
      &\steparrow* \li!\&(k,@v)\&! = \li!Ov'!
    \end{align*}
  }%
  with \step*{\li!VIEW v! ~> \li!k!} and \step*{\li!s k v! ~> \li!@v!}. Knowing
  \fmath{\li!v! \hasview \li!k!} by \textsc{A-O$_3$}, we have \step*{\li!r k @v!
    ~> \li!v!} by \textsc{A-O$_1$}. But that is \fmath{(\li!v!, \li!&(k,@v)&!)
    \in \interpV{\li!&@T!}}, as desired.

  Case on \coercible{\li!theta_1->theta_2! ~> \li!theta_1'->theta_2'!}: by
  assumption,
  \begin{itemize}
    \item \fmath{\li!v! = \li!?y=>e!} for some \lstinline|e|
    \item \fmath{\li!Ov! = \li!?y=>Oe!} for some \lstinline|Oe|
  \end{itemize}
  Suppose \step*{\li!COEROv! ~> \li!Ov'!}, i.e., \fmath{\li!Ov'! =
    \li!?y=>COER_2(Ov (COER_1y))!}. We want to show \fmath{(\li!v!, \li!Ov'!)
    \in \interpV{\li!theta_1'->theta_2'!}}. Towards this goal fix $i < n$ and
  \fmath{(\li!v_1!, \li!Ov_1'!) \in \interpV[i]{\li!theta_1'!}}. Now it suffices
  to show \fmath{(\li![v_1/y]e!, \li!COER_2(Ov (COER_1Ov_1'))!) \in
    \interpE[i]{\li!theta_2'!}}. Suppose $j < i$. We have the following trace:

  {\footnotesize\setlength{\abovedisplayskip}{0pt}
    \begin{align*}
      \li!COER_2(Ov (COER_1Ov_1'))!%
      &\steparrow[j_1] \li!COER_2(Ov Ov_1)! \\
      &\steparrow[1] \li!COER_2([Ov_1/y]Oe)! \\
      &\steparrow[j_2] \li!COER_2Ov_2! \\
      &\steparrow[j_3] \li!Ov_2'!
    \end{align*}
  }%
  where $j = j_1 + j_2 + j_3 + 1$, with \step[j_1]{\li!COER_1Ov_1'! ~>
    \li!Ov_1!} and \step[j_2]{\li![Ov_1/y]Oe! ~> \li!Ov_2!}. By induction
  hypothesis, we get \fmath{(\li!v_1!, \li!Ov_1!) \in
    \interpV[i]{\li!theta_1!}}. It follows, by assumption, that
  \fmath{(\li![v_1/y]e!, \li![Ov_1/y]Oe!) \in \interpE[i]{\li!theta_2!}}. Hence
  \step*{\li![v_1/y]e! ~> \li!v_2!} for some \lstinline|v_2| such that
  \fmath{(\li!v_2!, \li!Ov_2!) \in \interpV[i-j_2]{\li!theta_2!}}. It then
  follows by induction hypothesis that \fmath{(\li!v_2!, \li!Ov_2'!) \in
    \interpV[i-j_2]{\li!theta_2'!}}, which concludes the proof by
  \Cref{ap:thm:anti-mono}.
\end{proof}

\begin{theorem}[Correctness of declarative lifting of expressions]
  Suppose\/ \liftR{\sctx; \lctx; \gctx; \tctx |- \li!e! :: \li!theta! |>
    \li!Oe!} and\/ \lctxvalid[n]{\lctx}. Given a substitution\/ \fmath{\subst
    \in \interpG{\tctx}}, we have\/ \fmath{(\subst_1(\li!e!), \subst_2(\li!Oe!))
    \in \interpE{\li!theta!}}.
\end{theorem}
\begin{proof}
  By induction on the derivation of the declarative lifting judgment. Note that
  we do not fix step $n$ in the induction. That means we can instantiate it when
  apply induction hypothesis, but we also need to discharge \lctxvalid[n]{\lctx}
  for the $n$ we pick. We will not explicitly show this side-condition for
  brevity, because it is simply a consequence of \Cref{ap:thm:anti-mono} as long
  as we pick the same or a smaller step. The well-typedness side-conditions in
  logical relations are omitted, because they are trivial from
  \Cref{ap:thm:liftingR-reg} and preservation theorem. In addition, we only
  consider step $n > 0$ because the cases when $n = 0$ are always trivial.

  The cases \textsc{L-Unit}, \textsc{L-Lit}, \textsc{L-Var} and \textsc{L-Fun}
  are trivial.

  Case \textsc{T-Abs}: we need to show
  \fmath{(\li!?x:&[theta]&=>!\subst_1(\li!e!), \li!?x:theta=>!\subst_2(\li!Oe!))
    \in \interpE{\li!theta_1->theta_2!}}. Suppose $i < n$. Because lambda
  abstraction can not take step, $i = 0$, and it suffices to show
  \fmath{(\li!?x:&[theta]&=>!\subst_1(\li!e!), \li!?x:theta=>!\subst_2(\li!Oe!))
    \in \interpV{\li!theta_1->theta_2!}}. To this end fix $i < n$ and suppose
  that \fmath{(\li!v!, \li!Ov!) \in \interpV[i]{\li!theta_1!}} for some \li!v!
  and \li!Ov!. We know \fmath{\subst[\li!x! \mapsto (\li!v!, \li!Ov!)] \in
    \interpG[i]{\extctx{\li!x! :: \li!theta_1!}}}, by assumption and
  \Cref{ap:thm:anti-mono}. Now we specialize the induction hypothesis with
  $n = i$ and substitution \fmath{\subst[\li!x! \mapsto (\li!v!, \li!Ov!)]} to
  get \fmath{(\subst_1[\li!x! \mapsto \li!v!](\li!e!), \subst_2[\li!x! \mapsto
    \li!Ov!](\li!Oe!)) = (\li![v/x]!\subst_1(\li!e!),
    \li![Ov/x]!\subst_2(\li!Oe!)) \in \interpE[i]{\li!theta_2!}}, as required.

  Case \textsc{T-App}: we need to show %
  \fmath{(\li!$\subst_1($e_2$)$\ $\subst_1($e_1$)$!,%
    \li!$\subst_2($Oe_2$)$\ $\subst_2($Oe_1$)$!) \in \interpE{\li!theta_2!}}.
  Suppose that $i < n$, and %
  \step[i]{\li!$\subst_2($Oe_2$)$\ $\subst_2($Oe_1$)$! ~> \li!Ov!}. By the
  semantics definition, we have the following reduction trace:

  {\footnotesize\setlength{\abovedisplayskip}{0pt}
    \begin{align*}
      \li!$\subst_2($Oe_2$)$\ $\subst_2($Oe_1$)$!%
      &\steparrow[i_1] \li!$\subst_2($Oe_2$)$ Ov_1! \\
      &\steparrow[i_2] \li!Ov_2 Ov_1! \\
      &\steparrow[1] \li![Ov_1/x]Oe_2'! \\
      &\steparrow[i_3] \li!Ov!
    \end{align*}
  }%
  where \li!Ov_2! is \li!?x=>Oe_2'! for some \li!Oe_2'!, and
  $i = i_1 + i_2 + i_3 + 1$, with \step[i_1]{\subst_2(\li!Oe_1!) ~> \li!Ov_1!}
  and \step[i_2]{\subst_2(\li!Oe_2!) ~> \li!Ov_2!}. Instantiating the two
  induction hypotheses with step $n$ and \subst, we get:
  \begin{itemize}
    \item \fmath{(\subst_1(\li!e_2!), \subst_2(\li!Oe_2!)) \in
      \interpE{\li!theta_1->theta_2!}}
    \item \fmath{(\subst_1(\li!e_1!), \subst_2(\li!Oe_1!)) \in
      \interpE{\li!theta_1!}}
  \end{itemize}
  It follows that:
  \begin{itemize}
    \item \step*{\subst_1(\li!e_2!) ~> \li!v_2!} for some \lstinline|v_2|
    \item \fmath{(\li!v_2!, \li!Ov_2!) \in \interpV[n-i_2]{\li!theta_1->theta_2!}}
    \item \step*{\subst_1(\li!e_1!) ~> \li!v_1!} for some \lstinline|v_1|
    \item \fmath{(\li!v_1!, \li!Ov_1!) \in
      \interpV[n-i_1]{\li!theta_1!}}
  \end{itemize}
  where \li!v_2! is \li!?x=>e_2'! for some \li!e_2'!. It follows that
  \fmath{(\li![v_1/x]e_2'!, \li![Ov_1/x]Oe_2'!) \in
    \interpE[n-i_1-i_2-1]{\li!theta_2!}}, by specializing the value
  interpretation of function type to $n-i_1-i_2-1$. As $i_3 < n - i_1 -i_2 - 1$,
  we know \step*{\li![v_1/x]e_2'! ~> \li!v!} for some \lstinline|v|, such that
  \fmath{(\li!v!, \li!Ov!) \in \interpV[n-i_1-i_2-i_3-1]{\li!theta_2!}}.
  Therefore,

  {\footnotesize\setlength{\abovedisplayskip}{0pt}
    \begin{align*}
      \li!$\subst_1($e_2$)$\ $\subst_1($e_1$)$!%
      &\steparrow* \li!$\subst_1($e_2$)$ v_1! \\
      &\steparrow* \li!v_2 v_1! \\
      &\steparrow \li![v_1/x]e_2'! \\
      &\steparrow* \li!v!
    \end{align*}
  }%
  and \fmath{(\li!v!, \li!Ov!) \in \interpV[n-i]{\li!theta_2!}}, as desired.

  Case \textsc{L-Let}: observe the trace of the lifted expression:

  {\footnotesize\setlength{\abovedisplayskip}{0pt}
    \begin{align*}
      \li!let x =\ $\subst_2($Oe_1$)$ in\ $\subst_2($Oe_2$)$!%
      &\steparrow[i_1] \li!let x = Ov_1 in\ $\subst_2($Oe_2$)$! \\
      &\steparrow[1] \li![Ov_1/x]$\subst_2($Oe_2$)$! \\
      &\steparrow[i_2] \li!Ov!
    \end{align*}
  }%
  The rest of the proof is similar to \textsc{L-App} and \textsc{L-Abs}, with
  the induction hypothesis of the \li!let!-body specialized to step $n - i_1$.

  Case \textsc{L-Pair}: similarly with the trace:

  {\footnotesize\setlength{\abovedisplayskip}{0pt}
    \begin{align*}
      \li!($\subst_2($Oe_1$)$,$\subst_2($Oe_2$)$)!%
      &\steparrow[i_1] \li!(Ov_1,$\subst_2($Oe_2$)$)! \\
      &\steparrow[i_2] \li!(Ov_1,Ov_2)!
    \end{align*}
  }

  Case \textsc{L-Proj}: similarly with the trace:

  {\footnotesize\setlength{\abovedisplayskip}{0pt}
    \begin{align*}
      \li!pi_b\ $\subst_2($Oe$)$!%
      &\steparrow[i_1] \li!pi_b (Ov_1,Ov_2)! \\
      &\steparrow[1] \li!ite(b,Ov_1,Ov_2)!
    \end{align*}
  }

  Case \textsc{L-If$_1$}: similarly with the trace:

  {\footnotesize\setlength{\abovedisplayskip}{0pt}
    \begin{align*}
      \li!if\ $\subst_2($Oe_0$)$ then\ $\subst_2($Oe_1$)$ else\ $\subst_2($Oe_2$)$!%
      &\steparrow[i_0]%
        \li!if b then\ $\subst_2($Oe_1$)$ else\ $\subst_2($Oe_2$)$! \\
      &\steparrow[1] \li!ite(b,$\subst_2($Oe_1$)$,$\subst_2($Oe_2$)$)! \\
      &\steparrow[j] \li!Ov!
    \end{align*}
  }

  Case \textsc{L-Ctor$_1$}: similarly with the trace %
  \step[i]{\li!C_i\ $\subst_2($Oe$)$! ~> \li!C_i Ov!}. We also need to apply
  \Cref{ap:thm:public-type-interpV} to complete the proof.

  Case \textsc{L-Match$_1$}: similarly with the trace:

  {\footnotesize\setlength{\abovedisplayskip}{0pt}
    \begin{align*}
      \li!match\ $\subst_2($Oe_0$)$ with\ !\overline{\li!C x=>$\subst_2($Oe$)$!}%
      &\steparrow[j_0]%
        \li!match C_i Ov_0 with\ !\overline{\li!C x=>$\subst_2($Oe$)$!} \\
      &\steparrow[1] \li![Ov_0/x]$\subst_2($Oe_i$)$! \\
      &\steparrow[j_1] \li!Ov!
    \end{align*}
  }%
  Similar to \textsc{L-Ctor$_1$}, \Cref{ap:thm:public-type-interpV} is used.

  Case \textsc{L-If$_2$}: suppose $i < n$. We have the following trace:

  {\footnotesize\setlength{\abovedisplayskip}{0pt}
    \begin{align*}
      \li!@ite\ $\subst_2($Oe_0$)$\ $\subst_2($Oe_1$)$\ $\subst_2($Oe_2$)$!%
      &\steparrow[i_2]%
        \li!@ite\ $\subst_2($Oe_0$)$\ $\subst_2($Oe_1$)$ Ov_2! \\
      &\steparrow[i_1] \li!@ite\ $\subst_2($Oe_0$)$ Ov_1 Ov_2! \\
      &\steparrow[i_0] \li!@ite [b] Ov_1 Ov_2! \\
      &\steparrow[i_3] \li!Ov!
    \end{align*}
  }%
  where $i = i_0 + i_1 + i_2 + i_3$, with \step[i_0]{\subst_2(\li!Oe_0!) ~>
    \li![b]!}, \step[i_1]{\subst_2(\li!Oe_1!) ~> \li!Ov_1!} and
  \step[i_2]{\subst_2(\li!Oe_2!) ~> \li!Ov_2!}. By induction hypothesis, it
  follows that:
  \begin{itemize}
    \item \step*{\subst_1(\li!e_0!) ~> b}
    \item \step*{\subst_1(\li!e_1!) ~> \li!v_1!} for some \lstinline|v_1|
    \item \fmath{(\li!v_1!, \li!Ov_1!) \in \interpV[n-i_1]{\li!theta!}}
    \item \step*{\subst_1(\li!e_2!) ~> \li!v_2!} for some \lstinline|v_2|
    \item \fmath{(\li!v_2!, \li!Ov_2!) \in \interpV[n-i_2]{\li!theta!}}
  \end{itemize}
  Hence we have:

  {\footnotesize\setlength{\abovedisplayskip}{0pt}
    \begin{align*}
      \li!if\ $\subst_1($e_0$)$ then\ $\subst_1($e_1$)$ else\ $\subst_1($e_2$)$!%
      &\steparrow*%
        \li!if b then\ $\subst_1($e_1$)$ else\ $\subst_1($e_2$)$! \\
      &\steparrow \li!ite(b,$\subst_1($e_1$)$,$\subst_1($e_2$)$)! \\
      &\steparrow* \li!ite(b,v_1,v_2)!
    \end{align*}
  }%
  and \fmath{(\li!ite(b,v_1,v_2)!, \li!Ov!) \in \interpV[n-i]{\li!theta!}} by
  \Cref{ap:thm:mergeable-correct}, as required.

  Case \textsc{L-Ctor$_2$}: similarly with the trace:

  {\footnotesize\setlength{\abovedisplayskip}{0pt}
    \begin{align*}
      \li!@C_i\ $\subst_2($Oe$)$!%
      &\steparrow[j_1] \li!@C_i Ov_i! \\
      &\steparrow[j_2] \li!\&(k,@v)\&! = \li!Ov!
    \end{align*}
  }%
  We complete the proof by \textsc{A-I$_1$}.

  Case \textsc{L-Match$_2$}: similarly with the trace:

  {\footnotesize\setlength{\abovedisplayskip}{0pt}
    \begin{align*}
      \li!@Imatch\ $\subst_2($Oe_0$)$\ !\overline{\li!(?x=>$\subst_2($Oe$)$)!}%
      &\steparrow[j_0]%
        \li!@Imatch Ov_0\ !\overline{\li!(?x=>$\subst_2($Oe$)$)!} \\
      &\steparrow[j_1] \li!Ov!
    \end{align*}
  }%
  We complete the proof by \textsc{A-E$_1$}, whose assumptions are discharged by
  induction hypothesis.

  Case \textsc{L-Coerce}: the goal is \fmath{(\subst_1(\li!e!),
    \li!COER$\subst_2($Oe$)$!) \in \interpE{\li!theta'!}}. Suppose $i < n$. We
  have the trace:

  {\footnotesize\setlength{\abovedisplayskip}{0pt}
    \begin{align*}
      \li!COER$\subst_2($Oe$)$!%
      &\steparrow[i_1] \li!COEROv! \\
      &\steparrow[i_2] \li!Ov'!
    \end{align*}
  }%
  where $i = i_1 + i_2$, with \step[i_1]{\subst_2(\li!Oe!) ~> \li!Ov!}. It
  follows by induction hypothesis that \step*{\subst_1(\li!e!) ~> \li!v!} for
  some \lstinline|v| such that \fmath{(\li!v!, \li!Ov!) \in
    \interpV[n-i_1]{\li!theta!}}. Then, by \Cref{ap:thm:coercible-correct},
  \fmath{(\li!v!, \li!Ov'!) \in \interpV[n-i]{\li!theta'!}}, as desired.
\end{proof}

\begin{corollary}[Correctness of declarative lifting of closed terms]
  \label{ap:thm:correct-closed-exp}
  Suppose\/ \liftR{\sctx; \lctx; \gctx; \empctx |- \li!e! :: \li!theta! |>
    \li!Oe!} and\/ \lctxvalid[n]{\lctx}. We have \fmath{(\li!e!, \li!Oe!) \in
    \interpE{\li!theta!}}.
\end{corollary}

\begin{theorem}[Correctness of declarative lifting]
  \lctxderiv{\lctx} implies \lctxvalid{\lctx}.
\end{theorem}
\begin{proof}
  We need to show \lctxvalid[n]{\lctx} for any $n$, i.e., \fmath{(\li!x!,
    \li!Ox!) \in \interpE{\li!theta!}} for any \hastype{\li!x! :: \li!theta! |>
    \li!Ox!}. The proof is a straightforward induction on $n$.

  The base case is trivial, as \interpE[0]{\li!theta!} is a total relation as
  long as \li!x! and \li!Ox! are well-typed, which is immediate from
  \lctxderiv{\lctx}.

  Suppose \lctxvalid[n]{\lctx}. We need to show \fmath{(\li!x!, \li!Ox!) \in
    \interpE[n+1]{\li!theta!}}. From \lctxderiv{\lctx}, we know \li!x! and
  \li!Ox! are defined by \li!e! and \li!Oe! in \gctx, respectively. It is easy
  to see it suffices to prove \fmath{(\li!e!, \li!Oe!) \in
    \interpE{\li!theta!}}, because \li!x! and \li!Ox! take exactly one step to
  \li!e! and \li!Oe!. But that is immediate by \Cref{ap:thm:correct-closed-exp}.
\end{proof}

In the following theorem, we write \fmath{\subst(\tctx_2)} to obtain the second
projection of the typing context with the type variables substituted, i.e., each
\hastype{\li!x! :: \li!eta! ~ \li!X!} in \tctx{} is mapped to \hastype{\li!x! ::
  \subst(\li!X!)}. It may be counterintuitive that we do not
require \lstinline|theta| to be compatible with \lstinline|eta| or each pair in
the typing context to be compatible, but these side conditions are indeed not
needed; if the generated constraints (under empty context) are satisfiable, the
compatibility conditions should hold by construction.

\begin{theorem}[Soundness of algorithmic lifting of open terms]
  Suppose\/ \liftA{\gctx; \tctx |- \li!e! :: \li!eta! ~ \li!X! |> \li!Oe! ~>
    \cstrs}. Given a specification type \li!theta!, if\/ \cstrsat{\sctx; \lctx;
    \gctx; \subst |- \li!X! = \li!theta!, \cstrs}, then\/
  \fmath{\subst(\li!Oe!)} elaborates to an expression \li!Oe'!, such that\/
  \liftR{\sctx; \lctx; \gctx; \subst(\tctx_2) |- \li!e! :: \li!theta! |>
    \li!Oe'!}.
\end{theorem}
\begin{proof}
  By induction on the derivation of the algorithmic lifting judgment. The cases
  of \textsc{A-Unit}, \textsc{A-Lit} and \textsc{A-Fun} are trivial. The cases
  of \textsc{A-Let}, \textsc{A-Pair} and \textsc{A-Proj} are similar to
  \textsc{A-Abs} and \textsc{A-App}, so we only show the proofs of these two
  cases. The cases of \textsc{A-Ctor} and \textsc{A-Match} are similar to
  \textsc{A-If} (and other cases), so we only show the proof of \textsc{A-If}.

  Case \textsc{A-Var}: suppose \cstrsat{\li!X'! = \li!theta!, \li!MCOER(X,X')!}.
  It follows that:
  \begin{itemize}
    \item \fmath{\li!theta! = \subst(\li!X'!)}
    \item \coercible{\subst(\li!X!) ~> \subst(\li!X'!) |> \li!COER!}
    \item \ppx{\li!MCOER($\subst($X$)$,$\subst($X'$)$;x)! |>
      \li!COERx!}
  \end{itemize}
  Therefore, \fmath{\subst(\li!MCOER(X,X';x)!) =
    \li!MCOER($\subst($X$)$,$\subst($X'$)$;x)!} elaborates to \lstinline!COERx!.
  Since \fmath{\hastype{\li!x! :: \li!eta! ~ \li!X!} \in \tctx} by assumption,
  \fmath{\hastype{\li!x! :: \subst(\li!X!)} \in \subst(\tctx_2)}. We can then
  apply \textsc{L-Coerce} and \textsc{L-Var} to derive \liftR{\subst(\tctx_2) |-
    \li!x! :: \subst(\li!X'!) |> \li!COERx!}.

  Case \textsc{A-Abs}: suppose \cstrsat{\li!X! = \li!theta!, \incls{\li!X_1! ~
      \li!eta_1!}, \incls{\li!X_2! ~ \li!eta_2!}, \li!X! = \li!X_1->X_2!,
    \cstrs}. It follows that:
  \begin{itemize}
    \item \fmath{\li!theta! = \subst(\li!X!) =
      \li!$\subst($X_1$)$->$\subst($X_2$)$!}
    \item \fmath{\li!&[$\subst($X_1$)$]&! = \li!eta_1!}
    \item \fmath{\li!&[$\subst($X_2$)$]&! = \li!eta_2!}
  \end{itemize}
  Because \cstrsat{\li!X_2! = \subst(\li!X_2!), \cstrs}, we have by induction
  hypothesis:
  \begin{itemize}
    \item \fmath{\subst(\li!Oe!)} elaborates to some \lstinline|Oe'|
    \item \liftR{\extctx[\subst(\tctx_2)]{\li!x! :: \subst(\li!X_1!)} |- \li!e!
      :: \subst(\li!X_2!) |> \li!Oe'!}
  \end{itemize}
  Therefore, \fmath{\subst(\li!?x:X_1=>Oe!) =
    \li!?x:$\subst($X_1$)$=>$\subst($Oe$)$!} elaborates
  to \lstinline|?x:$\subst($X_1$)$=>Oe'|, and \liftR{\subst(\tctx_2) |-
    \li!?x:eta_1=>e! :: \li!$\subst($X_1$)$->$\subst($X_2$)$! |>
    \li!?x:$\subst($X_1$)$=>Oe'!}, by \textsc{L-Abs}.

  Case \textsc{A-App}: suppose \cstrsat{\li!X_2! = \li!theta!, \incls{\li!X_1! ~
      \li!eta_1!}, \li!X! = \li!X_1->X_2!, \cstrs}. It follows that:
  \begin{itemize}
    \item \fmath{\li!theta! = \subst(\li!X_2!)}
    \item \fmath{\li!&[$\subst($X_1$)$]&! = \li!eta_1!}
    \item \fmath{\subst(\li!X!) = \li!$\subst($X_1$)$->$\subst($X_2$)$!}
  \end{itemize}
  Because \cstrsat{\li!X_1! = \subst(\li!X_1!), \cstrs}, we have by induction
  hypothesis:
  \begin{itemize}
    \item \fmath{\subst(\li!Oe_1!)} elaborates to some \lstinline|Oe_1'|
    \item \liftR{\subst(\tctx_2) |- \li!x_1! :: \subst(\li!X_1!) |> \li!Oe_1'!}
  \end{itemize}
  Therefore, \fmath{\subst(\li!x_2 Oe_1!) = \li!x_2\ $\subst($Oe_1$)$!}
  elaborates to \lstinline|x_2 Oe_1'|, and \liftR{\subst(\tctx_2) |-%
    \li!x_2 x_1! :: \subst(\li!X_2!) |> \li!x_2 Oe_1'!}, by \textsc{L-App} and
  \textsc{L-Var}.

  Case \textsc{A-If}: suppose \cstrsat{\li!X! = \li!theta!, \li!Mite(X_0,X)!,
    \cstrs_1, \cstrs_2}. We know \fmath{\subst(\li!X!) = \li!theta!}. Because
  \cstrsat{\li!X! = \subst(\li!X!), \cstrs_1} and \cstrsat{\li!X! =
    \subst(\li!X!), \cstrs_2}, we have by induction hypothesis:
  \begin{itemize}
    \item \fmath{\subst(\li!Oe_1!)} elaborates to some \lstinline|Oe_1'|
    \item \liftR{\subst(\tctx_2) |- \li!e_1! :: \subst(\li!X!) |> \li!Oe_1'!}
    \item \fmath{\subst(\li!Oe_2!)} elaborates to some \lstinline|Oe_2'|
    \item \liftR{\subst(\tctx_2) |- \li!e_2! :: \subst(\li!X!) |> \li!Oe_2'!}
  \end{itemize}
  Since \cstrsat{\li!Mite(X_0,X)!},
  \ppx{\li!Mite($\subst($X_0$)$,$\subst($X$)$;x_0,Oe_1',Oe_2')! |> \li!Oe!} for
  some \lstinline|Oe|. By inverting this judgment, we consider two cases. The
  first case has condition type \lstinline|bool|:
  \begin{itemize}
    \item \ppx{\li!Mite($\subst($X_0$)$,$\subst($X$)$;x_0,Oe_1',Oe_2')! |>
      \li!if x_0 then Oe_1' else Oe_2'!}
    \item \fmath{\subst(\li!X_0!) = \li!bool!}
  \end{itemize}
  In this case,
  \fmath{\li!Mite($\subst($X_0$)$,$\subst($X$)$;x_0,$\subst($Oe_1$)$,$\subst($Oe_2$)$)!}
  elaborates to \li!if x_0 then Oe_1' else Oe_2'!, and, by \textsc{L-If$_1$} and
  \textsc{L-Var}, \liftR{\subst(\tctx_2) |- \li!if x_0 then e_1 else e_2! ::
    \subst(\li!X!) |> \li!if x_0 then Oe_1' else Oe_2'!}. The second case has
  condition type \lstinline|@bool|:
  \begin{itemize}
    \item \mergeable{\subst(\li!X!) |> \li!@ite!}
    \item \ppx{\li!Mite($\subst($X_0$)$,$\subst($X$)$;x_0,Oe_1',Oe_2')! |>
      \li!@ite x_0 Oe_1' Oe_2'!}
    \item \fmath{\subst(\li!X_0!) = \li!@bool!}
  \end{itemize}
  In this case,
  \fmath{\li!Mite($\subst($X_0$)$,$\subst($X$)$;x_0,$\subst($Oe_1$)$,$\subst($Oe_2$)$)!}
  elaborates to \li!@ite x_0 Oe_1' Oe_2'!, and, by \textsc{L-If$_2$} and
  \textsc{L-Var}, \liftR{\subst(\tctx_2) |- \li!if x_0 then e_1 else e_2! ::
    \subst(\li!X!) |> \li!@ite x_0 Oe_1' Oe_2'!}.
\end{proof}

\begin{corollary}[Soundness of algorithmic lifting]
  Suppose\/ \liftA{\gctx; \empctx |- \li!e! :: \li!eta! ~ \li!X! |> \li!Oe! ~>
    \cstrs}. Given a specification type \li!theta!, if\/ \cstrsat{\sctx; \lctx;
    \gctx; \subst |- \li!X! = \li!theta!, \cstrs}, then\/
  \fmath{\subst(\li!Oe!)} elaborates to an expression \li!Oe'!, such that\/
  \liftR{\sctx; \lctx; \gctx; \empctx |- \li!e! :: \li!theta! |> \li!Oe'!}.
\end{corollary}

\FloatBarrier

\section{Constraint Solving}

\begin{figure}[t]
  \footnotesize
  \begin{algorithmic}
    \Function{solve}{\fctx, \mctx, \lstinline!f!, \lstinline!theta!}
    \If{\fmath{(\li!f!, \li!theta!) \mapsto \subst \in \mctx}}
      \Return \mctx
    \EndIf

    \State \fmath{(\li!X!, \cstrs, \cstrs') \gets \fctx(\li!f!)}
    \State \fmath{\phi \gets} \Call{lower}{\fmath{\li!X! = \li!theta!, \cstrs}}
    \If{\Call{off-the-shelf}{\fmath{\phi}} returns unsat}
      \textbf{fail}
    \ElsIf{\Call{off-the-shelf}{\fmath{\phi}} returns sat with \subst}
      \If{there is a \fmath{\li!Mg(theta')! \in \cstrs'} s.t.
          \Call{solve}{\fctx,
            \fmath{\mctx[(\li!f!, \li!theta!) \mapsto \subst]},
            \lstinline!g!,
            \fmath{\subst(\li!theta'!)}} fails}
        \State \Call{solve}{\fmath{\fctx[\li!f! \mapsto
          (\li!X!, \cstrs \cup \set{\li!theta'! \neq \subst(\li!theta'!)},
           \cstrs')]},
          \mctx, \lstinline!f!, \lstinline!theta!}
      \Else
        \State \fmath{\bigcup} \Call{solve}{\fctx,
            \fmath{\mctx[(\li!f!, \li!theta!) \mapsto \subst]},
            \lstinline!g!,
            \fmath{\subst(\li!theta'!)}} for all
              \fmath{\li!Mg(theta')! \in \cstrs'}
      \EndIf
    \EndIf
    \EndFunction
  \end{algorithmic}

  \caption{Constraint solving algorithm}
  \label{ap:fig:solver}
\end{figure}

The constraint solving algorithm uses two maps. \fctx{} maps function names to
the generated constraints with a type variable \lstinline!X! as the placeholder
for the potential specification type, resulting from the lifting algorithm:
every entry of \fctx{} has the form \fmath{\li!f! \mapsto (\li!X!, \cstrs,
\cstrs')}. The generated constraints are partitioned into \cstrs,
which consists of all
constraints except for function call constraints, and \fmath{\cstrs'},
which consists of the functional call constraints (\lstinline!Mx!). Initially, \fctx{} consists of all
functions collected from the keyword \lstinline!Mlift! and the functions they depend
on. \mctx{} maps a pair of function name \lstinline!f! and its target type
\lstinline!theta!, called a \emph{goal}, to a type assignment \subst{}: each
entry has the form \fmath{(\li!f!, \li!theta!) \mapsto \subst}. Initially
\mctx{} is empty.

The constraint solving algorithm takes an initial \fctx{} and \mctx{}, and a
goal, and returns an updated \mctx{} that consists of the type assignments for
the goal and all the subgoals this goal depends on. The constraint solver is
applied to all goals from \lstinline!Mlift!, and the final \mctx{} is the union of
all returned \mctx{}. With \mctx{}, we can generate functions from the type
assignments for each goal. The global context is subsequently extended by these
generated functions, and the final lifting context is constructed by pairing each goal
with its corresponding generated function.

\Cref{ap:fig:solver} presents a naive algorithm for constraint
solving. The subroutine \textsc{Lower} is used to reduce all
constraints except for function call constraints to formulas in the
\emph{quantifier-free finite domain theory} (\verb!QF_FD!), then an
off-the-shelf solver (Z3) is then used to solve them. To do so, we
first decompose all compatibility class constraints, i.e.,
\incls{\li!X! ~ \li!eta!}, into \emph{atomic} classes. For example,
\incls{\li!X! ~ \li!list -> int!} is decomposed into \incls{\li!X_1! ~
  \li!list!} and \incls{\li!X_2! ~ \li!int!}, with \lstinline!X!
substituted by \lstinline!X_1->X_2! in all other constraints. Then the
newly generated compatibility class constraints can be reduced to a
disjunction of all possible specification types in this class. For
example, \incls{\li!X_1! ~ \li!list!}  reduces to \fmath{\li!X_1! =
  \li!list! \lor \li!X_1! = \li!&@list<=! \lor \li!X_1! =
  \li!&@list==!}, if \lstinline!@list<=! and \lstinline!@list==! are
the only OADTs for \lstinline!list!. After this step, all type
variables are compatible with an atomic simple type, so we can also
decompose other constraints into a set of base cases according to the
rules of their relations. Note that while mergeability is not one of
the constraints generated by the lifting algorithm, such constraints
still arise from \lstinline!Mite! and \lstinline!Mmatch!. When
multiple assignments are valid for a particular type variable, we
prefer the ``cheaper'', i.e., more permissive, solution. For example,
we prefer public type \lstinline!list! over OADTs if possible, and
prefer \lstinline!@list==! over \lstinline!@list<=!. We encode these
preferences as \emph{soft constraints}, and assign a bigger penalty to
more restrictive types. The penalty is inferred by analyzing
the coercion relations: a more permissive type can be coerced to a
more restrictive type but not the other way around, because the
restrictive type hides more information.

Once other constraints are solved, the type parameters of function call
constraints are concretized. The algorithm is then recursively applied to this
subgoal, i.e., the functional call paired with the concretized type, assuming the
original goal is solved by extending \mctx{} with the original goal and the type
assignments. This handles potential (mutual) recursion. If a subgoal fails,
its target type will be added as a refutation to the corresponding constraints,
and backtrack.

It is easy to see that this algorithm terminates: every recursive call to
\textsc{solve} either adds a refutation that reduces the search space of the
non-function-call constraints, or extends \mctx{} which is finitely bounded by
the number of functions and the number of OADTs, hence reducing the search space
of function-call constraints.

Our actual implementation maintains a more sophisticated state of \mctx{} so
that we do not re-solve the same goal if it has already been solved or rejected.
We also exploit the incremental solver of Z3 to help with performance.

\section{Case Study}

In addition to the applications in \taype's original evaluation, our
evaluation additionally includes the following case study.

\paragraph{Dating application} Consider a functionality of matching potential
soulmates. Each party owns their private profile with personal information, such
as gender, income and education. They also have a private preference for their
partner, encoded as predicates over profiles of \emph{both} parties. These
predicates are expressions with boolean connectives, integer arithmetics and
numeric comparisons. For example, one user may stipulate that the sum of both parties'
income exceeds a particular amount. The peer matching function takes these
private profiles and predicates, and returns a boolean indicating whether they
are a good match, by evaluating the predicate expression on the profiles. The
private predicates have many potential policies. As predicates are essentially
ASTs, participating parties may agree on disclosing only the depth of the
predicates, or revealing the AST nodes but not the operands, or even revealing
only the boolean connectives but keep the integer expressions secret. In
\taypsi, the peer matching function, its auxiliary functions, and the data types
they depend on can be implemented in the conventional way, without knowing the
policies. The private matching function can be obtained by composing with the
desired policy. Updating policy or updating the matching algorithm can be done
independently.

\section{Optimization}

\paragraph{Smart array} An oblivious data value is translated to an
array of oblivious integers (i.e., an oblivious array) as its private
representation. Thus, constructing and extracting oblivious data
values correspond to performing array operations. Our smart array
implementation, parameterized by the underlying cryptographic backend,
aims to reduce the cost of these operations. First, array slicing and
concatenation are zero-cost (even when they happen alternately). Our
implementation does not create new arrays when performing these
operations, and allow the results to follow the original structure in
\taypsi for as long as possible, until a \lstinline!mux! forces them
to be flattened. Conceptually, the smart arrays delay these operations,
performing them all at once when flattening is required. Second, smart
arrays keep track of values that are actually public or arbitrary,
even though they are seemingly ``encrypted''. This happens when we
apply a section function to a known value, and there is no reason to
actually ``forget'' what we already know. For example, instead of
evaluating %
\lstinline!@bool#s true! to the encrypted value \lstinline![true]!,
the smart array can simply marks this boolean \lstinline!true! as
``encrypted'' without actually doing so. This allows some expensive
cryptographic operations to be skipped, especially \lstinline!mux!, by
exploiting this known information. For example, when we know the left
and right branches of a \lstinline!mux! are \lstinline!true! and
\lstinline!false! respectively, then the result is simply the private
condition, without needing to perform the \lstinline!mux!. This
strategy works because \lstinline!mux [b] [true] [false]! is
equivalent to \lstinline![b]!. This kind of scenario appears
frequently in pattern matching, where the constructor ``tag'' of the
ADT returned in each branch may be publicly known. This optimization
can make the commonly used function \lstinline!map! free of
cryptographic operations (except for the higher-order function applied
to data structure's payload of course). Recall that the list
\lstinline!map! function returns \lstinline!Nil! if the input list is
a \lstinline!Nil!, and returns a \lstinline!Cons! if the input list is
a \lstinline!Cons!. Thus the constructor tag of the result is
\lstinline!mux tag [Nil] [Cons]!, where \lstinline!tag! is the
constructor tag of the input list. Similar to the previous example,
this expression can be reduced to \lstinline!tag! without performing
\lstinline!mux!.  Intuitively, \lstinline!map! function does not
modify the structure of the input data, so the constructor tags are
preserved. On the other hand, keeping track of arbitrary values that
are used for padding is also helpful. If we know a branch is arbitrary
in a \lstinline!mux!, we can simply return the other branch.

\paragraph{Reshape guard} Reshaping an OADT to the same public view
should do nothing. Reshaping with the same public view is a quite
common scenario, especially when the partial order defined on the
public view is a total order.  For example, the join of the public
views of \lstinline!@list<=! is simply the maximum one, which is equal
to one of these public views. Therefore, when we reshape the
\lstinline!@list<=!s to this common public view, one of them does not
require any additional work. The reshape guard optimization is very
simple: we instrument every reshape function with a check for public
view equality, and skip reshape operations when the views are the
same.

\paragraph{Private representation size memoization} An oblivious type
definition is translated to a function from a public view to a size
(of the corresponding array), in order to extract private data from
oblivious array or create new one.  Invoking these size functions in a
recursive function can potentially introduce asymptotic slowdown, as
the sizes of the private representation are calculated
repeatedly. Previous work (\taype) solves this problem by applying a
tupling optimization to the section and retraction functions, which
are the only functions that require calculating these sizes. However,
in \taypsi this size calculation can happen to any lifted private
functions, and the tupling optimization is quite brittle to apply to
arbitrary functions. To avoid this slowdown, we memoize the map from
public views to the representation sizes. If the public view is
integer, then the memoization is standard: we create a hash table for
each size function, such as \lstinline!@list<=!, and look up the table
if the size was calculated before, or add the new result to the
table. If the public view is an ADT, we automatically embed the
representation size within the public view type itself, and update the
introduction and elimination forms and the \taypsi programs using
them, accordingly. For example, a Peano number public view
\lstinline!nat! is augmented as follow.
\begin{lstlisting}
data nat = nat_memo * size
data nat_memo = Zero | Succ nat
\end{lstlisting}
Whenever we need to calculate the size, we simply project it from
\lstinline!nat!, without actually computing the size function.

\FloatBarrier

\section{Evaluation}

All runtime experiments report running time in milliseconds. A
\textcolor{red}{\bf failed} entry indicates the benchmark either timed out after
$5$ minutes or exceeded the memory bound of $8$\,GB.

\Cref{ap:fig:list-bench} and \Cref{ap:fig:tree-bench} presents the results of
running the list and tree micro-benchmarks, respectively. The \taypsi columns
reports the running time relative to \taype.

\Cref{ap:fig:list-opt} and \Cref{ap:fig:tree-opt} presents the results of the
ablation tests of list and tree micro-benchmarks, respectively. The table also
shows the slowdown relative the that of the fully optimized version reported in
\Cref{ap:fig:list-bench} and \Cref{ap:fig:tree-bench}.

\Cref{ap:fig:list-memo} and \Cref{ap:fig:tree-memo} presents the results of the
ADT memoization tests of list and tree micro-benchmarks, respectively. The ADT
public view used in list examples is the maximum length as a Peano number, while
the ADT public view used in tree examples is the maximum spine. The ``No
memoization'' column also reports the slowdown relative to ``Base'' column with
memoization enabled.

\Cref{ap:fig:compile-stats} presents the statistics in compilation. It reports
total compilation time (Total), time spent on constraint solving (Solver), and
the number of solver queries (\#Queries). The table also presents the number of
functions (\#Functions) being translated, the number of atomic types (\#Types)
and the total number of atomic components in function types (\#Atoms). The
reported running time is in seconds.

\begin{figure}[h]
  \centering
  \input{list-bench-full.tex}
  \caption{Micro-benchmarks for list examples}
  \label{ap:fig:list-bench}
\end{figure}

\begin{figure}[h]
  \centering
  \input{tree-bench-full.tex}
  \caption{Micro-benchmarks for tree examples}
  \label{ap:fig:tree-bench}
\end{figure}

\begin{figure}[h]
  \centering
  \input{list-opt-full.tex}
  \caption{Ablation tests for list examples}
  \label{ap:fig:list-opt}
\end{figure}

\begin{figure}[h]
  \centering
  \input{tree-opt-full.tex}
  \caption{Ablation tests for tree examples}
  \label{ap:fig:tree-opt}
\end{figure}

\begin{figure}[h]
  \centering
  \input{list-memo-full.tex}
  \caption{ADT memoization tests for list examples}
  \label{ap:fig:list-memo}
\end{figure}

\begin{figure}[h]
  \centering
  \input{tree-memo-full.tex}
  \caption{ADT memoization tests for tree examples}
  \label{ap:fig:tree-memo}
\end{figure}

\begin{figure}[h]
  \centering
  \input{compile-stats-full.tex}
  \caption{Compilation overhead}
  \label{ap:fig:compile-stats}
\end{figure}

%% file: list-bench-full.tex
\begin{tabular}{llll}
Benchmark & \taype (ms) & \taype-SA (ms) & \taypsi (ms) \\
\midrule
\verb|elem_1000| & 8.15 & 8.11 & 8.02 \hfill(98.47\%, 98.89\%) \\
\verb|hamming_1000| & 15.09 & 15.21 & 14.46 \hfill(95.79\%, 95.04\%) \\
\verb|euclidean_1000| & 67.43 & 67.55 & 67.32 \hfill(99.84\%, 99.66\%) \\
\verb|dot_prod_1000| & 66.12 & 66.19 & 66.41 \hfill(100.43\%, 100.33\%) \\
\verb|nth_1000| & 11.98 & 12.05 & 12.04 \hfill(100.54\%, 99.93\%) \\
\verb|map_1000| & 2139.55 & 5.07 & 5.14 \hfill(0.24\%, 101.44\%) \\
\verb|filter_200| & \textcolor{red}{\bf failed} & \textcolor{red}{\bf failed} & 86.86 \hfill(N/A, N/A) \\
\verb|insert_200| & 5796.69 & 88.92 & 88.07 \hfill(1.52\%, 99.04\%) \\
\verb|insert_list_100| & \textcolor{red}{\bf failed} & \textcolor{red}{\bf failed} & 4667.66 \hfill(N/A, N/A) \\
\verb|append_100| & 4274.7 & 45.09 & 44.18 \hfill(1.03\%, 97.99\%) \\
\verb|take_200| & 169.07 & 3.05 & 3.09 \hfill(1.83\%, 101.15\%) \\
\verb|flat_map_200| & \textcolor{red}{\bf failed} & \textcolor{red}{\bf failed} & 7.3 \hfill(N/A, N/A) \\
\verb|span_200| & 13529.34 & 124.79 & 91.22 \hfill(0.67\%, 73.09\%) \\
\verb|partition_200| & \textcolor{red}{\bf failed} & \textcolor{red}{\bf failed} & 176.49 \hfill(N/A, N/A) \\
\end{tabular}

%% file: tree-bench-full.tex
\begin{tabular}{llll}
Benchmark & \taype (ms) & \taype-SA (ms) & \taypsi (ms) \\
\midrule
\verb|elem_16| & 446.81 & 459.1 & 404.9 \hfill(90.62\%, 88.19\%) \\
\verb|prob_16| & 13082.52 & 12761.7 & 12735.16 \hfill(97.34\%, 99.79\%) \\
\verb|map_16| & 4414.69 & 262.14 & 215.67 \hfill(4.89\%, 82.27\%) \\
\verb|filter_16| & 8644.14 & 452.04 & 433.7 \hfill(5.02\%, 95.94\%) \\
\verb|swap_16| & \textcolor{red}{\bf failed} & \textcolor{red}{\bf failed} & 4251.36 \hfill(N/A, N/A) \\
\verb|path_16| & \textcolor{red}{\bf failed} & 6657.07 & 894.88 \hfill(N/A, 13.44\%) \\
\verb|insert_16| & 83135.81 & 8093.81 & 1438.87 \hfill(1.73\%, 17.78\%) \\
\verb|bind_8| & 21885.65 & 494.98 & 532.86 \hfill(2.43\%, 107.65\%) \\
\verb|collect_8| & \textcolor{red}{\bf failed} & \textcolor{red}{\bf failed} & 143.38 \hfill(N/A, N/A) \\
\end{tabular}

%% file: list-opt-full.tex
\begin{tabular}{llll}
Benchmark & No smart array (ms) & No reshape guard (ms) & No memoization (ms) \\
\midrule
\verb|elem_1000| & 18.37 \hfill(2.29x) & 8.06 \hfill(1.0x) & 17.76 \hfill(2.21x) \\
\verb|hamming_1000| & 51.73 \hfill(3.58x) & 14.53 \hfill(1.01x) & 35.5 \hfill(2.46x) \\
\verb|euclidean_1000| & 79.07 \hfill(1.17x) & 67.31 \hfill(1.0x) & 76.36 \hfill(1.13x) \\
\verb|dot_prod_1000| & 87.77 \hfill(1.32x) & 66.15 \hfill(1.0x) & 77.33 \hfill(1.16x) \\
\verb|nth_1000| & 22.69 \hfill(1.88x) & 12.18 \hfill(1.01x) & 20.53 \hfill(1.7x) \\
\verb|map_1000| & 2106.43 \hfill(409.89x) & 139.91 \hfill(27.23x) & 37.71 \hfill(7.34x) \\
\verb|filter_200| & 5757.28 \hfill(66.29x) & 93.93 \hfill(1.08x) & 114.7 \hfill(1.32x) \\
\verb|insert_200| & 255.43 \hfill(2.9x) & 94.61 \hfill(1.07x) & 89.32 \hfill(1.01x) \\
\verb|insert_list_100| & 22806.87 \hfill(4.89x) & 5186.07 \hfill(1.11x) & 4771.28 \hfill(1.02x) \\
\verb|append_100| & 4226.32 \hfill(95.66x) & 50.79 \hfill(1.15x) & 61.77 \hfill(1.4x) \\
\verb|take_200| & 169.45 \hfill(54.91x) & 12.92 \hfill(4.19x) & 4.68 \hfill(1.52x) \\
\verb|flat_map_200| & 5762.63 \hfill(789.08x) & 16.99 \hfill(2.33x) & 60.03 \hfill(8.22x) \\
\verb|span_200| & 5924.1 \hfill(64.95x) & 99.83 \hfill(1.09x) & 120.09 \hfill(1.32x) \\
\verb|partition_200| & 11528.0 \hfill(65.32x) & 185.16 \hfill(1.05x) & 231.06 \hfill(1.31x) \\
\end{tabular}

%% file: tree-opt-full.tex
\begin{tabular}{llll}
Benchmark & No smart array (ms) & No reshape guard (ms) & No memoization (ms) \\
\midrule
\verb|elem_16| & 433.73 \hfill(1.07x) & 404.05 \hfill(1.0x) & 402.15 \hfill(0.99x) \\
\verb|prob_16| & 13019.56 \hfill(1.02x) & 12746.24 \hfill(1.0x) & 12731.89 \hfill(1.0x) \\
\verb|map_16| & 4410.84 \hfill(20.45x) & 635.18 \hfill(2.95x) & 213.96 \hfill(0.99x) \\
\verb|filter_16| & 8674.71 \hfill(20.0x) & 1131.02 \hfill(2.61x) & 440.16 \hfill(1.01x) \\
\verb|swap_16| & 8671.52 \hfill(2.04x) & 5471.4 \hfill(1.29x) & 4246.39 \hfill(1.0x) \\
\verb|path_16| & 9108.54 \hfill(10.18x) & 1083.21 \hfill(1.21x) & 888.95 \hfill(0.99x) \\
\verb|insert_16| & 19101.36 \hfill(13.28x) & 2151.83 \hfill(1.5x) & 1432.92 \hfill(1.0x) \\
\verb|bind_8| & 19647.83 \hfill(36.87x) & 870.93 \hfill(1.63x) & 534.3 \hfill(1.0x) \\
\verb|collect_8| & 11830.6 \hfill(82.51x) & 152.29 \hfill(1.06x) & 186.92 \hfill(1.3x) \\
\end{tabular}

%% file: list-memo-full.tex
\begin{tabular}{lll}
Benchmark & Base (ms) & No memoization (ms) \\
\midrule
\verb|elem_1000| & 13.45 & 18.2 \hfill(1.35x) \\
\verb|hamming_1000| & 25.98 & 36.35 \hfill(1.4x) \\
\verb|euclidean_1000| & 73.22 & 77.07 \hfill(1.05x) \\
\verb|dot_prod_1000| & 77.65 & 77.92 \hfill(1.0x) \\
\verb|nth_1000| & 17.8 & 20.89 \hfill(1.17x) \\
\verb|map_1000| & 19.95 & 42.53 \hfill(2.13x) \\
\verb|filter_200| & 87.22 & 118.54 \hfill(1.36x) \\
\verb|insert_200| & 89.01 & 89.77 \hfill(1.01x) \\
\verb|insert_list_100| & 4719.13 & 4809.16 \hfill(1.02x) \\
\verb|append_100| & 45.05 & 63.83 \hfill(1.42x) \\
\verb|take_200| & 4.08 & 5.17 \hfill(1.27x) \\
\verb|flat_map_200| & 7.32 & 69.1 \hfill(9.45x) \\
\verb|span_200| & 92.73 & 124.06 \hfill(1.34x) \\
\verb|partition_200| & 176.68 & 238.27 \hfill(1.35x) \\
\end{tabular}

%% file: tree-memo-full.tex
\begin{tabular}{lll}
Benchmark & Base (ms) & No memoization (ms) \\
\midrule
\verb|elem_16| & 425.41 & 416.13 \hfill(0.98x) \\
\verb|prob_16| & 12772.3 & 12762.84 \hfill(1.0x) \\
\verb|map_16| & 268.07 & 255.66 \hfill(0.95x) \\
\verb|filter_16| & 494.38 & 485.96 \hfill(0.98x) \\
\verb|swap_16| & 4407.25 & 4329.39 \hfill(0.98x) \\
\verb|path_16| & 940.83 & 944.69 \hfill(1.0x) \\
\verb|insert_16| & 1603.88 & 1770.76 \hfill(1.1x) \\
\verb|bind_8| & 553.0 & 585.06 \hfill(1.06x) \\
\verb|collect_8| & 143.79 & 187.25 \hfill(1.3x) \\
\end{tabular}

%% file: compile-stats-full.tex
\begin{tabular}{l|rrr|rrr}
Suite & \#Functions & \#Types & \#Atoms & \#Queries & Total (s) & Solver (s) \\
\midrule
List & 20 & 7 & 70 & 84 & 0.47 & 0.081 \\
Tree & 14 & 9 & 44 & 31 & 0.47 & 0.024 \\
List (stress) & 20 & 12 & 70 & 295 & 3.45 & 2.8 \\
Dating & 4 & 13 & 16 & 10 & 0.58 & 0.019 \\
Medical Records & 20 & 19 & 58 & 51 & 0.48 & 0.072 \\
Secure Calculator & 2 & 9 & 6 & 5 & 1.34 & 0.013 \\
Decision Tree & 2 & 13 & 6 & 16 & 0.28 & 0.016 \\
K-means & 16 & 11 & 68 & 86 & 1.62 & 0.95 \\
Miscellaneous & 11 & 7 & 42 & 47 & 0.26 & 0.065 \\
\end{tabular}